\documentclass[12pt]{article}
\usepackage[utf8]{inputenc}
\usepackage{amsmath,amsthm,amssymb,fullpage,comment,microtype,graphicx}
\usepackage[colorlinks=true,allcolors=blue]{hyperref}

\newcommand{\ex}[2]{{\ifx&#1& \mathbb{E} \else \underset{#1}{\mathbb{E}} \fi \left[#2\right]}}
\newcommand{\exc}[3]{{\ifx&#1& \mathbb{E} \else \underset{#1}{\mathbb{E}} \fi \left[{#2}\middle|{#3}\right]}}
\newcommand{\pr}[2]{{\ifx&#1& \mathbb{P} \else \underset{#1}{\mathbb{P}} \fi \left[#2\right]}}
\newcommand{\prc}[3]{{\ifx&#1& \mathbb{P} \else \underset{#1}{\mathbb{P}} \fi \left[{#2}\middle|{#3}\right]}}
\newcommand{\var}[2]{{\ifx&#1& \mathbb{V} \else \underset{#1}{\mathbb{V}} \fi \left[#2\right]}}
\newcommand{\dr}[3]{\mathrm{D}_{#1}\left(#2\middle\|#3\right)}

\newcommand{\nope}[1]{}

\newtheorem{theorem}{Theorem}
\newtheorem{lemma}[theorem]{Lemma}
\newtheorem{definition}[theorem]{Definition}
\newtheorem{corollary}[theorem]{Corollary}
\newtheorem{proposition}[theorem]{Proposition}
\newtheorem{remark}[theorem]{Remark}

\usepackage[style=alphabetic,backend=bibtex,maxalphanames=10,maxbibnames=20,maxcitenames=10,giveninits=true,doi=false,url=true,backref=true]{biblatex}
\newcommand*{\citet}[1]{\AtNextCite{\AtEachCitekey{\defcounter{maxnames}{999}}}\textcite{#1}}
\newcommand*{\citep}[1]{\cite{#1}}

\title{Composition of Differential Privacy \&\\ Privacy Amplification by Subsampling}
\author{\href{http://www.thomas-steinke.net/}{Thomas Steinke}\thanks{Google Research~\dotfill~\texttt{steinke@google.com}}}

\begin{document}
\maketitle

\begin{abstract}
    This chapter is meant to be part of the book ``Differential Privacy for Artificial Intelligence Applications.''
    We give an introduction to the most important property of differential privacy -- composition: running multiple independent analyses on the data of a set of people will still be differentially private as long as each of the analyses is private on its own -- as well as the related topic of privacy amplification by subsampling.
    This chapter introduces the basic concepts and gives proofs of the key results needed to apply these tools in practice.
\end{abstract}

\newpage

\tableofcontents

\newpage

\section{Introduction}
    Our data is subject to many different uses. Many entities will have access to our data, including government agencies, healthcare providers, employers, technology companies, and financial institutions. Those entities will perform many different analyses that involve our data and those analyses will be updated repeatedly over our lifetimes. The greatest risk to privacy is that an attacker will combine multiple pieces of information from the same or different sources and that the combination of these will reveal sensitive details about us.
    Thus we cannot study privacy leakage in a vacuum; it is important that we can reason about the accumulated privacy leakage over multiple independent analyses. 
    
    As a concrete example to keep in mind, consider the following simple differencing attack: Suppose your employer provides healthcare benefits. The employer pays for these benefits and thus may have access to summary statistics like how many employees are currently receiving pre-natal care or currently are being treated for cancer. Your pregnancy or cancer status is highly sensitive information, but intuitively the aggregated count is not sensitive as it is not specific to you. However, this count may be updated on a regular basis and your employer may notice that the count increased on the day you were hired or on the day you took off for a medical appointment. This example shows how multiple pieces of information -- the date of your hire or medical appointment, the count before that date, and the count afterwards -- can be combined to reveal sensitive information about you, despite each piece of information seeming innocuous on its own. Attacks could combine many different statistics from multiple sources and hence we need to be careful to guard against such attacks, which leads us to differential privacy.
    
    Differential privacy has strong composition properties -- if multiple independent analyses are run on our data and each analysis is differentially private on its own, then the combination of these analyses is also differentially private. This property is key to the success of differential privacy. Composition enables building complex differentially private systems out of simple differentially private subroutines.  Composition allows the re-use data over time without fear of a catastrophic privacy failure. And, when multiple entities use the data of the same individuals, they do not need to coordinate to prevent an attacker from learning private details of individuals by combining the information released by those entities.
    To prevent the above differencing attack, we could independently perturb each count to make it differentially private; then taking the difference of two counts would be sufficiently noisy to obscure your pregnancy or cancer status.
    
    Composition is quantitative. The differential privacy guarantee of the overall system will depend on the number of analyses and the privacy parameters that they each satisfy. The exact relationship between these quantities can be complex. There are various composition theorems that give bounds on the overall parameters in terms of the parameters of the parts of the system.
    In this chapter, we will study several composition theorems (including the relevant proofs) and we will also look at some examples that demonstrate how to apply the composition theorems and why we need them.
    
    Composition theorems provide privacy bounds for a given system. A system designer must use composition theorems to design systems that simultaneously give good privacy and good utility (i.e., good statistical accuracy). This process often called ``privacy budgeting'' or ``privacy accounting.'' Intuitively, the system designer has some privacy constraint (i.e., the overall system must satisfy some final privacy guarantee) which can be viewed as analogous to a monetary budget that must be divided amongst the various parts of the system. Composition theorems provide the accounting rules for this budget. Allocating more of the budget to some part of the system makes that part more accurate, but then less budget is available for other parts of the system. Thus the system designer must also make a value judgement about which parts of the system to prioritize. 
    
\section{Basic Composition}\label{sec:basic_composition}
    The simplest composition theorem is what is known as basic composition. This applies to pure $\varepsilon$-DP (although it can be extended to approximate $(\varepsilon,\delta)$-DP). Basic composition says that, if we run $k$ independent $\varepsilon$-DP algorithms, then the composition of these is $k\varepsilon$-DP. More generally, we have the following result.
    
    \begin{theorem}[Basic Composition]\label{thm:basic_composition}
        Let $M_1, M_2, \cdots, M_k : \mathcal{X}^n \to \mathcal{Y}$ be randomized algorithms. Suppose $M_j$ is $\varepsilon_j$-DP for each $j \in [k]$.
        Define $M : \mathcal{X}^n \to \mathcal{Y}^k$ by $M(x)=(M_1(x),M_2(x),\cdots,M_k(x))$, where each algorithm is run independently. Then $M$ is $\varepsilon$-DP for $\varepsilon = \sum_{j=1}^k \varepsilon_j$.
    \end{theorem}
    \begin{proof}
        Fix an arbitrary pair of neighbouring datasets $x,x' \in \mathcal{X}^n$ and output $y \in \mathcal{Y}^k$.
        To establish that $M$ is $\varepsilon$-DP, we must show that $e^{-\varepsilon} \le \frac{\pr{}{M(x)=y}}{\pr{}{M(x')=y}} \le e^\varepsilon$. By independence, we have \[\frac{\pr{}{M(x)=y}}{\pr{}{M(x')=y}} = \frac{\prod_{j=1}^k\pr{}{M_j(x)=y_j}}{\prod_{j=1}^k\pr{}{M_j(x')=y_j}} =  \prod_{j=1}^k \frac{\pr{}{M_j(x)=y_j}}{\pr{}{M_j(x')=y_j}} \le \prod_{j=1}^k e^{\varepsilon_j} = e^{\sum_{j=1}^k \varepsilon_j} = e^\varepsilon,\] where the inequality follows from the fact that each $M_j$ is $\varepsilon_j$-DP and, hence, $e^{-\varepsilon_j} \le \frac{\pr{}{M_j(x)=y_j}}{\pr{}{M_j(x')=y_j}} \le e^{\varepsilon_j}$. Similarly, $\prod_{j=1}^k \frac{\pr{}{M_j(x)=y_j}}{\pr{}{M_j(x')=y_j}} \ge \prod_{j=1}^k e^{-\varepsilon_j}$, which completes the proof.
    \end{proof}
    
    Basic composition is already a powerful result, despite its simple proof; it establishes the versatility of differential privacy and allows us to begin reasoning about complex systems in terms of their building blocks. For example, suppose we have $k$ functions $f_1, \cdots, f_k : \mathcal{X}^n \to \mathbb{R}$ each of sensitivity $1$. For each $j \in [k]$, we know that adding $\mathsf{Laplace}(1/\varepsilon)$ noise to the value of $f_j(x)$ satisfies $\varepsilon$-DP. Thus, if we add independent $\mathsf{Laplace}(1/\varepsilon)$ noise to each value $f_j(x)$ for all $j \in [k]$, then basic composition tells us that releasing this vector of $k$ noisy values satisfies $k\varepsilon$-DP. If we want the overall system to be $\varepsilon$-DP, then we should add independent $\mathsf{Laplace}(k/\varepsilon)$ noise to each value $f_j(x)$.
    
    \subsection{Is Basic Composition Optimal?}\label{sec:basic_optimal}
    
    If we want to release $k$ values each of sensitivity $1$ (as above) and have the overall release be $\varepsilon$-DP, then, using basic composition, we can add $\mathsf{Laplace}(k/\varepsilon)$ noise to each value. The variance of the noise for each value is $2k^2/\varepsilon^2$, so the standard deviation is $\sqrt{2} k /\varepsilon$. In other words, the scale of the noise must grow linearly with the number of values $k$ if the overall privacy and each value's sensitivity is fixed. It is natural to wonder whether the scale of the Laplace noise can be reduced by improving the basic composition result. We now show that this is not possible.
    
    For each $j \in [k]$, let $M_j : \mathcal{X}^n \to \mathbb{R}$ be the algorithm that releases $f_j(x)$ with $\mathsf{Laplace}(k/\varepsilon)$ noise added. Let $M : \mathcal{X}^n \to \mathbb{R}^k$ be the composition of these $k$ algorithms. Then $M_j$ is $\varepsilon/k$-DP for each $j \in [k]$ and basic composition tells us that $M$ is $\varepsilon$-DP. The question is whether $M$ satisfies a better DP guarantee than this -- i.e., does $M$ satisfy $\varepsilon_*$-DP for some $\varepsilon_*<\varepsilon$?
    Suppose we have neighbouring datasets $x,x'\in\mathcal{X}^n$ such that $f_j(x) = f_j(x')+1$ for each $j \in [k]$. Let $y=(a,a,\cdots,a) \in \mathbb{R}^k$ for some $a \ge \max_{j=1}^k f_j(x)$.
    Then 
    \begin{align*}
        \frac{\pr{}{M(x)=y}}{\pr{}{M(x')=y}} &= \frac{\prod_{j=1}^k \pr{}{f_j(x)+\mathsf{Laplace}(k/\varepsilon)=y_j}}{\prod_{j=1}^k \pr{}{f_j(x')+\mathsf{Laplace}(k/\varepsilon)=y_j}} \\
         &= \prod_{j=1}^k \frac{\pr{}{\mathsf{Laplace}(k/\varepsilon)=y_j-f_j(x)}}{\pr{}{\mathsf{Laplace}(k/\varepsilon)=y_j-f_j(x')}} \\
         &= \prod_{j=1}^k \frac{\frac{\varepsilon}{2k}\exp\left(-\frac{\varepsilon}{k} |y_j-f_j(x)| \right)}{\frac{\varepsilon}{2k}\exp\left(-\frac{\varepsilon}{k} |y_j-f_j(x')| \right)} \\
         &= \prod_{j=1}^k \frac{\exp\left(-\frac{\varepsilon}{k} (y_j-f_j(x)) \right)}{\exp\left(-\frac{\varepsilon}{k} (y_j-f_j(x')) \right)} \tag{$y_j \ge f_j(x)$ and $y_j \ge f_j(x')$} \\
         &= \prod_{j=1}^k \exp\left(\frac{\varepsilon}{k}\left(f_j(x)-f_j(x')\right)\right) \\
         &= \exp\left( \frac{\varepsilon}{k} \sum_{j=1}^k \left(f_j(x)-f_j(x')\right)\right)= e^\varepsilon.
    \end{align*}
    This shows that basic composition is optimal. For this example, we cannot prove a better guarantee than what is given by basic composition.
    
    Is there some other way to improve upon basic composition that circumvents this example? Note that we assumed that there are neighbouring datasets $x,x'\in\mathcal{X}^n$ such that $f_j(x) = f_j(x')+1$ for each $j \in [k]$. In some settings, no such worst case datasets exist. In that case, instead of scaling the noise linearly with $k$, we can scale the Laplace noise according to the $\ell_1$ sensitivity $\Delta_1 := \sup_{x,x' \in \mathcal{X}^n \atop \text{neighbouring}} \sum_{j=1}^k |f_j(x)-f_j(x')|$. 
    
    Instead of adding assumptions to the problem, we will look more closely at the example above.
    We showed that there exists some output $y \in \mathbb{R}^d$ such that $\frac{\pr{}{M(x)=y}}{\pr{}{M(x')=y}} = e^\varepsilon$.
    However, such outputs $y$ are very rare, as we require $y_j \ge \max\{f_j(x),f_j(x')\}$ for each $j \in [k]$ where $y_j = f_j(x) + \mathsf{Laplace}(k/\varepsilon)$. Thus, in order to observe an output $y$ such that the likelihood ratio is maximal, all of the $k$ Laplace noise samples must be positive, which happens with probability $2^{-k}$. 
    The fact that outputs $y$ with maximal likelihood ratio are exceedingly rare turns out to be a general phenomenon and not specific to the example above. %
    
    Can we improve on basic composition if we only ask for a high probability bound? That is, instead of demanding $\frac{\pr{}{M(x)=y}}{\pr{}{M(x')=y}} \le e^{\varepsilon_*}$ for all $y \in \mathcal{Y}$, we demand $\pr{Y \gets M(x)}{\frac{\pr{}{M(x)=Y}}{\pr{}{M(x')=Y}} \le e^{\varepsilon_*}} \ge 1-\delta$ for some $0 < \delta \ll 1$. Can we prove a better bound $\varepsilon_* < \varepsilon$ in this relaxed setting? The answer turns out to be yes.
    
    The limitation of pure $\varepsilon$-DP is that events with tiny probability -- which are negligible in real-world applications -- can dominate the privacy analysis. This motivates us to move to relaxed notions of differential privacy, such as approximate $(\varepsilon,\delta)$-DP and concentrated DP, which are less sensitive to low probability events. In particular, these relaxed notions of differential privacy allow us to prove quantitatively better composition theorems. The rest of this chapter develops this direction further.
    
\section{Privacy Loss Distributions}
    Qualitatively, an algorithm $M : \mathcal{X}^n \to \mathcal{Y}$ is differentially private if, for all neighbouring datasets $x,x' \in \mathcal{X}^n$, the output distributions $M(x)$ and $M(x')$ are ``indistinguishable'' or ``close.''
    The key question is how do we quantify the closeness or indistinguishability of a pair of distributions?
    
    Pure DP (a.k.a.~pointwise DP) \cite{dwork2006calibrating} uniformly bounds the likelihood ratio -- $\frac{\pr{}{M(x)=y}}{\pr{}{M(x')=y}} \le e^{\varepsilon}$ for all $y \in \mathcal{Y}$.
    As discussed at the end of the section on basic composition (\S\ref{sec:basic_composition}), this can be too strong as the outputs $y$ that maximize this likelihood ratio may be very rare.
    
    \newcommand{\tvd}[2]{\mathrm{d}_{\text{TV}}\left(#1,#2\right)}
    We could also consider the total variation distance (a.k.a.~statistical distance): \[\tvd{M(x)}{M(x')} := \sup_{S \subset \mathcal{Y}} \left( \pr{}{M(x) \in S} - \pr{}{M(x') \in S} \right).\]
    Another option would be the KL divergence (a.k.a.~relative entropy).
    Both TV distance and KL divergence turn out to give poor privacy-utility tradeoffs; that is, to rule out bad algorithms $M$, we must set these parameters very small, but that also rules out all the good algorithms. Intuitively, both TV and KL are not sensitive enough to low-probability bad events (whereas pure DP is too sensitive). We need to introduce a parameter ($\delta$) to determine what level of low probability events we can ignore.
    
    Approximate $(\varepsilon,\delta)$-DP \cite{dwork2006our} is a combination of pure $\varepsilon$-DP and $\delta$ TV distance. Specifically, $M$ is $(\varepsilon,\delta)$-DP if, for all neighbouring datasets $x,x' \in \mathcal{X}^n$ and all measurable $S \subset \mathcal{Y}$, $\pr{}{M(x) \in S} \le e^\varepsilon \cdot \pr{}{M(x') \in S} + \delta$. Intuitively, $(\varepsilon,\delta)$-DP is like $\varepsilon$-DP except we can ignore events with probability $\le \delta$. That is, $\delta$ represents a failure probability, so it should be small (e.g., $\delta \le 10^{-6}$), while $\varepsilon$ can be larger (e.g., $\varepsilon \approx 1$); having two parameters with very different values allows us to circumvent the limitations of either pure DP or TV distance as a similarity measure.
    
    All of these options for quantifying indistinguishability can be viewed from the perspective of the privacy loss distribution. The privacy loss distribution also turns out to be essential to the analysis of composition. Approximate $(\varepsilon,\delta)$-DP bounds are usually proved via the privacy loss distribution.
    
    We now formally define the privacy loss distribution and relate it to the various quantities we have considered. Then (in \S\ref{sec:privloss_gauss}) we will calculate the privacy loss distribution corresponding to the Gaussian mechanism, which is a particularly nice example. In the next subsection (\S\ref{sec:statistical_perspective}), we explain how the privacy loss distribution arises naturally via statistical hypothesis testing. To conclude this section (\S\ref{sec:privloss_adp}), we precisely relate the privacy loss back to approximate $(\varepsilon,\delta)$-DP. In the next section (\S\ref{sec:comp_privloss}), we will use the privacy loss distribution as a tool to analyze composition.
    
    \newcommand{\privloss}[2]{\mathsf{PrivLoss}\left({#1}\middle\|{#2}\right)}
    \newcommand{\llr}[2]{f_{\left.{#1}\middle\|{#2}\right.}}
    \begin{definition}[Privacy Loss Distribution]\label{defn:priv_loss}
        Let $P$ and $Q$ be two probability distributions on $\mathcal{Y}$. Define $\llr{P}{Q} : \mathcal{Y} \to \mathbb{R}$ by $\llr{P}{Q}(y) = \log(P(y)/Q(y))$.\footnote{The function $\llr{P}{Q}$ is called the log likelihood ratio of $P$ with respect to $Q$. Formally, $\llr{P}{Q}$ is the natural logarithm of the Radon-Nikodym derivative of $P$ with respect to $Q$. This function is defined by the property that $P(S) = \ex{Y \gets P}{\mathbb{I}[Y \in S]} = \ex{Y \gets Q}{e^{\llr{P}{Q}(Y)} \cdot \mathbb{I}[Y \in S]}$ for all measurable $S \subset \mathcal{Y}$. For this to exist, we must assume that $P$ and $Q$ have the same sigma-algebra and that $P$ is absolutely continuous with respect to $Q$ and vice versa -- i.e., $\forall S \subset \mathcal{Y} ~~~ Q(S) = 0 \iff P(S) = 0$.}
        The privacy loss random variable is given by $Z = \llr{P}{Q}(Y)$ for $Y \gets P$.
        The distribution of $Z$ is denoted $\privloss{P}{Q}$.
    \end{definition}
    In the context of differential privacy, the distributions $P=M(x)$ and $Q=M(x')$ correspond to the outputs of the algorithm $M$ on neighbouring inputs $x,x'$. Successfully distinguishing these distributions corresponds to learning some fact about an individual person's data.
    The randomness of the privacy loss random variable $Z$ comes from the randomness of the algorithm $M$ (e.g., added noise). 
    Intuitively, the privacy loss tells us which input ($x$ or $x'$) is more likely given the observed output ($Y \gets M(\cdot)$). If $Z>0$, then the hypothesis $Y \gets P = M(x)$ explains the observed output better than the hypothesis $Y \gets Q = M(x')$ and vice versa. The magnitude of the privacy loss $Z$ indicates how strong the evidence for this conclusion is. If $Z=0$, both hypotheses explain the output equally well, but, if $Z \to \infty$, then we can be nearly certain that the output came from $P$, rather than $Q$. A very negative privacy loss $Z \ll 0$ means that the observed output $Y \gets P$ strongly supports the wrong hypothesis (i.e., $Y \gets Q$).
    
    As long as the privacy loss distribution is well-defined,\footnote{The privacy loss distribution is not well-defined if absolute continuity fails to hold. Intuitively, this corresponds to the privacy loss being infinite. We can extend most of these definitions to allow for an infinite privacy loss. For simplicity, we do not delve into these issues.} we can easily express almost all the quantities of interest in terms of it: 
    \begin{itemize}
        \item Pure $\varepsilon$-DP of $M$ is equivalent to demanding that $\pr{Z \gets \privloss{M(x)}{M(x')}}{Z \le \varepsilon} = 1$ for all neighbouring $x,x'$.\footnote{Note that, by the symmetry of the neighbouring relation (i.e., if $x,x'$ are neighbouring datasets then $x',x$ are also neighbours), we also have $\pr{Z \gets \privloss{M(x)}{M(x')}}{Z \ge -\varepsilon} = 1$ as a consequence of $\pr{Z' \gets \privloss{M(x')}{M(x)}}{Z' \le \varepsilon} = 1$.}
        \item The KL divergence is the expectation of the privacy loss: $\dr{1}{P}{Q} := \ex{Z \gets \privloss{P}{Q}}{Z}$.\footnote{The expectation of the privacy loss is always non-negative. Intuitively, this is because we take the expectation of the log likelihood ratio $\llr{P}{Q}(Y)$ with respect to $Y \gets P$ -- i.e., the true answer is $P$, so on average the log likelihood ratio should point towards the correct answer.}
        \item The TV distance is given by \[\tvd{P}{Q} = \ex{Z \gets \privloss{P}{Q}}{\max\{0,1-\exp(-Z)\}} = \frac12 \ex{Z \gets \privloss{P}{Q}}{\left|1-\exp(-Z)\right|}.\]
        \item Approximate $(\varepsilon,\delta)$-DP of $M$ is implied by $\pr{Z \gets \privloss{M(x)}{M(x')}}{Z\le\varepsilon}\ge1-\delta$ for all neighbouring $x,x'$. So we should think of approximate DP as a tail bound on the privacy loss. 
        To be precise, $(\varepsilon,\delta)$-DP of $M$ is equivalent to \[\ex{Z \gets \privloss{M(x)}{M(x')}}{\max\{0,1-\exp(\varepsilon-Z)\}}\le\delta\] for all neighbouring $x,x'$. (See Proposition \ref{prop:privloss_adp}.)
    \end{itemize}
    
    \subsection{Privacy Loss of Gaussian Noise Addition}\label{sec:privloss_gauss}
    
    As an example, we will work out the privacy loss distribution corresponding to the addition of Gaussian noise to a bounded-sensitivity query. This example is particularly clean, as the privacy loss distribution is also a Gaussian, and it will turn out to be central to the story of composition.
    
    \begin{proposition}[Privacy Loss Distribution of Gaussian]\label{prop:gauss_privloss}
        Let $P=\mathcal{N}(\mu,\sigma^2)$ and $Q=\mathcal{N}(\mu',\sigma^2)$. Then $\privloss{P}{Q} = \mathcal{N}(\rho,2\rho)$ for $\rho = \frac{(\mu-\mu')^2}{2\sigma^2}$.
    \end{proposition}
    \begin{proof}
        We have $P(y) = \frac{1}{\sqrt{2\pi\sigma^2}} \exp\left(-\frac{(y-\mu)^2}{2\sigma^2}\right)$ and $Q(y) = \frac{1}{\sqrt{2\pi\sigma^2}} \exp\left(-\frac{(y-\mu')^2}{2\sigma^2}\right)$. Thus the log likelihood ratio is
        \begin{align*}
            \llr{P}{Q}(y) &= \log\left(\frac{P(y)}{Q(y)}\right) \\
            &= \log\left(\frac{\frac{1}{\sqrt{2\pi\sigma^2}} \exp\left(-\frac{(y-\mu)^2}{2\sigma^2}\right)}{\frac{1}{\sqrt{2\pi\sigma^2}} \exp\left(-\frac{(y-\mu')^2}{2\sigma^2}\right)}\right) \\
            &= - \frac{(y-\mu)^2}{2\sigma^2} + \frac{(y-\mu')^2}{2\sigma^2} \\
            &= \frac{(y^2-2\mu'y+\mu'^2)-(y^2 - 2\mu y + \mu^2)}{2\sigma^2} \\
            &= \frac{2(\mu-\mu')y - \mu^2 + \mu'^2}{2\sigma^2} \\
            &= \frac{(\mu-\mu')(2y -\mu-\mu')}{2\sigma^2}.
        \end{align*}
        The log likelihood ratio $\llr{P}{Q}$ is an affine linear function. Thus the privacy loss random variable $Z = \llr{P}{Q}(Y)$ for $Y \gets P = \mathcal{N}(\mu,\sigma^2)$ will also follow a Gaussian distribution.
        Specifically, $\ex{}{Y}=\mu$, so \[\ex{}{Z} = \frac{(\mu-\mu')(2\ex{}{Y}-\mu-\mu')}{2\sigma^2} = \frac{(\mu-\mu')^2}{2\sigma^2} = \rho\]
        and, similarly, $\var{}{Y}=\sigma^2$, so \[\var{}{Z} = \frac{((2(\mu-\mu'))^2}{(2\sigma^2)^2}\cdot\var{}{Y} = \frac{(\mu-\mu')^2}{\sigma^2} = 2\rho.\]
    \end{proof}
    
    To relate Proposition \ref{prop:gauss_privloss} to the standard Gaussian mechanism $M : \mathcal{X}^n \to \mathbb{R}$, recall that $M(x) = \mathcal{N}(q(x),\sigma^2)$, where $q$ is a sensitivity-$\Delta$ query -- i.e., $|q(x)-q(x')| \le \Delta$ for all neighbouring datasets $x,x' \in \mathcal{X}^n$. Thus, for neighbouring datasets $x,x'$, we have $\privloss{M(x)}{M(x')} = \mathcal{N}(\rho,2\rho)$ for some $\rho \le \frac{\Delta^2}{2\sigma^2}$.
    
    The privacy loss of the Gaussian mechanism is unbounded; thus it does not satisfy pure $\varepsilon$-DP. However, the Gaussian distribution is highly concentrated, so we can say that with high probability the privacy loss is not too large. This is the basis of the privacy guarantee of the Gaussian mechanism.
    
    \subsection{Statistical Hypothesis Testing Perspective}\label{sec:statistical_perspective}
    
    To formally quantify differential privacy, we must measure the closeness or indistinguishability of the distributions $P=M(x)$ and $Q=M(x')$ corresponding to the outputs of the algorithm $M$ on neighbouring inputs $x,x'$. 
    Distinguishing a pair of distributions is precisely the problem of (simple) hypothesis testing in the field of statistical inference.
    Thus it is natural to look at hypothesis testing tools to quantify the (in)distinguishability of a pair of distributions.
    
    In the language of hypothesis testing, the two distributions $P$ and $Q$ would be the null hypothesis and the alternate hypothesis, which correspond to a positive or negative example. We are given a sample $Y$ drawn from one of the two distributions and our task is to determine which.
    Needless to say, there is, in general, no hypothesis test that perfectly distinguishes the two distributions and, when choosing a hypothesis test, we face a non-trivial tradeoff between false positives and false negatives. There are many different ways to measure how good a given hypothesis test is.
    
    For example, we could measure the accuracy of the hypothesis test evenly averaged over the two distributions. In this case, given the sample $Y$, an optimal test chooses $P$ if $P(Y) \ge Q(Y)$ and otherwise chooses $Q$; the accuracy of this test is \[ \frac12 \pr{Y \gets P}{P(Y) \ge Q(Y)} + \frac12 \pr{Y \gets Q}{P(Y) < Q(Y)}  = \frac12 + \frac12\tvd{P}{Q}.\] This measure of accuracy thus corresponds to TV distance. The greater the TV distance between the distributions, the more accurate this test is. However, as we mentioned earlier, TV distance does not yield good privacy-utility tradeoffs. Intuitively, the problem is that this hypothesis test doesn't care about how confident we are. That is, the test only asks whether $P(Y) \ge Q(Y)$, but not how big the difference or ratio is. Hence we want a more refined measure of accuracy that does not count false positives and false negatives equally.
    
    Regardless of how we measure how good the hypothesis test is, there is an optimal test statistic, namely the log likelihood ratio. This test statistic gives a real number and thresholding that value yields a binary hypothesis test; \emph{any} binary hypothesis test is dominated by some value of the threshold. In other words, the tradeoff between false positives and false negatives reduces to picking a threshold. This remarkable -- yet simple -- fact is established by the Neyman-Pearson lemma:
    \begin{lemma}[Neyman-Pearson Lemma {\cite{neyman1933ix}}]
        Fix distributions $P$ and $Q$ on $\mathcal{Y}$ and define the log-likelihood ratio test statistic $\llr{P}{Q} : \mathcal{Y} \to \mathbb{R}$ by $\llr{P}{Q}(y)=\log\left(\frac{P(y)}{Q(y)}\right)$.
        Let $T : \mathcal{Y} \to \{P,Q\}$ be any (possibly randomized) test.
        Then there exists some $t \in \mathbb{R}$ such that \[\pr{Y \gets P}{T(Y)=P} \le \pr{Y \gets P}{\llr{P}{Q}(Y) \ge t} ~~~~ \text{ and } ~~~~ \pr{Y \gets Q}{T(Y)=Q} \le \pr{Y \gets Q}{\llr{P}{Q}(Y) \le t}.\]
    \end{lemma}
    
    How is this related to the privacy loss distribution? The test statistic $Z=\llr{P}{Q}(Y)$ under the hypothesis $Y \gets P$ is precisely the privacy loss random variable $Z \gets \privloss{P}{Q}$. 
    Thus the Neyman-Pearson lemma tells us that the privacy loss distribution $\privloss{P}{Q}$ captures everything we need to know about distinguishing $P$ from $Q$.
    
    Note that the Neyman-Pearson lemma also references the test statistic $\llr{P}{Q}(Y)$ under the hypothesis $Y \gets Q$. This is fundamentally not that different from the privacy loss. There are two ways we can relate this quantity back to the usual privacy loss: 
    First, we can relate it to $\privloss{Q}{P}$ and this distribution is something we should be able to handle due to the symmetry of differential privacy guarantees.
    \begin{remark}\label{rem:dual_privloss}
        Fix distributions $P$ and $Q$ on $\mathcal{Y}$ such that the log likelihood ratio $\llr{P}{Q}(y)=\log\left(\frac{P(y)}{Q(y)}\right)$ is well-defined for all $y \in \mathcal{Y}$.
        Since $\llr{P}{Q}(y)=-\llr{Q}{P}(y)$ for all $y \in \mathcal{Y}$, if $Z \gets \privloss{Q}{P}$, then $-Z$ follows the distribution of $\llr{P}{Q}(Y)$ under the hypothesis $Y \gets Q$.
    \end{remark}
    Second, if we need to compute an expectation of some function $g$ of $\llr{P}{Q}(Y)$ under the hypothesis $Y \gets Q$, then we can still express this in terms of the privacy loss $\privloss{P}{Q}$: 
    \begin{lemma}[Change of Distribution for Privacy Loss]\label{lem:dual_privloss}
        Fix distributions $P$ and $Q$ on $\mathcal{Y}$ such that the log likelihood ratio $\llr{P}{Q}(y)=\log\left(\frac{P(y)}{Q(y)}\right)$ is well-defined for all $y \in \mathcal{Y}$.
        Let $g : \mathbb{R} \to \mathbb{R}$ be measurable.
        Then
        \[\ex{Y \gets Q}{g(\llr{P}{Q}(Y))} = \ex{Z \gets \privloss{P}{Q}}{g(Z) \cdot e^{-Z}}.\]
    \end{lemma}
    \begin{proof}
        By the definition of the log likelihood ratio (see Definition \ref{defn:priv_loss}), we have $\ex{Y \gets P}{h(Y)} = \ex{Y \gets Q}{h(Y) \cdot e^{\llr{P}{Q}(Y)}}$ for all measurable functions $h$. Setting $h(y) = g(\llr{P}{Q}(y)) \cdot e^{-\llr{P}{Q}(y)}$ yields $\ex{Z \gets \privloss{P}{Q}}{g(Z) \cdot e^{-Z}} = \ex{Y \gets P}{h(Y)} = \ex{Y \gets Q}{h(Y) \cdot e^{\llr{P}{Q}(Y)}} = \ex{Y \gets Q}{g(\llr{P}{Q}(Y))}$, as required. We can also write these expressions out as an integral to obtain a more intuitive proof:
        \begin{align*}
            \ex{Y \gets Q}{g(\llr{P}{Q}(Y))} &= \int_\mathcal{Y} g(\llr{P}{Q}(y)) \cdot Q(y) \mathrm{d}y \\
            &= \int_\mathcal{Y} g(\llr{P}{Q}(y)) \cdot \frac{Q(y)}{P(y)} \cdot P(y) \mathrm{d}y \\
            &= \int_\mathcal{Y} g(\llr{P}{Q}(y)) \cdot e^{-\log(P(y)/Q(y))} \cdot P(y) \mathrm{d}y \\
            &= \int_\mathcal{Y} g(\llr{P}{Q}(y)) \cdot e^{-\llr{P}{Q}(y)} \cdot P(y) \mathrm{d}y \\
            &= \ex{Y \gets P}{g(\llr{P}{Q}(Y)) \cdot e^{-\llr{P}{Q}(Y)}} \\
            &= \ex{Z \gets \privloss{P}{Q}}{g(Z) \cdot e^{-Z}}.
        \end{align*}
    \end{proof}

    \subsection{Approximate DP \& the Privacy Loss Distribution}\label{sec:privloss_adp}
    
    So far, in this section, we have defined the privacy loss distribution, given an example, and illustrated that it is a natural quantity to consider that captures essentially everything we need to know about the (in)distinguishability of two distributions.
    To wrap up this section, we will relate the privacy loss distribution back to the definition of approximate $(\varepsilon,\delta)$-DP:
    
    \begin{proposition}[Conversion from Privacy Loss Distribution to Approximate Differential Privacy]\label{prop:privloss_adp}
        Let $P$ and $Q$ be two probability distributions on $\mathcal{Y}$ such that the privacy loss distribution $\privloss{P}{Q}$ is well-defined.
        Fix $\varepsilon \ge 0$ and define \[\delta := \sup_{S \subset \mathcal{Y}} P(S) - e^\varepsilon \cdot Q(S).\]
        Then 
        \begin{align*}
            \delta 
            &= \pr{Z \gets \privloss{P}{Q}}{Z>\varepsilon} - e^\varepsilon \cdot \pr{Z' \gets \privloss{Q}{P}}{-Z'>\varepsilon} \\
            &= \ex{Z \gets \privloss{P}{Q}}{\max\{0,1-\exp(\varepsilon-Z)\}} \\
            &= \int_\varepsilon^\infty e^{\varepsilon - Z} \cdot \pr{Z \gets \privloss{P}{Q}}{Z>z} \mathrm{d}z \\
            &\le \pr{Z \gets \privloss{P}{Q}}{Z > \varepsilon}.
        \end{align*}
    \end{proposition}
    \begin{proof}
        For any measurable $S \subset \mathcal{Y}$, we have \[P(S) - e^\varepsilon \cdot Q(S) = \int_\mathcal{Y} \mathbb{I}[y \in S] \cdot \left( P(y) - e^\varepsilon \cdot Q(y) \right) \mathrm{d}y,\] where $\mathbb{I}$ denotes the indicator function -- it takes the value $1$ if the condition is true and $0$ otherwise.
        To maximize this expression, we want $y \in S$ whenever $P(y) - e^\varepsilon \cdot Q(y) >0$ and we want $y \notin S$ when this is negative. Thus $\delta = P(S_*)-e^\varepsilon\cdot Q(S_*)$ for \[S_* := \left\{ y \in \mathcal{Y} : P(y) - e^\varepsilon \cdot Q(y) > 0 \right\} = \left\{y \in \mathcal{Y} : \llr{P}{Q}(y) > \varepsilon \right\}.\]
        
        Now \[P(S_*) = \pr{Y \gets P}{\llr{P}{Q}(Y) > \varepsilon} = \pr{Z \gets \privloss{P}{Q}}{Z > \varepsilon}\]
        and, by Remark \ref{rem:dual_privloss}, \[Q(S_*) = \pr{Y \gets Q}{\llr{P}{Q}(Y) > \varepsilon} = \pr{Z' \gets \privloss{Q}{P}}{-Z' > \varepsilon}.\]
        This gives the first expression in the result: \[\delta = P(S_*)-e^\varepsilon\cdot Q(S_*) = \pr{Z \gets \privloss{P}{Q}}{Z > \varepsilon} - e^\varepsilon \cdot \pr{Z' \gets \privloss{Q}{P}}{-Z' > \varepsilon}. \]
        
        Alternatively, $P(S_*) = \ex{Z \gets \privloss{P}{Q}}{\mathbb{I}[Z>\varepsilon]}$ and, by Lemma \ref{lem:dual_privloss}, \[Q(S_*) = \ex{Y \gets Q}{\mathbb{I}[\llr{P}{Q}(Y) > \varepsilon]} = \ex{Z \gets \privloss{P}{Q}}{\mathbb{I}[Z > \varepsilon] \cdot e^{-Z}},\]
        which yields \[\delta = P(S_*)-e^\varepsilon\cdot Q(S_*) = \ex{Z \gets \privloss{P}{Q}}{(1-e^\varepsilon \cdot e^{-Z}) \cdot \mathbb{I}[Z>\varepsilon]}.\]
        Note that $(1-e^\varepsilon \cdot e^{-z}) \cdot \mathbb{I}[z>\varepsilon] = \max\{0,1-e^{\varepsilon-z}\}$ for all $z \in \mathbb{R}$. This produces the second expression in our result.
        
        To obtain the third expression in the result, we apply integration by parts to the second expression:
        Let $F(z) := \pr{Z \gets \privloss{P}{Q}}{Z>z}$ be the complement of the cumulative distribution function of the privacy loss distribution. Then the probability density function of $Z$ evaluated at $z$ is given by the negative derivative, $-F'(z)$.\footnote{In general, the privacy loss may not be continuous -- i.e., $F$ may not be differentiable. Nevertheless, the final result still holds in this case.}
        Then
        \begin{align*}
            \delta &= \ex{Z \gets \privloss{P}{Q}}{\max\{0,1-e^{\varepsilon-Z}\}} \\
            &= \int_\mathbb{R} \max\{0,1-e^{\varepsilon-z}\} \cdot (-F'(z)) \mathrm{d} z \\
            &= \int_\varepsilon^\infty (1-e^{\varepsilon-z}) \cdot (-F'(z)) \mathrm{d} z \\
            &= \int_\varepsilon^\infty \left( \frac{\mathrm{d}}{\mathrm{d}z} (1-e^{\varepsilon-z}) \cdot (-F(z)) \right) - (0-e^{\varepsilon-z}\cdot(-1)) \cdot (-F(z)) \mathrm{d}z \tag{product rule} \\
            &= \lim_{z \to \infty} (1-e^{\varepsilon-z}) \cdot (-F(z)) - (1-e^{\varepsilon-\varepsilon}) \cdot (-F(\varepsilon)) -  \int_\varepsilon^\infty e^{\varepsilon-z} \cdot (-F(z)) \mathrm{d}z \tag{fundamental theorem of calculus} \\
            &= -\lim_{z \to \infty} \pr{Z \gets \privloss{P}{Q}}{Z>z} + \int_\varepsilon^\infty e^{\varepsilon - z} \cdot \pr{Z \gets \privloss{P}{Q}}{Z>z} \mathrm{d}z.
        \end{align*}
        If the privacy loss is well-defined, then $\lim_{z \to \infty} \pr{Z \gets \privloss{P}{Q}}{Z>z}  = 0$.

        The final expression (an upper bound, rather than a tight characterization) is easily obtained from any of the other three expressions. In particular, dropping the second term $- e^\varepsilon \cdot \pr{Z' \gets \privloss{Q}{P}}{-Z'>\varepsilon} \le 0$ from the first expression yields the upper bound.
    \end{proof}

    The expression \(\delta = \sup_{S \subset \mathcal{Y}} P(S) - e^\varepsilon \cdot Q(S)\) in Proposition \ref{prop:privloss_adp} is known as the ``hockey stick divergence'' and it determines the smallest $\delta$ for a given $\varepsilon$ such that $P(S) \le e^\varepsilon Q(S) + \delta$ for all $S \subset \mathcal{Y}$. If $P=M(x)$ and $Q=M(x')$ for arbitrary neighbouring datasets $x,x'$, then this expression gives the best approximate $(\varepsilon,\delta)$-DP guarantee.
    
    Proposition \ref{prop:privloss_adp} gives us three equivalent ways to calculate $\delta$, each of which will be useful in different circumstances.
    To illustrate how to use Proposition \ref{prop:privloss_adp}, we combine it with Proposition \ref{prop:gauss_privloss} to prove a tight approximate differential privacy guarantee for Gaussian noise addition:
    
    \begin{corollary}[Tight Approximate Differential Privacy for Univariate Gaussian]\label{cor:gauss_adp_exact}
        Let $q : \mathcal{X}^n \to \mathbb{R}$ be a deterministic function and let $\Delta : = \sup_{x,x'\in\mathcal{X}^n \atop \text{neighbouring}} |q(x)-q(x')|$ be its sensitivity.
        Define a randomized algorithm $M : \mathcal{X}^n \to \mathbb{R}$ by $M(x) = \mathcal{N}(q(x),\sigma^2)$ for some $\sigma^2>0$.
        Then, for any $\varepsilon \ge 0$, $M$ satisfies $(\varepsilon,\delta)$-DP with \[\delta = \overline\Phi\left(\frac{\varepsilon-\rho_*}{\sqrt{2\rho_*}}\right) - e^\varepsilon \cdot \overline\Phi\left(\frac{\varepsilon+\rho_*}{\sqrt{2\rho_*}}\right),\] where $\rho_* := \Delta^2/2\sigma^2$ and $\overline\Phi(z) := \pr{G \gets \mathcal{N}(0,1)}{G>z} = \frac{1}{\sqrt{2\pi}} \int_z^\infty \exp(-t^2/2) \mathrm{d}t$.
        
        Furthermore, this guarantee is optimal -- for every $\varepsilon \ge 0$, there is no $\delta'<\delta$ such that $M$ is $(\varepsilon,\delta')$-DP for general $q$.
    \end{corollary}
    \begin{proof}
        Fix arbitrary neighbouring datasets $x,x'\in\mathcal{X}^n$ and $S \subset \mathcal{Y}$. Let $\mu=q(x)$ and $\mu'=q(x')$. Let $P=M(x)=\mathcal{N}(\mu,\sigma^2)$ and $Q=M(x')=\mathcal{N}(\mu',\sigma^2)$. We must show $P(S) \le e^\varepsilon \cdot Q(S) + \delta$ for arbitrary $\varepsilon \ge 0$ and the value $\delta$ given in the result.
        
        By Proposition \ref{prop:gauss_privloss}, $\privloss{P}{Q} = \privloss{Q}{P} = \mathcal{N}(\rho,2\rho)$, where $\rho = \frac{(\mu-\mu')^2}{2\sigma^2} \le \rho_* = \frac{\Delta^2}{2\sigma^2}$. 
        
        By Proposition \ref{prop:privloss_adp}, we have $P(S) \le e^\varepsilon \cdot Q(S) + \delta$, where 
        \begin{align*}
            \delta &= \pr{Z \gets \privloss{P}{Q}}{Z>\varepsilon} - e^\varepsilon \cdot \pr{Z' \gets \privloss{Q}{P}}{-Z'>\varepsilon} \\
            &= \pr{Z \gets \mathcal{N}(\rho,2\rho)}{Z>\varepsilon} - e^\varepsilon \cdot \pr{Z' \gets \mathcal{N}(\rho,2\rho)}{-Z'>\varepsilon} \\
            &= \pr{G \gets \mathcal{N}(0,1)}{\rho + \sqrt{2\rho}\cdot G >\varepsilon} - e^\varepsilon \cdot \pr{G \gets \mathcal{N}(0,1)}{-\rho + \sqrt{2\rho}\cdot G >\varepsilon} \\
            &= \pr{G \gets \mathcal{N}(0,1)}{G > \frac{\varepsilon-\rho}{\sqrt{2\rho}}} - e^\varepsilon \cdot \pr{G \gets \mathcal{N}(0,1)}{G >\frac{\varepsilon+\rho}{\sqrt{2\rho}}} \\
            &= \overline\Phi\left(\frac{\varepsilon-\rho}{\sqrt{2\rho}}\right) - e^\varepsilon \cdot \overline\Phi\left(\frac{\varepsilon+\rho}{\sqrt{2\rho}}\right).
        \end{align*}
        Since $\rho \le \rho_*$ and the above expression is increasing in $\rho$, we can substitute in $\rho_*$ as an upper bound.
        
        Optimality follows from the fact that both Propositions \ref{prop:gauss_privloss} and \ref{prop:privloss_adp} give exact characterizations. Note that we must assume that there exist neighbouring $x,x'$ such that $\rho = \rho_*$.
    \end{proof}

    The guarantee of Corollary \ref{cor:gauss_adp_exact} is exact, but it is somewhat hard to interpret. We can easily obtain a more interpretable upper bound: 
    \begin{align*}
        \delta &= \overline\Phi\left(\frac{\varepsilon-\rho_*}{\sqrt{2\rho_*}}\right) - e^\varepsilon \cdot \overline\Phi\left(\frac{\varepsilon+\rho_*}{\sqrt{2\rho_*}}\right) \\
        &\le \overline\Phi\left(\frac{\varepsilon-\rho_*}{\sqrt{2\rho_*}}\right) = \pr{G \gets \mathcal{N}(0,1)}{G > \frac{\varepsilon-\rho_*}{\sqrt{2\rho_*}}} \\
        &\le \frac{\exp\left(-\frac{(\varepsilon-\rho_*)^2}{4\rho_*}\right)}{\max\left\{2, \sqrt{\frac{\pi}{\rho_*}} \cdot (\varepsilon-\rho_*)\right\}}. \tag{assuming $\varepsilon \ge \rho_*$} %
    \end{align*}
    
\section{Composition via the Privacy Loss Distribution}\label{sec:comp_privloss}

    The privacy loss distribution captures essentially everything about the (in)distinguishability of a pair of distributions. 
    It is also the key to understanding composition. 
    Suppose we run multiple differentially private algorithms on the same dataset and each has a well-defined privacy loss distribution.
    The composition of these algorithms corresponds to the convolution of the privacy loss distributions.
    That is, the privacy loss random variable corresponding to running all of the algorithms independently is equal to the sum of the independent privacy loss random variables of each of the algorithms:
    
    \begin{theorem}[Composition is Convolution of Privacy Loss Distributions]\label{thm:privloss_composition}
        For each $j \in [k]$, let $P_j$ and $Q_j$ be distributions on $\mathcal{Y}_j$ and assume $\privloss{P_j}{Q_j}$ is well defined.
        Let $P = P_1 \times P_2 \times \cdots \times P_k$ denote the product distribution on $\mathcal{Y} = \mathcal{Y}_1 \times \mathcal{Y}_2 \times \cdots \times \mathcal{Y}_k$ obtained by sampling independently from each $P_j$.
        Similarly, let $Q = Q_1 \times Q_2 \times \cdots \times Q_k$ denote the product distribution on $\mathcal{Y}$ obtained by sampling independently from each $Q_j$.
        Then $\privloss{P}{Q}$ is the convolution of the distributions $\privloss{P_j}{Q_j}$ for all $j \in [k]$.
        That is, sampling $Z \gets \privloss{P}{Q}$ is equivalent to $Z=\sum_{j=1}^k Z_j$ when $Z_j \gets \privloss{P_j}{Q_j}$ independently for each $j \in [k]$.
    \end{theorem}
    \begin{proof}
        For all $y \in \mathcal{Y}$, the log likelihood ratio (Definition \ref{defn:priv_loss}) satisfies
        \begin{align*}
            \llr{P}{Q}(y) &= \log\left(\frac{P(y)}{Q(y)}\right) \\
            &= \log\left(\frac{P_1(y_1) \cdot P_2(y_2) \cdot \cdots \cdot P_k(y_k)}{Q_1(y_1) \cdot Q_2(y_2) \cdot \cdots \cdot Q_k(y_k)}\right) \\
            &= \log\left(\frac{P_1(y_1)}{Q_1(y_1)}\right) + \log\left(\frac{P_2(y_2)}{Q_2(y_2)}\right) + \cdots + \log\left(\frac{P_k(y_k)}{Q_k(y_k)}\right) \\
            &= \llr{P_1}{Q_1}(y_1) + \llr{P_2}{Q_2}(y_2) + \cdots + \llr{P_k}{Q_k}(y_k).
        \end{align*}
        
        Since $P$ is a product distribution, sampling $Y \gets P$ is equivalent to sampling $Y_1 \gets P_1$, $Y_2 \gets P_2$, $\cdots$, $Y_k \gets P_k$ independently.
        
        A sample from the privacy loss distribution $Z \gets \privloss{P}{Q}$ is given by $Z = \llr{P}{Q}(Y)$ for $Y \gets P$.
        By the above two facts, this is equivalent to $Z = \llr{P_1}{Q_1}(Y_1) + \llr{P_2}{Q_2}(Y_2) + \cdots + \llr{P_k}{Q_k}(Y_k)$ for $Y_1 \gets P_1$, $Y_2 \gets P_2$, $\cdots$, $Y_k \gets P_k$ independently.
        For each $j \in [k]$, sampling $Z_j \gets \privloss{P_j}{Q_j}$ is given by $Z_j = \llr{P_j}{Q_j}(Y_j)$ for $Y_j \gets P_j$. 
        Thus sampling $Z \gets \privloss{P}{Q}$ is equivalent to $Z = Z_1 + Z_2 + \cdots + Z_k$ where $Z_1 \gets \privloss{P_1}{Q_1}$, $Z_2 \gets \privloss{P_2}{Q_2}$, $\cdots$, $Z_k \gets \privloss{P_k}{Q_k}$ are independent.
    \end{proof}
    
    Theorem \ref{thm:privloss_composition} is the key to understanding composition of differential privacy.
    More concretely, we should think of a pair of neighbouring inputs $x,x'$ and $k$ algorithms $M_1, \cdots, M_k$. Suppose $M$ is the composition of $M_1, \cdots, M_k$. Then the the differential privacy of $M$ can be expressed in terms of the privacy loss distribution $\privloss{M(x)}{M(x')}$. Theorem \ref{thm:privloss_composition} allows us to decompose this privacy loss as the sum/convolution of the privacy losses of the constituent algorithms $\privloss{M_j(x)}{M_j(x')}$ for $j \in [k]$. Thus if we have differential privacy guarantees for each $M_j$, this allows us to prove differential privacy guarantees for $M$.
    
    \paragraph{Basic Composition, Revisited:}
    We can revisit basic composition (Theorem \ref{thm:basic_composition}, \S\ref{sec:basic_composition}) with the perspective of privacy loss distributions.
    Suppose $M_1, M_2, \cdots, M_k : \mathcal{X}^n \to \mathcal{Y}$ are each $\varepsilon$-DP. Fix neighbouring datasets $x,x' \in \mathcal{X}^n$.
    This means that $\pr{Z_j \gets \privloss{M_j(x)}{M_j(x')}}{Z_j \le \varepsilon} = 1$ for each $j \in [k]$.
    Now let $M : \mathcal{X}^n \to \mathcal{Y}^k$ be the composition of these algorithms.
    We can express the privacy loss $Z \gets \privloss{M(x)}{M(x')}$ as $Z = Z_1 + Z_2 + \cdots + Z_k$ where $Z_j \gets \privloss{M_j(x)}{M_j(x')}$ for each $j \in [k]$.
    Basic composition simply adds up the upper bounds: \[Z = Z_1 + Z_2 + \cdots + Z_k \le \varepsilon + \varepsilon + \cdots + \varepsilon = k\varepsilon.\]
    This bound is tight if each $Z_j$ is a point mass (i.e., $\pr{}{Z_j=\varepsilon}=1$).
    However, this is not the case. (It is possible to prove, in general, that $\pr{}{Z_j=\varepsilon}\le \frac{1}{1+e^{-\varepsilon}}$.)
    The way we will prove better composition bounds is by applying concentration of measure bounds to this sum of independent random variables.
    That way we can prove that the privacy loss is small with high probability, which yields a better differential privacy guarantee.
    
    Intuitively, we will apply the central limit theorem. The privacy loss random variable of the composed algorithm $M$ can be expressed as the sum of independent bounded random variables. That means the privacy loss distribution $\privloss{M(x)}{M(x')}$ is well-approximated by a Gaussian, which is the information we need to prove a composition theorem. What is left to do is to obtain bounds on the mean and variance of the summands and make this Gaussian approximation precise.
    
    \paragraph{Gaussian Composition:}
    It is instructive to look at composition when each constituent algorithm $M_j$ is the Gaussian noise addition mechanism. In this case the privacy loss distribution is exactly Gaussian and convolutions of Gaussians are also Gaussian. This is the ideal case and our general composition theorem will be an approximation to this ideal.
    
    Specifically, we can prove a multivariate analog of Corollary \ref{cor:gauss_adp_exact}:
    
    \begin{corollary}[Tight Approximate Differential Privacy for Multivariate Gaussian]\label{cor:gauss_adp_exact_multi}
        Let $q : \mathcal{X}^n \to \mathbb{R}^d$ be a deterministic function and let $\Delta : = \sup_{x,x'\in\mathcal{X}^n \atop \text{neighbouring}} \|q(x)-q(x')\|_2$ be its sensitivity in the $2$-norm.
        Define a randomized algorithm $M : \mathcal{X}^n \to \mathbb{R}^d$ by $M(x) = \mathcal{N}(q(x),\sigma^2I)$ for some $\sigma^2>0$, where $I$ is the identity matrix.
        Then, for any $\varepsilon \ge 0$, $M$ satisfies $(\varepsilon,\delta)$-DP with \[\delta = \overline\Phi\left(\frac{\varepsilon-\rho_*}{\sqrt{2\rho_*}}\right) - e^\varepsilon \cdot \overline\Phi\left(\frac{\varepsilon+\rho_*}{\sqrt{2\rho_*}}\right),\] where $\rho_* := \Delta^2/2\sigma^2$ and $\overline\Phi(z) := \pr{G \gets \mathcal{N}(0,1)}{G>z} = \frac{1}{\sqrt{2\pi}} \int_z^\infty \exp(-t^2/2) \mathrm{d}t$.
        
        Furthermore, this guarantee is optimal -- for every $\varepsilon \ge 0$, there is no $\delta'<\delta$ such that $M$ is $(\varepsilon,\delta')$-DP for general $q$.
    \end{corollary}
    \begin{proof}
        Fix arbitrary neighbouring datasets $x,x'\in\mathcal{X}^n$ and $S \subset \mathcal{Y}$. Let $\mu=q(x), \mu'=q(x') \in \mathbb{R}^d$. Let $P=M(x)=\mathcal{N}(\mu,\sigma^2I)$ and $Q=M(x')=\mathcal{N}(\mu',\sigma^2I)$. We must show $P(S) \le e^\varepsilon \cdot Q(S) + \delta$ for arbitrary $\varepsilon \ge 0$ and the value $\delta$ given in the result.
        
        Now both $P$ and $Q$ are product distributions: For $j \in [d]$, let $P_j=\mathcal{N}(\mu_j,\sigma^2)$ and $Q_j=\mathcal{N}(\mu'_j,\sigma^2)$. Then $P=P_1 \times P_2 \times \cdots P_d$ and $Q = Q_1 \times Q_2 \times \cdots \times Q_d$.
        
        By Theorem \ref{thm:privloss_composition}, $\privloss{P}{Q} = \sum_{j=1}^d \privloss{P_j}{Q_j}$ and $\privloss{Q}{P} = \sum_{j=1}^d \privloss{Q_j}{P_j}$.
        
        By Proposition \ref{prop:gauss_privloss}, $\privloss{P_j}{Q_j} = \privloss{Q_j}{P_j} = \mathcal{N}(\rho_j,2\rho_j)$, where $\rho_j = \frac{(\mu_j-\mu'_j)^2}{2\sigma^2}$ for all $j \in [d]$.
        
        Thus $\privloss{P}{Q} = \privloss{Q}{P} = \sum_{j=1}^d \mathcal{N}(\rho_j,2\rho_j) = \mathcal{N}(\rho,2\rho)$, where 
        $\rho = \sum_{j=1}^d \rho_j = \frac{\|\mu-\mu'\|_2^2}{2\sigma^2} \le \rho_* = \frac{\Delta^2}{2\sigma^2}$. 
        
        By Proposition \ref{prop:privloss_adp}, we have $P(S) \le e^\varepsilon \cdot Q(S) + \delta$, where 
        \begin{align*}
            \delta &= \pr{Z \gets \privloss{P}{Q}}{Z>\varepsilon} - e^\varepsilon \cdot \pr{Z' \gets \privloss{Q}{P}}{-Z'>\varepsilon} \\
            &= \pr{Z \gets \mathcal{N}(\rho,2\rho)}{Z>\varepsilon} - e^\varepsilon \cdot \pr{Z' \gets \mathcal{N}(\rho,2\rho)}{-Z'>\varepsilon} \\
            &= \pr{G \gets \mathcal{N}(0,1)}{\rho + \sqrt{2\rho}\cdot G >\varepsilon} - e^\varepsilon \cdot \pr{G \gets \mathcal{N}(0,1)}{-\rho + \sqrt{2\rho}\cdot G >\varepsilon} \\
            &= \pr{G \gets \mathcal{N}(0,1)}{G > \frac{\varepsilon-\rho}{\sqrt{2\rho}}} - e^\varepsilon \cdot \pr{G \gets \mathcal{N}(0,1)}{G >\frac{\varepsilon+\rho}{\sqrt{2\rho}}} \\
            &= \overline\Phi\left(\frac{\varepsilon-\rho}{\sqrt{2\rho}}\right) - e^\varepsilon \cdot \overline\Phi\left(\frac{\varepsilon+\rho}{\sqrt{2\rho}}\right).
        \end{align*}
        Since $\rho \le \rho_*$ and the above expression is increasing in $\rho$, we can substitute in $\rho_*$ as an upper bound.
        
        Optimality follows from the fact that Propositions \ref{prop:gauss_privloss} and \ref{prop:privloss_adp} and Theorem \ref{thm:privloss_composition} give exact characterizations. Note that we must assume that there exist neighbouring $x,x'$ such that $\rho = \rho_*$.
    \end{proof}
    
    The key to the analysis of Gaussian composition in the proof of Corollary \ref{cor:gauss_adp_exact_multi} is that sums of Gaussians are Gaussian.
    In general, the privacy loss of each component is not Gaussian, but the sum still behaves much like a Gaussian and this observation is the basis for improving the composition analysis. 

    \paragraph{Composition via Gaussian Approximation:}    
    After analyzing Gaussian composition, our next step is to analyze the composition of $k$ independent $\varepsilon$-DP algorithms.
    We will use the same tools as we did for Gaussian composition and we will develop a new tool, which is called concentrated differential privacy.
    
    Let $M_1, \cdots, M_k : \mathcal{X}^n \to \mathcal{Y}$ each be $\varepsilon$-DP and let $M : \mathcal{X}^n \to \mathcal{Y}^k$ be the composition of these algorithms. Let $x,x'\in\mathcal{X}^n$ be neighbouring datasets. For notational convenience, let $P_j = M_j(x)$ and $Q_j=M_j(x')$ for all $j \in [k]$ and let $P=M(x) = P_1 \times P_2 \times \cdots \times P_k$ and $Q=M(x') = Q_1 \times Q_2 \times \cdots \times Q_k$.
    
    For each $j \in [k]$, the algorithm $M_j$ satisfies $\varepsilon$-DP, which ensures that the privacy loss random variable $Z_j \gets \privloss{P_j}{Q_j} = \privloss{M_j(x)}{M_j(x')}$ is supported on the interval $[-\varepsilon,\varepsilon]$. The privacy loss being bounded immediately implies a bound on the variance: $\var{}{Z_j} \le \ex{}{Z_j^2} \le \varepsilon^2$.
    We also can prove a bound on the expectation: $\ex{}{Z_j} \le \frac12 \varepsilon^2$. We will prove this bound formally later (in Proposition \ref{prop:pdp2cdp}). For now, we give some intuition: Clearly $\ex{}{Z_j} \le \varepsilon$ and the only way this can be tight is if $Z_j = \varepsilon$ with probability $1$. But $Z_j=\log(P_j(Y_j)/Q_j(Y_j))$ for $Y_j \gets P_j$. Thus $\ex{}{Z_j}=\varepsilon$ implies $P_j(Y_j) = e^\varepsilon \cdot Q_j(Y_j)$ with probability $1$. This yields a contradiction: $1 = \sum_y P_j(y) = \sum_y e^\varepsilon \cdot Q_j(y) = e^\varepsilon \cdot 1$. Thus we conclude $\ex{}{Z_j} < \varepsilon$ and, with a bit more work, we can obtain the bound $\ex{}{Z_j} \le \frac12 \varepsilon^2$ from the fact that $|Z_j|\le\varepsilon$ and $\sum_y P_j(y) = \sum_y Q_j(y) = 1$. 
    
    Our goal is to understand the privacy loss $Z \gets \privloss{P}{Q}=\privloss{M(x)}{M(x')}$ of the composed algorithm. Theorem \ref{thm:privloss_composition} tells us that this is the convolution of the constituent privacy losses. That is, we can write $Z = \sum_{j=1}^k Z_j$ where $Z_j \gets \privloss{P_j}{Q_j} = \privloss{M_j(x)}{M_j(x')}$ independently for each $j \in [k]$. 
    
    By independence, we have \[\ex{}{Z} = \sum_{j=1}^k \ex{}{Z_j} \le \frac12 \varepsilon^2 \cdot k ~~~\text{ and }~~~ \var{}{Z} = \sum_{j=1}^k \var{}{Z_j} \le \varepsilon^2 \cdot k.\] Since $Z$ can be written as the sum of independent bounded random variables, the central limit theorem tells us that it is well approximated by a Gaussian -- i.e., \[\privloss{P}{Q} = \privloss{M(x)}{M(x')} \approx \mathcal{N}(\ex{}{Z},\var{}{Z}).\]
    
    Are we done? Can we substitute this approximation into Proposition \ref{prop:privloss_adp} to complete the proof of a better composition theorem?
    We must make this approximation precise. Unfortunately, the approximation guarantee of the quantitative central limit theorem (a.k.a., the Berry-Esseen Theorem) is not quite strong enough. To be precise, converting the guarantee to approximate $(\varepsilon,\delta)$-DP would incur an error of $\delta \ge \Omega(1/\sqrt{k})$, which is larger than we want.
    
    Our approach is to look at the moment generating function -- i.e., the expectation of an exponential function -- of the privacy loss distribution. To be precise, we will show that, for all $t \ge 0$,
    \begin{align*}
        \ex{Z \gets \privloss{P}{Q}}{\exp(tZ)} &= \prod_{j=1}^k \ex{Z_j \gets \privloss{P_j}{Q_j}}{\exp(tZ_j)} \\
        &\le \exp\left(\frac12\varepsilon^2 t (t+1) \cdot k \right) \\
        &= \ex{\tilde{Z} \gets \mathcal{N}(\frac12 \varepsilon^2 k , \varepsilon^2 k)}{\exp(t\tilde{Z})}.
    \end{align*}
    In other words, rather than attempting to prove a Gaussian approximation, we prove a one-sided bound. Informally, this says that $\privloss{P}{Q} \le  \mathcal{N}(\frac12 \varepsilon^2 k , \varepsilon^2 k)$. The expectation of an exponential function turns out to be a nice way to formalize this inequality, because, if $X$ and $Y$ are independent, then $\ex{}{\exp(X+Y)}=\ex{}{\exp(X)}\cdot\ex{}{\exp(Y)}$.
    
    To formalize this approach, we next introduce concentrated differential privacy.
    
    \subsection{Concentrated Differential Privacy}
    
    Concentrated differential privacy \cite{dwork2016concentrated,bun2016concentrated} is a variant of differential privacy (like pure DP and approximate DP). The main advantage of concentrated DP is that it composes well. Thus we will use it as a tool to prove better composition results.

    \begin{definition}[Concentrated Differential Privacy]\label{defn:cdp}
        Let $M : \mathcal{X}^n \to \mathcal{Y}$ be a randomized algorithm. We say that $M$ satisfies $\rho$-concentrated differential privacy ($\rho$-zCDP) if, for all neighbouring inputs $x,x'\in\mathcal{X}^n$, the privacy loss distribution $\privloss{M(x)}{M(x')}$ is well-defined (see Definition \ref{defn:priv_loss}) and \[\forall t \ge 0 ~~~~~ \ex{Z \gets \privloss{M(x)}{M(x')}}{\exp(tZ)} \le \exp(t(t+1)\cdot\rho).\] 
    \end{definition}
    
    To contextualize this definition, we begin by showing that the Gaussian mechanism satisfies it.
    
    \begin{lemma}[Gaussian Mechanism is Concentrated DP]\label{lem:gauss_cdp}
        Let $q : \mathcal{X}^n \to \mathbb{R}^d$ have sensitivity $\Delta$ -- that is, $\|q(x)-q(x')\|_2 \le \Delta$ for all neighbouring $x,x'\in\mathcal{X}^n$. Let $\sigma>0$.
        Define a randomized algorithm $M : \mathcal{X}^n \to \mathbb{R}^d$ by $M(x) = \mathcal{N}(q(x),\sigma^2 I_d)$. Then $M$ is $\rho$-zCDP for $\rho = \frac{\Delta^2}{2\sigma^2}$.
    \end{lemma}
    \begin{proof}
        Fix neighbouring inputs $x,x'\in\mathcal{X}^n$ and $t \ge 0$.
        By Proposition \ref{prop:gauss_privloss}, for each $j \in [d]$,\\$\privloss{M(x)_j}{M(x')_j} = \mathcal{N}(\hat\rho_j,2\hat\rho_j)$ for $\hat\rho_j = \frac{(q(x)_j-q(x')_j)^2}{2\sigma^2}$.
        By Theorem \ref{thm:privloss_composition},\\$\privloss{M(x)}{M(x')} = \sum_{j=1}^d \mathcal{N}(\hat\rho_j,2\hat\rho_j) = \mathcal{N}(\hat\rho,2\hat\rho)$ for $\hat\rho = \sum_{j=1}^d \hat\rho_j = \frac{\|q(x)-q(x')\|_2^2}{2\sigma^2} \le \rho$. Thus $\ex{Z \gets \privloss{M(x)}{M(x')}}{\exp(tZ)} = \exp(t(t+1)\hat\rho) \le \exp(t(t+1)\rho)$, as required.
    \end{proof}
    
    To analyze the composition of $k$ independent $\varepsilon$-DP algorithms, we will prove three results: (i) Pure $\varepsilon$-DP implies $\frac12\varepsilon^2$-zCDP. (ii) The composition of $k$ independent $\frac12\varepsilon^2$-zCDP algorithms satisfies $\frac12\varepsilon^2k$-zCDP. (iii) $\frac12\varepsilon^2k$-zCDP implies approximate $(\varepsilon',\delta)$-DP with $\delta \in (0,1)$ arbitrary and $\varepsilon' =\varepsilon \cdot \sqrt{2k\log(1/\delta)} + \frac12\varepsilon^2 k$.
    We begin with composition, as this is the raison d'\^etre for concentrated DP:
    
    \begin{theorem}[Composition for Concentrated Differential Privacy]\label{thm:cdp_composition}
        Let $M_1, M_2, \cdots, M_k : \mathcal{X}^n \to \mathcal{Y}$ be randomized algorithms. Suppose $M_j$ is $\rho_j$-zCDP for each $j \in [k]$.
        Define $M : \mathcal{X}^n \to \mathcal{Y}^k$ by $M(x)=(M_1(x),M_2(x),\cdots,M_k(x))$, where each algorithm is run independently. Then $M$ is $\rho$-zCDP for $\rho = \sum_{j=1}^k \rho_j$.  
    \end{theorem}
    \begin{proof}
        Fix neighbouring inputs $x,x'\in\mathcal{X}^n$.
        By our assumption that each algorithm $M_j$ is $\rho_j$-zCDP, \[\forall t \ge 0 ~~~~~ \ex{Z_j \gets \privloss{M_j(x)}{M_j(x')}}{\exp(t Z_j)} \le \exp(t(t+1)\cdot\rho_j).\]
        By Theorem \ref{thm:privloss_composition}, $Z \gets \privloss{M(x)}{M(x')}$ can be written as $Z=\sum_{j=1}^k Z_j$, where $Z_j \gets \privloss{M_j(x)}{M_j(x')}$ independently for each $j \in [k]$.
        
        Thus, for any $t \ge 0$, we have
        \begin{align*}
            \ex{Z \gets \privloss{M(x)}{M(x')}}{\exp(t Z)} &= \ex{\forall j \in [k] ~~ Z_j \gets \privloss{M_j(x)}{M_j(x')} \atop \text{independent}}{\exp\left(t \sum_{j=1}^k Z_j \right)} \\
            &= \prod_{j=1}^k \ex{Z_j \gets \privloss{M_j(x)}{M_j(x')}}{\exp(t Z_j)} \\
            &\le \prod_{j=1}^k \exp(t(t+1) \cdot \rho_j) \\
            &= \exp\left(t(t+1) \cdot \sum_{j=1}^k \rho_j\right) \\
            &= \exp(t(t+1) \cdot \rho).
        \end{align*}
        Since $x$ and $x'$ were arbitrary, this proves that $M$ satisfies $\rho$-zCDP, as required.
    \end{proof}
    
    Next we show how to convert from concentrated DP to approximate DP, which applies the tools we developed earlier. (This conversion is fairly tight, but not completely optimal; Asoodeh, Liao, Calmon, Kosut, and Sankar \cite{asoodeh2020better} give an optimal conversion.)
    
    \begin{proposition}[Conversion from Concentrated DP to Approximate DP]\label{prop:cdp2adp}
        For any $M : \mathcal{X}^n \to \mathcal{Y}$ and any $\varepsilon,t \ge 0$, $M$ satisfies $(\varepsilon,\delta)$-DP with
        \begin{align*}
            \delta &= \sup_{x,x'\in\mathcal{X}^n \atop \text{neighbouring}} \ex{Z \gets \privloss{M(x)}{M(x')}}{\exp(tZ)} \cdot \frac{\exp(-\varepsilon t)}{t+1} \cdot \left( 1 - \frac{1}{t+1} \right)^t\\
            &\le \sup_{x,x'\in\mathcal{X}^n \atop \text{neighbouring}} \ex{Z \gets \privloss{M(x)}{M(x')}}{\exp(t(Z-\varepsilon))}.
        \end{align*}
        In particular, if $M$ satisfies $\rho$-zCDP, then $M$ satisfies $(\varepsilon,\delta)$-DP for any $\varepsilon \ge \rho$ with
        \begin{align*}
            \delta &= \inf_{t > 0} ~\exp(t(t+1)\rho-\varepsilon t) \cdot \frac{1}{t+1} \cdot \left( 1 - \frac{1}{t+1} \right)^t \\
            &\le \exp(-(\varepsilon-\rho)^2/4\rho).
        \end{align*} 
    \end{proposition}
    \begin{proof}
        Fix arbitrary neighbouring inputs $x,x'$. 
        Fix $\varepsilon, t \ge 0$. 
        We must show that for all $S$ we have $\pr{}{M(x) \in S} \le e^\varepsilon \cdot \pr{}{M(x')\in S} + \delta$ for the value of $\delta$ given in the statement above.
        
        Let $Z \gets \privloss{M(x)}{M(x')}$. 
        By Proposition \ref{prop:privloss_adp}, it suffices to show \[\ex{}{\max\{0,1-\exp(\varepsilon-Z)\}} \le \delta\] for the value of $\delta$ given in the statement above.
        
        Let $c>0$ be a constant such that, with probability 1, \[\max\{0,1-\exp(\varepsilon-Z)\} \le c \cdot \exp(t Z).\]
        Taking expectations of both sides we have $\ex{}{\max\{0,1-\exp(\varepsilon-Z)\}} \le c \cdot \ex{}{\exp(tZ)}$, which is the kind of bound we need. It only remains to identify the appropriate value of $c$ to obtain the desired bound.
        
        We trivially have $0 \le c \cdot \exp(t Z)$ as long as $c > 0$. Thus we only need to ensure $1-\exp(\varepsilon-Z) \le c \cdot \exp(t Z)$. That is, for any value of $t>0$, we can set
        \begin{align*}
            c &= \sup_{z \in \mathbb{R}} \frac{1-\exp(\varepsilon-z)}{\exp(tz)} \\
            &= \sup_{z \in \mathbb{R}} \exp(-tz) - \exp(\varepsilon-(t+1)z) \\
            &= \frac{\exp(-\varepsilon t)}{t+1} \cdot \left( 1- \frac{1}{t+1} \right)^t,
        \end{align*}
        where the final equality follows from using calculus to determine that $z = \varepsilon + \log(1+1/t)$ is the optimal value of $z$.
        Thus $\ex{}{\max\{0,1-\exp(\varepsilon-Z)\}} \le \ex{}{\exp(tZ)} \cdot \frac{\exp(-\varepsilon t)}{t+1} \cdot \left( 1- \frac{1}{t+1} \right)^t$, which proves the first part of the statement.
        
        Now assume $M$ is $\rho$-zCDP. Thus \[\forall t \ge 0 ~~~~~ \ex{}{\exp(t Z)} \le \exp(t (t+1) \cdot \rho),\] which immediately yields the equality in the second part of the statement. 
        
        To obtain the inequality in the second part of the statement, we observe that \[\max\{0,1-\exp(\varepsilon-Z)\} \le \mathbb{I}[Z>\varepsilon] \le \exp(t (Z-\varepsilon)),\] whence $c \le \exp(-\varepsilon t)$. Substituting in this upper bound on $c$ and setting $t=(\varepsilon-\rho)/2\rho$ completes the proof 
    \end{proof}
    
    \begin{remark}\label{rem:rho}
        Proposition \ref{prop:cdp2adp} shows that $\rho$-zCDP implies $(\varepsilon, \delta=\exp(-(\varepsilon-\rho)^2/4\rho))$-DP for all $\varepsilon \ge \rho$. Equivalently, $\rho$-zCDP implies $(\varepsilon = \rho + 2\sqrt{\rho \cdot \log(1/\delta)}, \delta)$-DP for all $\delta >0$. Also, to obtain a given a target $(\varepsilon,\delta)$-DP guarantee, it suffices to have $\rho$-zCDP with \[ \frac{\varepsilon^2}{4\log(1/\delta) + 4\varepsilon} \le \rho = \left( \sqrt{\log(1/\delta) + \varepsilon} - \sqrt{\log(1/\delta)} \right)^2 \le \frac{\varepsilon^2}{4\log(1/\delta)}.\] This gives a sufficient condition; tighter bounds can be obtained from Proposition \ref{prop:cdp2adp}.
        For example, if we add $\mathcal{N}(0,\sigma^2)$ to a query of sensitivity 1, then, by Lemma \ref{lem:gauss_cdp}, to ensure $(\varepsilon,\delta)$-DP it suffices to set $\sigma^2 = \frac{2}{\varepsilon^2}\cdot\left( \log(1/\delta) + \varepsilon \right)$.
    \end{remark}

    The final piece of the puzzle is the conversion from pure DP to concentrated DP.
    
    \begin{proposition}\label{prop:pdp2cdp}
        Suppose $M$ satisfies $\varepsilon$-DP, then $M$ satisfies $\frac12 \varepsilon^2$-zCDP.
    \end{proposition}
    \begin{proof}
        Fix neighbouring inputs $x,x'$.
        Let $Z \gets \privloss{M(x)}{M(x')}$.
        By our $\varepsilon$-DP assumption, $Z$ is supported on the interval $[-\varepsilon,+\varepsilon]$.
        Our task is to prove that $\ex{}{\exp(t Z)} \le \exp(\frac12 \varepsilon^2 t (t+1))$ for all $t>0$.
        
        The key additional fact is the following consequence of Lemma \ref{lem:dual_privloss} %
        \[\ex{Z \gets \privloss{P}{Q}}{e^{-Z}} = \ex{Y \gets P}{e^{-\llr{P}{Q}(Y)}} = \ex{Y \gets Q}{e^{\llr{P}{Q}(Y)} \cdot e^{-\llr{P}{Q}(Y)}} = \ex{Y \gets Q}{1} = 1.\]
        We can write this out as an integral to make it clear:
        \begin{align*}
            \ex{Z \gets \privloss{P}{Q}}{\exp(-Z)} &= \ex{Y \gets P}{\exp(-\llr{P}{Q}(Y))} \\
            &= \ex{Y \gets P}{\exp(-\log(P(Y)/Q(Y)))} \\
            &= \ex{Y \gets P}{\frac{Q(Y)}{P(Y)}} \\
            &= \int_{\mathcal{Y}} \frac{Q(y)}{P(y)} P(y) \mathrm{d}y \\
            &= \int_{\mathcal{Y}} Q(y) \mathrm{d}y \\
            &= 1.
        \end{align*}
        The combination of these two facts -- $Z \in [-\varepsilon,\varepsilon]$ and $\ex{}{\exp(-Z)}=1$ -- is all we need to know about $Z$ to prove the result.
        The technical ingredient is Hoeffding's lemma \cite{hoeffding58probability}:
        
        \begin{lemma}[Hoeffding's lemma]\label{lem:hoeffding}
            Let $Z$ be a random variable supported on the interval $[-\varepsilon,+\varepsilon]$. Then for all $t \in \mathbb{R}$, $\ex{}{\exp(t Z)} \le \exp(t \ex{}{Z} + t^2 \varepsilon^2 / 2)$.
        \end{lemma}
        \begin{proof}
            To simplify things, we can assume without loss of generality that $Z$ is supported on the discrete set $\{-\varepsilon,+\varepsilon\}$.
            To prove this claim, let $\tilde{Z} \in \{-\varepsilon,+\varepsilon\}$ be a randomized rounding of $Z$. That is, $\exc{\tilde{Z}}{\tilde{Z}}{Z=z}=z$ for all $z \in [-\varepsilon,+\varepsilon]$.
            By Jensen's inequality, since $\exp(tz)$ is a convex function of $z \in \mathbb{R}$ for any fixed $t \in \mathbb{R}$, we have \[\ex{Z}{\exp(tZ)} = \ex{Z}{\exp\left(t\exc{\tilde{Z}}{\tilde{Z}}{Z}\right)} \le \ex{Z}{\exc{\tilde{Z}}{\exp(t \tilde{Z})}{Z}} = \ex{\tilde{Z}}{\exp(t \tilde{Z})}.\]
            Note that $\ex{}{\tilde{Z}} = \ex{}{Z}$.
            Thus it suffices to prove $\ex{}{\exp(t \tilde{Z})} \le \exp(t \ex{}{\tilde{Z}} + \frac12 \varepsilon^2 t^2)$ for all $t \in \mathbb{R}$.

            The final step in the proof is some calculus: Let $p:=\pr{}{\tilde{Z}=\varepsilon}=1-\pr{}{\tilde{Z}=-\varepsilon}$. Then $\ex{}{Z} = \ex{}{\tilde{Z}} = \varepsilon p - \varepsilon (1-p)= \varepsilon (2p-1)$.
            Define $f : \mathbb{R} \to \mathbb{R}$ by \[f(t) := \log \ex{}{\exp(t\tilde{Z})} = \log(p\cdot e^{t\varepsilon} + (1-p) \cdot e^{-t\varepsilon}) = \log(1-p+p\cdot e^{2t\varepsilon})-t\varepsilon.\]
            For all $t \in \mathbb{R}$, \[f'(t) = \frac{2\varepsilon p \cdot e^{2t\varepsilon}}{1-p+p\cdot e^{2t\varepsilon}}-\varepsilon\]
            and 
            \begin{align*}
                f''(t) &= \frac{(2\varepsilon)^2 p \cdot e^{2t\varepsilon} \cdot (1-p+p\cdot e^{2t\varepsilon}) - (2\varepsilon p \cdot e^{2t\varepsilon})^2}{(1-p+p\cdot e^{2t\varepsilon})^2} \\
                &= (2\varepsilon)^2 \cdot \frac{p \cdot e^{2t\varepsilon}}{1-p+p\cdot e^{2t\varepsilon}} \cdot \left( 1 - \frac{ p \cdot e^{2t\varepsilon}}{1-p+p\cdot e^{2t\varepsilon}} \right) \\
                &= (2\varepsilon)^2  \cdot x \cdot (1-x) \le (2\varepsilon)^2 \cdot \frac14 =\varepsilon^2.
            \end{align*}
            The final line sets $x=\frac{p \cdot e^{2t\varepsilon}}{1-p+p\cdot e^{2t\varepsilon}}$ and uses the fact that the function $x \cdot (1-x)$ is maximized at $x=\frac12.$
            
            Note that $f(0)=0$ and $f'(0)=2\varepsilon p - \varepsilon = \ex{}{\tilde{Z}} = \ex{}{Z}$.
            By the fundamental theorem of calculus, for all $t \in \mathbb{R}$,
            \[f(t) = f(0) + f'(0) \cdot t + \int_0^t \int_0^s f''(r) \mathrm{d}r \mathrm{d}s \le 0 + \ex{}{Z} \cdot t + \int_0^t \int_0^s \varepsilon^2 \mathrm{d}r \mathrm{d}s = \ex{}{Z} \cdot t + \frac12 \varepsilon^2 t^2.\]
            This proves the lemma, as $\ex{}{\exp(t Z)} \le \ex{}{\exp(t \tilde{Z})} = \exp(f(t)) \le \exp( \ex{}{Z} \cdot t + \frac12 \varepsilon^2 t^2 )$.
        \end{proof}
        
        If we substitute $t=-1$ into Lemma \ref{lem:hoeffding}, we have 
        \[1 = \ex{}{\exp(-Z)} \le \exp(-\ex{}{Z} + \frac12 \varepsilon^2),\]
        which rearranges to $\ex{}{Z} \le \frac12\varepsilon^2$.
        
        Substituting this bound on the expectation back into Lemma \ref{lem:hoeffding} yields the result: For all $t>0$, we have
        \[\ex{}{\exp(t Z)} \le \exp\left( t \cdot \ex{}{Z} + \frac12 \varepsilon^2 t^2 \right) \le \exp\left( \frac12 \varepsilon^2 t (t+1)\right).\]
    \end{proof}
    
    \begin{figure}
        \centering
        \includegraphics[width=0.75\textwidth]{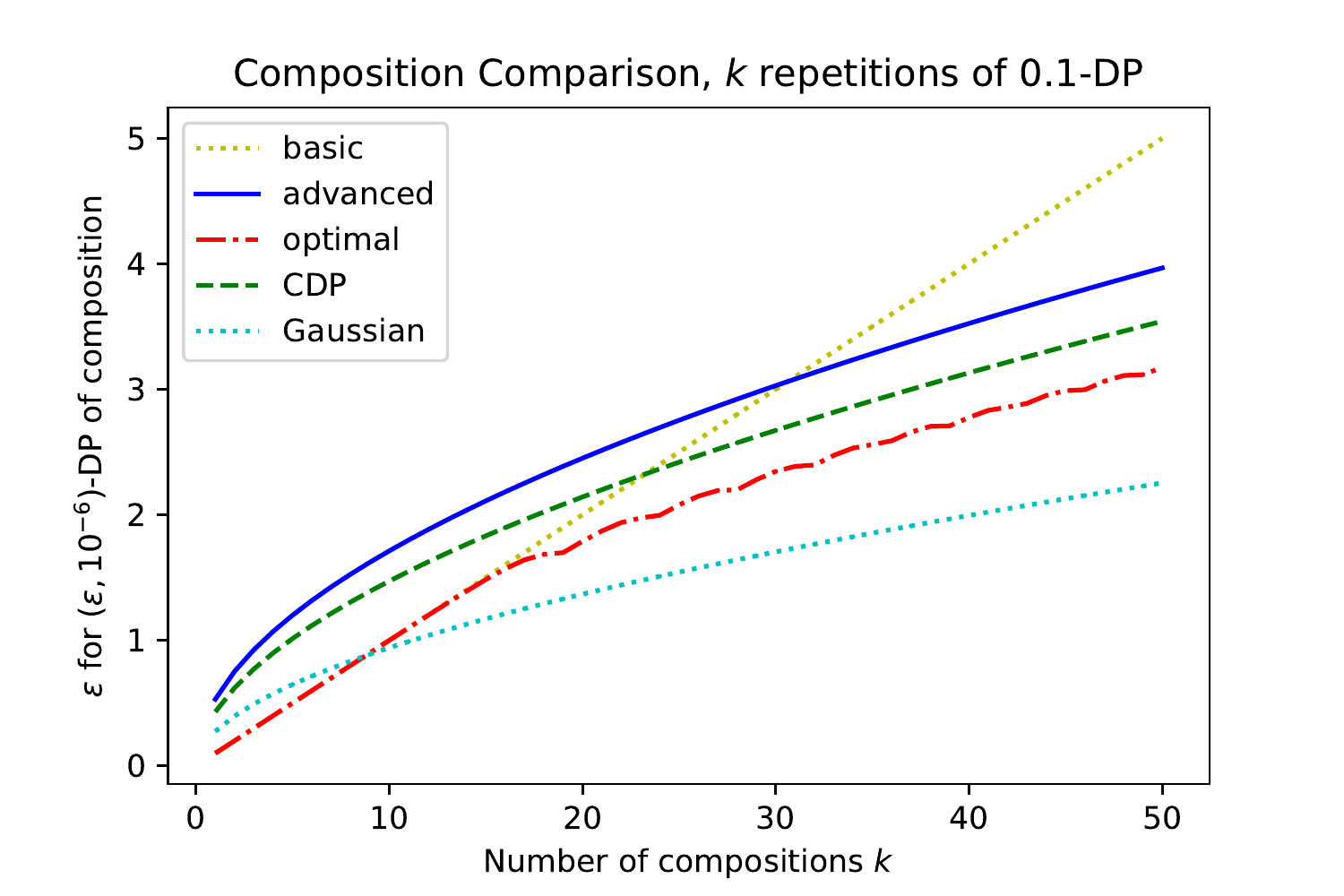}
        \caption{Comparison of different composition bounds. We compose $k$ independent $0.1$-DP algorithms to obtain a $(\varepsilon,10^{-6})$-DP guarantee. Theorem \ref{thm:basic_composition} -- \texttt{basic} composition -- gives $\varepsilon=k\cdot 0.1$. For comparison, we have \texttt{advanced} composition (Theorem \ref{thm:advancedcomposition_pure}), an \texttt{optimal} bound \cite{kairouz2015composition}, and Concentrated DP (\texttt{CDP}) with the improved conversion from Proposition \ref{prop:cdp2adp}. For comparison, we also consider composing the \texttt{Gaussian} mechanism using Corollary \ref{cor:gauss_adp_exact_multi}, where the Gaussian noise is scaled to have the same variance as Laplace noise would have to attain $0.1$-DP.}
        \label{fig:composition_comparison}
    \end{figure}
    
    Combining these three results lets us prove what is known as the advanced composition theorem where we start with each individual algorithm satisfying pure DP \cite{dwork2010boosting}:
    
    \begin{theorem}[Advanced Composition Starting with Pure DP]\label{thm:advancedcomposition_pure}
        Let $M_1, M_2, \cdots, M_k : \mathcal{X}^n \to \mathcal{Y}$ be randomized algorithms. Suppose $M_j$ is $\varepsilon_j$-DP for each $j \in [k]$.
        Define $M : \mathcal{X}^n \to \mathcal{Y}^k$ by $M(x)=(M_1(x),M_2(x),\cdots,M_k(x))$, where each algorithm is run independently. Then $M$ is $(\varepsilon,\delta)$-DP for any $\delta>0$ with \[\varepsilon = \frac12 \sum_{j=1}^k \varepsilon_j^2 + \sqrt{2\log(1/\delta) \sum_{j=1}^k \varepsilon_j^2}.\]
    \end{theorem}
    \begin{proof}[Proof of Theorem \ref{thm:advancedcomposition_pure}.]
        By Proposition \ref{prop:pdp2cdp}, for each $j \in [k]$, $M_j$ satisfies $\rho_j$-zCDP with $\rho_j = \frac12 \varepsilon_j^2$.
        By composition of concentrated DP (Theorem \ref{thm:cdp_composition}), $M$ satisfies $\rho$-zCDP with $\rho=\sum_{j=1}^k \rho_j$.
        Finally, Proposition \ref{prop:cdp2adp} can convert this concentrated DP guarantee to approximate DP: $M$ satisfies $(\varepsilon,\delta)$-DP for all $\varepsilon \ge \rho$ and $\delta = \exp(-(\varepsilon-\rho)^2/4\rho)$. We can rearrange this so that $\delta>0$ is arbitrary and $\varepsilon = \rho + \sqrt{4\rho \log(1/\delta)}$.
    \end{proof}

    Recall that the basic composition theorem (Theorem \ref{thm:basic_composition}) gives $\delta=0$ and $\varepsilon = \sum_{j=1}^k \varepsilon_j$. That is, basic composition scales with the 1-norm of the vector $(\varepsilon_1, \varepsilon_2, \cdots, \varepsilon_k)$, whereas advanced composition scales with the 2-norm of this vector (and the squared 2-norm).
    Neither bound strictly dominates the other. However, asymptotically (in a sense we will make precise in the next paragraph) advanced composition dominates basic composition.
    
    Suppose we have a fixed $(\varepsilon,\delta)$-DP guarantee for the entire system and we must answer $k$ queries of sensitivity $1$.
    Using basic composition, we can answer each query by adding $\mathsf{Laplace}(k/\varepsilon)$ noise to each answer.
    However, using advanced composition, we can answer each query by adding $\mathsf{Laplace}(\sqrt{k/2\rho})$ noise to each answer, where \[\rho %
    \ge \frac{\varepsilon^2}{4\log(1/\delta)+4\varepsilon}\] (per Remark \ref{rem:rho}).
    If the privacy parameters $\varepsilon,\delta>0$ are fixed (which implies $\rho$ is fixed) and $k \to \infty$, we can see that asymptotically advanced composition gives noise per query scaling as $\Theta(\sqrt{k})$, while basic composition results in noise scaling as $\Theta(k)$.

    \subsection{Adaptive Composition \& Postprocessing}
    
    Thus far we have only considered non-adaptive composition. That is, we assume that the algorithms $M_1,M_2,\cdots,M_k$ being composed are independent. More generally, adaptive composition considers the possibility that $M_j$ can depend on the outputs of $M_1, \cdots, M_{j-1}$. This kind of dependence arises very often, either in an iterative algorithm, or an interactive system where a human chooses analyses to perform sequentially. Fortunately, adaptive composition is easy to deal with.
    
    \begin{proposition}[Adaptive Composition of Concentrated DP]\label{prop:adaptivecomp}
        Let $M_1 : \mathcal{X}^n \to \mathcal{Y}_1$ be $\rho_1$-zCDP. Let $M_2 : \mathcal{X}^n \times \mathcal{Y}_1 \to \mathcal{Y}_2$ be such that, for all $y_1 \in \mathcal{Y}_1$, the algorithm $x \mapsto M(x,y_1)$ is $\rho_2$-zCDP. That is, $M_2$ is $\rho_2$-zCDP in terms of its first argument for any fixed value of the second argument.
        Define $M : \mathcal{X}^n \to \mathcal{Y}_2$ by $M(x) = M_2(x,M_1(x))$. Then $M$ is $(\rho_1+\rho_2)$-zCDP.
    \end{proposition}
    Proposition \ref{prop:adaptivecomp} only considers the composition of two algorithms, but it can be extended to $k$ algorithms by induction.
    \begin{proof}
        Fix neighbouring inputs $x,x' \in \mathcal{X}^n$. Fix $t \ge 0$. 
        Let $Z \gets \privloss{M(x)}{M(x')}$. We must prove $\ex{}{\exp(tZ)} \le \exp(t(t+1)(\rho_1+\rho_2))$.
        
        For non-adaptive composition, we could write $Z = Z_1 + Z_2$ where $Z_1 \gets \privloss{M_1(x)}{M_1(x')}$ and $Z_2 \gets \privloss{M_2(x)}{M_2(x')}$ are independent. However, we cannot do this in the adaptive case -- the two privacy losses are not independent. Instead, we use the fact that, conditioned on the value of the first privacy loss $Z_1$, the privacy loss $Z_2$ still satisfies the bound on the moment generating function. That is, for all $z_1$, we have $\ex{}{\exp(tZ_2) \mid Z_1=z_1} \le \exp(t(t+1)\rho_2)$. To make this argument precise, we must expand out the relevant definitions.

        For now, we make a simplifying technical assumption (which we will justify later): We assume that, given $y_2=M(x,y_1)$, we can determine $y_1$. This means we can decompose $\llr{M(x)}{M(x')}(y_2) = \llr{M_1(x)}{M_1(x')}(y_1) + \llr{M_2(x,y_1)}{M_2(x',y_1)}(y_2)$. Thus
        \begin{align*}
            &\ex{Z \gets \privloss{M(x)}{M(x')}}{\exp(tZ)} \\
            &~= \ex{Y \gets M_2(x,M_1(x))}{\exp\left(t \cdot \llr{M(x)}{M(x')}(Y)\right)} \\
            &~= \ex{Y_1 \gets M_1(x)}{\ex{Y_2 \gets M_2(x,Y_1)}{\exp\left(t \cdot \left( \llr{M_1(x)}{M_1(x')}(Y_1) + \llr{M_2(x,Y_1)}{M_2(x',Y_1)}(Y_2) \right) \right)}}\\
            &~= \ex{Y_1 \gets M_1(x)}{\exp\left(t \cdot \llr{M_1(x)}{M_1(x')}(Y_1)\right) \cdot \ex{Y_2 \gets M_2(x,Y_1)}{\exp\left(t \cdot  \llr{M_2(x,Y_1)}{M_2(x',Y_1)}(Y_2) \right)}}\\
            &~\le \ex{Y_1 \gets M_1(x)}{\exp\left(t \cdot \llr{M_1(x)}{M_1(x')}(Y_1)\right)} \cdot \sup_{y_1} \ex{Y_2 \gets M_2(x,y_1)}{\exp\left(t \cdot  \llr{M_2(x,y_1)}{M_2(x',y_1)}(Y_2) \right)}\\
            &~= \ex{Z_1 \gets \privloss{M_1(x)}{M_1(x')}}{\exp\left(t \cdot Z_1 \right)} \cdot \sup_{y_1} \ex{Z_2 \gets \privloss{M_2(x,y_1)}{M_2(x',y_2)}}{\exp\left(t \cdot Z_2 \right)}\\
            &~\le \exp(t(t+1)\rho_1) \cdot \exp(t(t+1)\rho_2) \\
            &~= \exp(t(t+1)(\rho_1+\rho_2)),
        \end{align*}
        as required. All that remains is to justify our simplifying technical assumption. We can perforce ensure this assumption holds by defining $\hat{M} : \mathcal{X}^n \to \mathcal{Y}_1 \times \mathcal{Y}_2$ by $\hat{M}(x) = (y_1,y_2)$ where $y_1=M_1(x)$ and $y_2=M_2(x,y_1)$ and proving the theorem for $\hat{M}$ in lieu of $M$. Since the output of $\hat{M}$ includes both outputs, rather than just the last output, the above decomposition works. The result holds in general because $M$ is a \emph{postprocessing} of $\hat{M}$. That is, we can obtain $M(x)$ by running $\hat{M}(x)$ and discarding the first part of the output. Intuitively, discarding part of the output cannot hurt privacy. Formally, this is the postprocessing property of concentrated DP, which we prove in Lemma \ref{lem:postprocessing} and Corollary \ref{cor:postprocesing}.
    \end{proof}
    \begin{lemma}[Postprocessing for Concentrated DP]\label{lem:postprocessing}
        Let $\hat{P}$ and $\hat{Q}$ be distributions on $\hat{\mathcal{Y}}$ and let $g : \hat{\mathcal{Y}} \to \mathcal{Y}$ be an arbitrary function. Define $P=g(\hat{P})$ and $Q=g(\hat{Q})$ to be the distributions on $\mathcal{Y}$ obtained by applying $g$ to a function from $\hat{P}$ and $\hat{Q}$ respectively.
        Then, for all $t \ge 0$, \[\ex{Z \gets \privloss{P}{Q}}{\exp(tZ)} \le \ex{\hat{Z} \gets \privloss{\hat{P}}{\hat{Q}}}{\exp(t\hat{Z})}.\]
    \end{lemma}
    \begin{proof}
        To generate a sample from $Y \gets Q$, we sample $\hat{Y} \gets \hat{Q}$ and set $Y = g(\hat{Y})$.
        We consider the reverse process: Given $y \in \mathcal{Y}$, define $\hat{Q}_y$ to be the conditional distribution of $\hat{Y} \gets \hat{Q}$ conditioned on $g(\hat{Y})=y$. That is, $\hat{Q}_y$ is a distribution such that we can generate a sample $\hat{Y} \gets \hat{Q}$ by first sampling $Y \gets Q$ and then sampling $\hat{Y} \gets \hat{Q}_Y$.
        Note that if $g$ is an injective function, then $\hat{Q}_y$ is a point mass.

        We have the following key identity. Formally, this relates the Radon-Nikodym derivative of the postprocessed distributions ($P$ with respect to $Q$) to the Radon-Nikodym derivative of the original distributions ($\hat{P}$ with respect to $\hat{Q}$) via the conditional distribution $\hat{Q}_y$. 
        \[\forall y \in \mathcal{Y} ~~~~~ \frac{P(y)}{Q(y)} = \ex{\hat{Y} \gets \hat{Q}_y}{\frac{\hat{P}(\hat{Y})}{\hat{Q}(\hat{Y})}}.\]
        To see where this identity comes from, write
        \begin{align*}
            \ex{\hat{Y} \gets \hat{Q}_y}{\frac{\hat{P}(\hat{Y})}{\hat{Q}(\hat{Y})}} 
            &= \int_{\{\hat{y} : g(\hat{y})=y\}} \frac{\hat{P}(\hat{y})}{\hat{Q}(\hat{y})} \cdot \hat{Q}_y(\hat{y}) \mathrm{d}\hat{y} \\
            &= \int_{\{\hat{y} : g(\hat{y})=y\}} \frac{\hat{P}(\hat{y})}{\hat{Q}(\hat{y})} \cdot \frac{\hat{Q}(\hat{y})}{\int_{\{\tilde{y} : g(\tilde{y})=y\}} \hat{Q}(\tilde{y}) \mathrm{d}\tilde{y}} \mathrm{d}\hat{y} \\
            &= \frac{\int_{\{\hat{y} : g(\hat{y})=y\}} \hat{P}(\hat{y}) \mathrm{d}\hat{y}}{\int_{\{\tilde{y} : g(\tilde{y})=y\}} \hat{Q}(\tilde{y}) \mathrm{d}\tilde{y}} \\
            &= \frac{P(y)}{Q(y)}.
        \end{align*}
        
        Finally, we have
        \begin{align*}
            \ex{Z \gets \privloss{P}{Q}}{\exp(t Z)} 
            &= \ex{Y \gets P}{\exp(t \cdot \llr{P}{Q}(Y))} \\
            &= \ex{Y \gets Q}{\exp((t+1) \cdot \llr{P}{Q}(Y))} \tag{Lemma \ref{lem:dual_privloss}}\\
            &= \ex{Y \gets Q}{\left(\frac{P(Y)}{Q(Y)}\right)^{t+1}}\\
            &= \ex{Y \gets Q}{\left( \ex{\hat{Y} \gets \hat{Q}_Y}{\frac{\hat{P}(\hat{Y})}{\hat{Q}(\hat{Y})}}\right)^{t+1}}\\
            &\le \ex{Y \gets Q}{\ex{\hat{Y} \gets \hat{Q}_Y}{\left( \frac{\hat{P}(\hat{Y})}{\hat{Q}(\hat{Y})}\right)^{t+1}}} \tag{Jensen}\\
            &= {\ex{\hat{Y} \gets \hat{Q}}{\left( \frac{\hat{P}(\hat{Y})}{\hat{Q}(\hat{Y})}\right)^{t+1}}}\\
            &= \ex{\hat{Y} \gets \hat{Q}}{\exp((t+1) \cdot \llr{\hat{P}}{\hat{Q}}(\hat{Y}))} \\
            &= \ex{\hat{Y} \gets \hat{P}}{\exp(t \cdot \llr{\hat{P}}{\hat{Q}}(\hat{Y}))} \tag{Lemma \ref{lem:dual_privloss}}\\
            &= \ex{\hat{Z} \gets \privloss{\hat{P}}{\hat{Q}}}{\exp(t \hat{Z})},
        \end{align*}
        where the inequality follows from Jensen's inequality and the convexity of the function $v \mapsto v^{t+1}$.
    \end{proof}
    \begin{corollary}\label{cor:postprocesing}
        Let $\hat{M} : \mathcal{X}^n \to \hat{\mathcal{Y}}$ satisfy $\rho$-zCDP. Let $g : \hat{\mathcal{Y}} \to \mathcal{Y}$ be an arbitrary function. Define $M : \mathcal{X}^n \to \mathcal{Y}$ by $M(x) = g(\hat{M}(x))$. Then $M$ is also $\rho$-zCDP.
    \end{corollary}
    \begin{proof}
        Fix neighbouring inputs $x,x' \in \mathcal{X}^n$. 
        Let $P=M(x)$, $Q=M(x')$, $\hat{P}=\hat{M}(x)$, and $\hat{Q}=\hat{M}(x')$.
        By Lemma \ref{lem:postprocessing} and the assumption that $\hat{M}$ is $\rho$-zCDP, for all $t \ge 0$, 
        \begin{align*}
            \ex{Z \gets \privloss{M(x)}{M(x')}}{\exp(tZ)} &= \ex{Z \gets \privloss{P}{Q}}{\exp(tZ)} \\ 
            &\le \ex{\hat{Z} \gets \privloss{\hat{P}}{\hat{Q}}}{\exp(t\hat{Z})} \\
            &= \ex{\hat{Z} \gets \privloss{\hat{M}(x)}{\hat{M}(x')}}{\exp(t\hat{Z})} \\
            &\le \exp(t(t+1)\rho),
        \end{align*}
        which implies that $M$ is also $\rho$-zCDP.
    \end{proof}
    
    \subsection{Composition of Approximate $(\varepsilon,\delta)$-DP}
    
    Thus far we have only considered the composition of pure DP mechanisms (Theorems \ref{thm:basic_composition} \& \ref{thm:advancedcomposition_pure}) and the Gaussian mechanism (Corollary \ref{cor:gauss_adp_exact_multi}). What about approximate $(\varepsilon,\delta)$-DP?
    
    We have the following result which extends Theorems \ref{thm:basic_composition} \& \ref{thm:advancedcomposition_pure} to approximate DP and to adaptive composition.
    
    \begin{theorem}[Advanced Composition Starting with Approximate DP]\label{thm:advancedcomposition_approx}
        For $j \in [k]$, let $M_j : \mathcal{X}^n \times \mathcal{Y}_{j-1} \to \mathcal{Y}_j$ be randomized algorithms. Suppose $M_j$ is $(\varepsilon_j,\delta_j)$-DP for each $j \in [k]$.
        For $j \in [k]$, inductively define $M_{1 \cdots j} : \mathcal{X}^n \to \mathcal{Y}_j$ by $M_{1 \cdots j}(x)=M_j(x,M_{1 \cdots (j-1)}(x))$, where each algorithm is run independently and $M_{1 \cdots 0}(x) = y_0$ for some fixed $y_0 \in \mathcal{Y}_0$. Then $M_{1 \cdots k}$ is $(\varepsilon,\delta)$-DP for any $\delta>\sum_{j=1}^k \delta_j$ with \[\varepsilon = \min\left\{ \sum_{j=1}^k \varepsilon_j , \frac12 \sum_{j=1}^k \varepsilon_j^2 + \sqrt{2\log(1/\delta') \sum_{j=1}^k \varepsilon_j^2} \right\},\] where $\delta' = \delta - \sum_{j=1}^k \delta_j$.
    \end{theorem}

    Intuitively, if you consider the privacy loss $\privloss{M(x)}{M(x')}$ (where $x,x' \in \mathcal{X}^n$ are arbitrary neighbouring inputs), then $M$ being $(\varepsilon,\delta)$-DP is equivalent to the privacy loss being in $[-\varepsilon,+\varepsilon]$ with probability at least $1-\delta$; otherwise the privacy loss can be arbitrary (including possibly infinite).
    Informally, the proof of Theorem \ref{thm:advancedcomposition_approx} uses a union bound to show that with probability at least $1-\sum_{j=1}^k \delta_j$ all of the privacy losses of the $k$ algorithms are bounded by their respective $\varepsilon_j$s. Once we condition on this event, the proof proceeds as before.
    
    Formally, rather than reasoning about possibly infinite privacy losses, we use the following decomposition result.
    
    \begin{lemma}\label{lem:decomposition}
        Let $P$ and $Q$ be probability distributions over $\mathcal{Y}$. Fix $\varepsilon, \delta \ge 0$. Suppose that, for all measurable $S \subset \mathcal{Y}$, we have $P(S) \le e^\varepsilon \cdot Q(S) + \delta$ and $Q(S) \le e^\varepsilon P(S) + \delta$.
        
        Then there exist distributions $P',Q',P'',Q''$ over $\mathcal{Y}$ with the following properties.
        We can express $P$ and $Q$ as convex combinations of these distributions, namely $P = (1-\delta)P'+\delta P''$ and $Q = (1-\delta) Q' + \delta Q''$.
        And, for every measurable $S \subset \mathcal{Y}$, we have $e^{-\varepsilon} \cdot Q'(S) \le P'(S) \le e^\varepsilon \cdot Q'(S)$.
    \end{lemma}
    \begin{proof}
        Fix $\varepsilon_1, \varepsilon_2 \in [0,\varepsilon]$ to be determined later.
        Define distributions $P'$, $P''$, $Q'$, and $Q''$ as follows.\footnote{Formally, $P(y)$, $P'(y)$, $P''(y)$, $Q(y)$, $Q'(y)$, and $Q''(y)$ denote the Radon-Nikodym derivative of these distributions with respect to some base measure -- usually either the counting measure (in which case these quantities are probability mass functions) or Lebesgue measure (in which case these quantities are probability density functions) -- in any case, we can take $P+Q$ to be the base measure.}
        For all points $y \in \mathcal{Y}$,
        \begin{align*}
            P'(y) &= \frac{\min\{ P(y) , e^{\varepsilon_1} \cdot Q(y) \}}{1-\delta_1}, \\
            P''(y) &= \frac{P(y) - (1-\delta_1)P'(y)}{\delta_1} = \frac{\max\{0, P(y) - e^{\varepsilon_1} \cdot Q(y)\}}{\delta_1}, \\
            Q'(y) &= \frac{\min\{ Q(y), e^{\varepsilon_2} \cdot P(y) \}}{1-\delta_2}, \\
            Q''(y) &= \frac{Q(y) - (1-\delta_2) Q'(y)}{\delta_2} = \frac{\max\{0, Q(y) - e^{\varepsilon_2} \cdot P(y)\}}{\delta_2},
        \end{align*}
        where $\delta_1$ and $\delta_2$ are appropriate normalizing constants.
        
        By construction, $(1-\delta_1)P' + \delta_1 P'' = P$ and $(1-\delta_2)Q'+\delta_2 Q'' = Q$.
        
        If $\delta_1=\delta_2=\delta$, then we have the appropriate decomposition and, for all $y \in \mathcal{Y}$, we have \[e^{-\varepsilon} \le e^{-\varepsilon_2} \le \frac{P'(y)}{Q'(y)} = \frac{\min\{P(y), e^{\varepsilon_1} \cdot Q(y)\}}{\min\{Q(y), e^{\varepsilon_2} \cdot P(y)\}} \le e^{\varepsilon_1} \le e^\varepsilon ,\] as required.
        If $\delta_1 = \delta_2 < \delta$, we can change the decomposition to \[P = (1-\delta) P' + (\delta-\delta_1) P' + (1-\delta_1) P'' = (1-\delta) P' + \delta \cdot \left(\frac{\delta-\delta_1}{\delta} P' + \frac{\delta_1}{\delta} P''\right)\] and likewise for $Q$, which also yields the result.
        
        It only remains to show that we can ensure that $\delta_1 = \delta_2 \le \delta$ by appropriately setting $\varepsilon_1, \varepsilon_2 \in [0,\varepsilon]$.
        We have \[\delta_1 = \int_{\mathcal{Y}} \max\{0, P(y) - e^{\varepsilon_1} \cdot Q(y)\} \mathrm{d}y = \int_S P(y) - e^{\varepsilon_1} \cdot Q(y) \mathrm{d}y = P(S) - e^{\varepsilon_1} Q(S),\] where $S = \{y \in \mathcal{Y} : P(y) \ge e^{\varepsilon_1} \cdot Q(y)\}$. If $\varepsilon_1=\varepsilon$, then $\delta_1 \le \delta$ by the assumptions of the Lemma. If $\varepsilon_1 = 0$, then $\delta_1=\tvd{P}{Q}$.
         By decreasing $\varepsilon_1$, we continuously increase $\delta_1$. Thus we can pick $\varepsilon_1 \in [0,\varepsilon]$ such that $\delta_1 = \min\{ \delta, \tvd{P}{Q} \}$. Similarly, we can pick $\varepsilon_2 \in [0,\varepsilon]$, such that $\delta_2 = \min\{ \delta, \tvd{P}{Q} \}$.
    \end{proof}
    
    We can extend Lemma \ref{lem:decomposition} to show that any pair of distributions staisfying the $(\varepsilon,\delta)$-DP guarantee can be represented as a postprocessing of $(\varepsilon,\delta)$-DP randomized response:
    \begin{corollary}\label{cor:kov}
        Let $P$ and $Q$ be probability distributions over $\mathcal{Y}$. Fix $\varepsilon, \delta \ge 0$. Suppose that, for all measurable $S \subset \mathcal{Y}$, we have $P(S) \le e^\varepsilon \cdot Q(S) + \delta$ and $Q(S) \le e^\varepsilon P(S) + \delta$.
        
        Then there exist distributions $A$, $B$, $P''$, and $Q''$ over $\mathcal{Y}$ such that
        \begin{align*}
            P &= (1-\delta)\frac{e^\varepsilon}{e^\varepsilon+1} A + (1-\delta) \frac{1}{e^\varepsilon+1} B + \delta P'' ,\\
            Q &= (1-\delta)\frac{e^\varepsilon}{e^\varepsilon+1} B + (1-\delta) \frac{1}{e^\varepsilon+1} A + \delta Q'' .
        \end{align*}
    \end{corollary}
    To interpret Corollary \ref{cor:kov}, imagine $P=M(x)$ and $Q=M(x')$ are the outputs of a $(\varepsilon,\delta)$-DP mechanism on neighbouring inputs. Define an $(\varepsilon,\delta)$-DP analog of randomized response $R : \{0,1\} \to \{0,1\} \times \{\top,\bot\}$ by $\pr{}{R(b)=(b,\bot)} = \delta$ and $\pr{}{R(b)=(b,\top)} = (1-\delta) \frac{e^\varepsilon}{e^\varepsilon+1}$ and $\pr{}{R(b)=(1-b,\top)} = (1-\delta) \frac{1}{e^\varepsilon+1}$ for both $b \in \{0,1\}$. Then Corollary \ref{cor:kov} tells us that we can simulate $M$ by mapping the pair of inputs $x \mapsto 0$ and $x' \mapsto 1$ and then postprocessing the outputs with the randomized function $F$ defined by $F(0,\top)=P''$, $F(1,\top)=Q''$, $F(0,\bot)=A$, and $F(1,\bot)=B$. That is, $M(x)=F(R(0))$ and $M(x')=F(R(1))$.
    \begin{proof}[Proof of Corollary \ref{cor:kov}.]
        By Lemma \ref{lem:decomposition}, there exist distributions $P',Q',P'',Q''$ over $\mathcal{Y}$ such that $P = (1-\delta)P'+\delta P''$ and $Q = (1-\delta) Q' + \delta Q''$
        and, for every measurable $S \subset \mathcal{Y}$, we have $e^{-\varepsilon} \cdot Q'(S) \le P'(S) \le e^\varepsilon \cdot Q'(S)$.
        
        If $\varepsilon=0$, let $A=B=P'=Q'$. Otherwise, let 
        \[
            A = \frac{e^\varepsilon P' - Q'}{e^\varepsilon-1}, ~~~\text{ and }~~~
            B = \frac{e^\varepsilon Q' - P'}{e^\varepsilon-1}.
        \]
        We can verify that $A$ and $B$ are probability distributions, since, for all $S$, we have $e^{-\varepsilon} \cdot Q'(S) \le P'(S) \le e^\varepsilon \cdot Q'(S)$, which implies $A(S) \ge 0$ and $B(S) \ge 0$. Also $A(\mathcal{Y})=B(\mathcal{Y})=1$. And we have $\frac{e^\varepsilon}{e^\varepsilon+1} A + \frac{1}{e^\varepsilon+1} B =P'$ and $\frac{e^\varepsilon}{e^\varepsilon+1} B + \frac{1}{e^\varepsilon+1} A =Q'$, as required
    \end{proof}
    
    The proof of Theorem \ref{thm:advancedcomposition_approx} is, unfortunately, quite technical. Most of the steps are the same as we have seen in the pure DP case. The only novelty is applying the decomposition of Lemma \ref{lem:decomposition} inductively; this requires cumbersome notation, but is otherwise straightforward.
    \begin{proof}[Proof of Theorem \ref{thm:advancedcomposition_approx}.]
        Fix neighbouring datasets $x,x'\in\mathcal{X}^n$.
        We inductively define distributions $P_j$ and $Q_j$ on $\mathcal{Y}_0 \times \mathcal{Y}_1 \times \cdots \times \mathcal{Y}_j$ as follows.
        For $j \in [k]$, $P_j = (Y_0, Y_1, \cdots, Y_{j-1}, M_j(x,Y_{j-1}))$, where $(Y_1, \cdots, Y_{j-1}) \gets P_{j-1}$, and $Q_j = (Y_0, Y_1, \cdots, Y_{j-1}, M_j(x', Y_{j-1}))$, where $(Y_1, \cdots, Y_{j-1}) \gets Q_{j-1}$. We define $P_0=Q_0$ to be the point mass on $y_0$.
        
        We will prove by induction that, for each $j \in [k]$, there exist distributions $P_j'$, $P_j''$, $Q_j'$, and $Q_j''$ on $\mathcal{Y}_0 \times \mathcal{Y}_1 \times \cdots \times \mathcal{Y}_j$ such that \[P_j = \prod_{\ell=1}^j(1-\delta_\ell) P_j' + \left( 1 - \prod_{\ell=1}^j(1-\delta_\ell) \right) P_j''\] and \[Q_j = \prod_{\ell=1}^j(1-\delta_\ell) Q_j' + \left( 1 - \prod_{\ell=1}^j(1-\delta_\ell) \right) Q_j''\]
        and, for all $t \ge 0$, \[\ex{Z_j' \gets \privloss{P_j'}{Q_j'}}{\exp(t Z_j')} \le \exp\left( \frac{t (t+1)}{2} \sum_{\ell=1}^j \varepsilon_\ell^2 \right)\] and, for all measurable $S \subset \mathcal{Y}_0 \times \mathcal{Y}_1 \times \cdots \times \mathcal{Y}_j$, $P_j'(S) \le \exp\left( \sum_{\ell=1}^j \varepsilon_\ell \right) \cdot Q_j'(S)$.
        
        Before proving the inductive claim, we show that it suffices to prove the result. Fix an arbitrary measurable $S \subset \mathcal{Y}_k$ and let $\tilde{S} = \mathcal{Y}_0 \times \mathcal{Y}_1 \times \cdots \times \mathcal{Y}_{k-1} \times S$. We have
        \begin{align*}
            \pr{}{M(x) \in S} &= P_k(\tilde{S}) \tag{Postprocessing} \\
            &= \prod_{\ell=1}^k(1-\delta_\ell) P_k'(\tilde{S}) + \left( 1 - \prod_{\ell=1}^k(1-\delta_\ell) \right) P_k''(\tilde{S}) \\
            &\le \prod_{\ell=1}^k(1-\delta_\ell) P_k'(\tilde{S}) + \sum_{j=1}^k\delta_j \tag{$P_k''(\tilde{S}) \le 1$ and $1 - \prod_{\ell=1}^k(1-\delta_\ell) \le \sum_{j=1}^k\delta_j$} \\
            &\le \prod_{\ell=1}^k(1-\delta_\ell) \left( e^{\varepsilon} \cdot Q_k'(\tilde{S}) + \delta' \right) + \sum_{j=1}^k\delta_j \tag{*}\\
            &\le e^{\varepsilon} \cdot \prod_{\ell=1}^k(1-\delta_\ell) \cdot Q_k'(\tilde{S}) + \delta \tag{$\delta = \delta' + \sum_{j=1}^k\delta_j$} \\
            &\le e^{\varepsilon} \cdot  Q_k(\tilde{S}) + \delta \tag{$Q_k = \prod_{\ell=1}^k(1-\delta_\ell) Q_k' + \left( 1 - \prod_{\ell=1}^k(1-\delta_\ell) \right) Q_k''$} \\
            &= e^\varepsilon \cdot \pr{}{M(x') \in S} + \delta.
        \end{align*}
        The inequality $P_k'(\tilde{S}) \le e^\varepsilon \cdot Q_k'(\tilde{S}) + \delta'$ (*) follows the proof we have seen before. Our inductive conclusion includes a pure DP result -- $P_j'(\tilde{S}) \le \exp\left( \sum_{\ell=1}^j \varepsilon_\ell \right) \cdot Q_j'(\tilde{S})$ -- and a concentrated DP result -- for all $t \ge 0$, we have $\ex{Z_j' \gets \privloss{P_j'}{Q_j'}}{\exp(t Z_j')} \le \exp\left( \frac{t (t+1)}{2} \sum_{\ell=1}^j \varepsilon_\ell^2 \right)$, which implies
        \begin{align*}
            P_k'(\tilde{S}) &\le e^\varepsilon \cdot Q_k'(\tilde{S}) + \pr{Z_k' \gets \privloss{P_k'}{Q_k'}}{Z_k' > \varepsilon} \tag{Proposition \ref{prop:privloss_adp}} \\ &\le e^\varepsilon \cdot Q_k'(\tilde{S}) + \ex{Z_k' \gets \privloss{P_k'}{Q_k'}}{\exp(t(Z_k' - \varepsilon))} \tag{$\mathbb{I}[Z_k'>\varepsilon] \le \exp(t(Z_k'-\varepsilon))$}\\
            &\le e^\varepsilon \cdot Q_k'(\tilde{S}) + \exp\left(\frac{t(t+1)}{2} \sum_{j=1}^k \varepsilon_j^2 \right) \cdot \exp(-t\varepsilon) \tag{Induction conclusion} \\
            &\le e^\varepsilon \cdot Q_k'(\tilde{S}) + \delta',
        \end{align*}
        where the final inequality holds for the case $\varepsilon = \frac12 \sum_{j=1}^k \varepsilon_j^2 + \sqrt{2\log(1/\delta') \sum_{j=1}^k \varepsilon_j^2}$ and requires setting $t = \frac{\varepsilon}{\sum_{j=1}^k \varepsilon_j^2} - \frac12 = \sqrt{\frac{2\log(1/\delta')}{\sum_{j=1}^k \varepsilon_j^2}}$.
        
        It only remains for us to perform the induction. The base case ($j=0$) is trivial.
        
        Fix $j \in [k]$ and assume the induction hypothesis holds for $j-1$. The distribution $P_j$ is defined as a mixture (i.e., convex combination) of $P_j|_Y$ for $Y \gets P_{j-1}$, where $P_j|_Y := (Y,M_j(x,Y_{j-1}))$.
        For every $y$, we apply Lemma \ref{lem:decomposition} to the conditional distribution $P_j|_y$ and then we take the convex combination of these decompositions to obtain a decomposition of $P_j$. Of course, we must also decompose $Q_j$ at the same time.
        
        For each $y \in \mathcal{Y}_0 \times \mathcal{Y}_1 \times \cdots \mathcal{Y}_{j-1}$, the conditional distributions satisfy $\forall S ~~ P_j|_y(S) \le e^{\varepsilon_j} Q_j|_y(S) + \delta_j$ and vice versa. Thus Lemma \ref{lem:decomposition} allows us to decompose the conditional distributions $P_j|_y$ and $Q_j|_y$ as $P_j|_y = (1-\delta_j) P_j'|_y + \delta_j P_j''|_y$ and $Q_j|_y = (1-\delta_j) Q_j'|_y + \delta_j Q_j''|_y$ where $e^{-\varepsilon_j}\cdot Q_j'|_y(S) \le P_j'|_y(S) \le e^{\varepsilon_j}\cdot Q_j'|_y(S)$ for all $S$. This gives us the desired decomposition:
        \begin{align*}
            P_j &= \ex{Y \gets P_{j-1}}{P_j|_Y}\\
            &= \ex{Y \gets P_{j-1}}{(1-\delta_j)P_j'|_Y+\delta_jP_j''|_Y}\\
            &= \prod_{\ell=1}^{j-1}(1\!-\!\delta_\ell) \! \ex{Y \gets P_{j-1}'}{(1\!-\!\delta_j)P_j'|_Y\!+\!\delta_jP_j''|_Y}  \!+\! \left( 1 \!- \!\prod_{\ell=1}^{j-1}(1\!-\!\delta_\ell)\right)\! \ex{Y \gets P_{j-1}''}{(1\!-\!\delta_j)P_j'|_Y\!+\!\delta_jP_j''|_Y}\\
            &= \prod_{\ell=1}^{j}(1-\delta_\ell)  \ex{Y \gets P_{j-1}'}{P_j'|_Y}  + \delta_j \prod_{\ell=1}^{j-1}(1-\delta_\ell)  \ex{Y \gets P_{j-1}'}{P_j''|_Y} \\&~~~~~+ \left( 1 - \prod_{\ell=1}^{j-1}(1-\delta_\ell)\right)(1-\delta_j) \ex{Y \gets P_{j-1}''}{P_j'|_Y} + \left( 1 - \prod_{\ell=1}^{j-1}(1-\delta_\ell)\right)\delta_j \ex{Y \gets P_{j-1}''}{P_j''|_Y}.
        \end{align*}
        Thus we define the new decomposition as $P_j' = P_j'|_Y$ for $Y \gets P_{j-1}'$ and $Q_j' = Q_j'|_Y$ for $Y \gets Q_{j-1}'$. The ``weight'' of $P_j'$ is the product of the weight of $P_{j-1}'$ (i.e., $\prod_{\ell=1}^{j-1} (1-\delta_\ell)$) and the weight of $P_j'|_Y$ (i.e., $1-\delta_j$), as required. The remaining parts of the decomposition are combined to define $P_j'' = $ and $Q_j''$; note that $P_j''$ includes both $P_j'|_Y$ for $Y \gets P_{j-1}''$ and $P_j''|_Y$ for $Y \gets P_{j-1}$.
        It is easy to verify that this decomposition satisfies the requirements of the induction:
        \begin{align*}
            &\ex{Z_j' \gets \privloss{P_j'}{Q_j'}}{\exp(tZ_j')} \\
            &~= \ex{Y_j' \gets P_j'}{\exp\left(t \cdot \llr{P_j'}{Q_j'}(Y_j')\right)} \\
            &~= \ex{Y_{j-1}' \gets P_{j-1}'}{\ex{Y_j' \gets P_j'|_{Y_{j-1}}}{\exp\left(t \cdot \left(\llr{P_{j-1}'}{Q_{j-1}'}(Y_{j-1}') + \llr{P_j'|_{Y_{j-1}}}{Q_j'|_{Y_{j-1}}}(Y_j')\right)\right)}}\\
            &~= \ex{Y_{j-1}' \gets P_{j-1}'}{\exp\left(t \cdot \llr{P_{j-1}'}{Q_{j-1}'}(Y_{j-1}')\right) \cdot \ex{Y_j' \gets P_j'|_{Y_{j-1}}}{\exp\left(t \cdot  \llr{P_j'|_{Y_{j-1}}}{Q_j'|_{Y_{j-1}}}(Y_j')\right)}}\\
            &~= \ex{Y_{j-1}' \gets P_{j-1}'}{\exp\left(t \cdot \llr{P_{j-1}'}{Q_{j-1}'}(Y_{j-1}')\right) \cdot \ex{Z_j' \gets \privloss{P_j'|_{Y_{j-1}}}{Q_j'|_{Y_{j-1}}}}{\exp\left(t \cdot  Z_j'\right)}}\\
            &~\le \ex{Y_{j-1}' \gets P_{j-1}'}{\exp\left(t \cdot \llr{P_{j-1}'}{Q_{j-1}'}(Y_{j-1}')\right) \cdot \exp\left(t (t+1) \frac12 \varepsilon_j^2\right)} \tag{Proposition \ref{prop:pdp2cdp} \& $|Z_j'|\le\varepsilon_j$}\\
            &~\le\exp\left(\frac{t(t+1)}{2} \sum_{\ell=1}^{j-1} \varepsilon_\ell^2 \right) \cdot \exp\left(t (t+1) \frac12 \varepsilon_j^2\right) \tag{Induction hypothesis}\\
            &~=\exp\left(\frac{t(t+1)}{2} \sum_{\ell=1}^{j} \varepsilon_\ell^2 \right).
        \end{align*}
        And, for pure DP, we have $P_j'(S) = \ex{Y_{j-1}' \gets P_{j-1}'}{P_j'|_{Y_{j-1}}(S)} \le \ex{Y_{j-1}' \gets P_{j-1}'}{e^{\varepsilon_j} Q_j'|_{Y_{j-1}}(S)} \le \exp\left(\sum_{\ell=1}^{j-1} \varepsilon_\ell\right) \cdot \ex{Y_{j-1}' \gets Q_{j-1}'}{e^{\varepsilon_j} Q_j'|_{Y_{j-1}}(S)} = \exp\left(\sum_{\ell=1}^{j} \varepsilon_\ell\right) \cdot Q_j'(S)$ for all measurable $S$.
    \end{proof}
    
\section{Asymptotic Optimality of Composition}
    Is the advanced composition theorem optimal? That is, could we prove a result that is stronger?
    This is an important question, but we first need to think about what optimality even means. Recall that, in Section \ref{sec:basic_optimal}, we proved that basic composition is optimal, but then we showed that we could do better by relaxing the requirement from pure DP to approximate DP or concentrated DP. To prove asymptotic optimality of advanced composition, we will show that no algorithm can provide better accuracy than advanced composition gives (except for constant factors) subject to approximate DP. Furthermore, we will see that the analysis is not specific to approximate DP.

    Combining advanced composition (Theorem \ref{thm:advancedcomposition_pure} or \ref{thm:advancedcomposition_approx}) with Laplace noise addition shows that we can answer $k$ bounded sensitivity queries (e.g., counting queries) with noise scale $\Theta(\sqrt{k/\rho})$ for each query, where $\rho$ only depends on the privacy parameters, e.g., $\rho = \Theta(\varepsilon^2/\log(1/\delta))$ for $(\varepsilon,\delta)$-DP. (Gaussian noise addition also gives the same asymptotics, per Corollary \ref{cor:gauss_adp_exact_multi}.)
    
    We can prove that this asymptotics -- average error per query $\Omega(\sqrt{k})$ -- is optimal. Formally, we have the following result.
    
    \begin{theorem}[Negative Result for Error of Private Mean Estimation]\label{thm:lowerbound}
        Let $\mathcal{X}=\{0,1\}^k$ and $\mathcal{Y} = [0,1]^k$.
        Let $M : \mathcal{X}^n \to \mathcal{Y}$ satisfy $(\varepsilon,\delta)$-DP. If $\delta \le 1/100n$ and $k \ge 200 (e^\varepsilon-1)^2 n$, then there exists some $x \in \mathcal{X}^n$ such that \[\sqrt{\ex{}{\frac1k\|M(x)-\overline{x}\|_2^2}} \ge \min\left\{ \frac{\sqrt{k}}{16 \cdot n \cdot (e^\varepsilon-1)} , \frac{1}{10}\right\},\] where $\overline{x} = \frac1n \sum_{i=1}^n x_i \in [0,1]^k$ is the mean of input dataset.
    \end{theorem}
    
    Theorem \ref{thm:lowerbound} shows that any DP algorithm answering $k$ queries must have error per query scaling with $\Omega(\sqrt{k})$, which matches the guarantees of the advanced composition theorem.
    We briefly remark on some of the properties of this theorem:
    
    First, $M$ could just output $\frac12$ for each coordinate; this is trivially private and has root mean square error at most $\frac12$. The theorem must apply to such an algorithm too, which is the fundamental reason why the lower bound in the conclusion of Theorem \ref{thm:lowerbound} cannot be larger than a constant $\frac{1}{10}$.
    
    Second, the assumption $\delta \le 1/100n$ is also necessary, up to constant factors. If $\delta \gg 1/n$, then $M$ could sample $n\delta$ of the inputs and return the sample mean. This would be $(0,\delta)$-DP and would give accuracy $\sqrt{\ex{}{\frac1k\|M(x)-\overline{x}\|_2^2}} \le \frac{1}{\sqrt{n\delta}} \ll 1$.
    Note that the advanced composition theorem includes a $\sqrt{\log(1/\delta)}$ term. It is possible to extend the negative results to include such a term too \citep{steinke2015between} (see also Lemma 2.3.6 of \citet{bun2016new}), but we do not do this here for simplicity.
    
    Third, the assumption $k \ge 200 (e^\varepsilon-1)^2 n$ is not really necessary; it is an artifact of our analysis. If $k \ll \varepsilon^2 n$, then the privacy error is lower than the sampling error (if we think of $x$ as consisting of $n$ samples from some distribution). A different analysis is possible in this case.
    
    Fourth, Theorem \ref{thm:lowerbound} has $e^\varepsilon-1$ in the denominator, where our positive results have $\varepsilon$. For small $\varepsilon$, we have $e^\varepsilon-1\approx\varepsilon$. But, for large $\varepsilon$, there is an exponential difference. Surprisingly, this is inherent; by using discrete noise \citep{canonne2020discrete} in place of continuous Laplace noise it is possible to improve the positive results to yield this asymptotic behaviour. However, we are generally not interested in the large $\varepsilon$ setting.
    
    Finally, fifth, this theorem is not merely an esoteric impossibility result. It corresponds to realistic attacks, which are known as ``tracing attacks'' \citep{dwork2017exposed} or ``membership inference attacks'' \citep{shokri2017membership}.
    
    \begin{proof}[Proof of Theorem \ref{thm:lowerbound}.]
        The theorem guarantees that there exists a specific input $x$ on which $M$ has high error. In general, $x$ must depend on $M$.
        To prove this we show that, for a random input from a carefully chosen distribution, any $M$ must have high error. It follows that for each specific $M$ there must exist some fixed input with high error.
        
        For $p \in [0,1]^k$, let $\mathcal{D}_p$ be the product distribution over $\{0,1\}^k$ with mean $p$. Our random input $X \in \mathcal{X}^n$ will consist of $n$ independent draws from $\mathcal{D}_p$. Furthermore, we select the mean parameter randomly too. That is, $P \in [0,1]^d$ is uniformly random and $X$ consists of $n$ conditionally independent draws from $\mathcal{D}_P$.
        
        We analyze the quantity \[Z := \sum_{i=1}^n \left\langle M(X) - P , X_i - P \right\rangle .\]
        Applying Lemma \ref{lem:fingerprinting} with $f(x) = \ex{}{M(X)_j | X_j=x}$ and summing over the coordinates $j \in [k]$ shows that \[\ex{P \gets \mathsf{Uniform}([0,1]^k) \atop X \gets \mathcal{D}_P^n}{Z \!+\! \|M(X)\!-\!\overline{X}\|_2^2} \!=\! \sum_{j=1}^k \! \ex{P \gets \mathsf{Uniform}([0,1]^k) \atop X \gets \mathcal{D}_P^n}{\begin{array}{c}(M(X)_j-P_j) \cdot \sum_{i=1}^n (X_{i,j}-P_j) \\ + (M(X)_j-\overline{X}_j)^2 \end{array}} \ge \frac{k}{12}.\]
        Denoting $\alpha^2k = \ex{}{\|M(X)-\overline{X}\|_2^2}$, we have $\ex{}{Z} \ge \frac{k}{12} - \alpha^2 k$.
        Intuitively, $Z$ measures the total correlation between the output of $M$ and its inputs. What Lemma \ref{lem:fingerprinting} shows is that, if $M$ is accurate -- i.e., $\ex{}{\|M(X)-\overline{X}\|_2^2} \le o(k)$ -- then this correlation must be large. 
        
        The punchline of the proof is that we show that differential privacy means the correlation must be small, which conflicts with the fact that we have proven it must be large. Ergo, we will obtain the desired impossibility result.

        For $i \in [n]$, define \[Z_i := \langle M(X) - P , X_i - P \rangle,\] so that $Z = \sum_{i=1}^n Z_i$.
        Let $X_0$ be a fresh sample from $\mathcal{D}_P$ that is (conditionally) independent from $X_1, \cdots, X_n$.
        Let $M(X_0,X_{-i})$ denote running $M$ on the dataset $X$ where $X_i$ has been replaced by $X_0$ and define \[\tilde{Z}_i := \langle M(X_0,X_{-i}) - P , X_i - P \rangle.\]
        By differential privacy, $M(X_0,X_{-i})$ is indistinguishable from $M(X)$, even if we condition on $X_0,X_1, \cdots, X_n$. Thus the distributions of $\tilde{Z}_i$ and $Z_i$ are also indistinguishable.
        
        Since $M(X_0,X_{-i})$ and $X_i$ are independent (conditioned on $P$) and $\ex{}{X_i-P}=\vec{0}$ (conditioned on $P$), we have $\ex{}{\tilde{Z}_i} = 0$ and \[\ex{P,X,M}{\tilde{Z}_i^2} = \sum_{j=1}^k \ex{P}{ P_j(1-P_j) \cdot \ex{X,M}{(M(X_0,X_{-i})_j-P_j)^2}}\le \frac14 \ex{P,X,M}{\|M(X)-P\|_2^2} .\]
        Now $\ex{}{\|M(X)-P\|_2^2} \le 2 \ex{}{\|M(X)-\overline{X}\|_2^2} + 2 \ex{}{\|\overline{X}-P\|_2^2} \le 2\alpha^2 k + \frac{k}{3n}$.\footnote{ $\ex{}{\|\overline{X}-P\|_2^2} = \sum_{j=1}^k \ex{P_j \gets \mathsf{Uniform}([0,1])}{\ex{Y_j \gets \mathsf{Binomial}(n,P_j)}{(\frac1n Y_j - P_j)^2}} = k \cdot \int_0^1 \frac{p(1-p)}{n} \mathrm{d}p= \frac{k}{6n}$.}

        Lemma \ref{lem:indistinguishable_expectation}, $|Z_i| \le k$, $|\tilde{Z}_i| \le k$, and $\ex{}{|\tilde{Z}_i|} \le \sqrt{\ex{}{\tilde{Z}_i^2}}$ (i.e., Jensen's inequality) gives \[\ex{}{Z_i} \le \ex{}{\tilde{Z}_i} + (e^\varepsilon-1)\ex{}{|\tilde{Z}_i|} + 2\delta k \le \frac{e^\varepsilon -1}{2} \sqrt{2\alpha^2 k + \frac{k}{3n}} + 2\delta k.\]
        Putting things together, we have \[\frac{k}{12} - \alpha^2 k \le \ex{}{Z} = \sum_{i=1}^n \ex{}{Z_i} \le n \cdot \left( \frac{e^\varepsilon -1}{2} \sqrt{2\alpha^2 k + \frac{k}{3n}} + 2\delta k \right).\]
        Ignoring terms that are (hopefully) low order, this is $\Omega(k) \le O(n \cdot \varepsilon \sqrt{\alpha^2 k})$, which rearranges to $\alpha = \sqrt{\ex{}{\frac1k\|M(X)-\overline{X}\|_2^2}} \ge \Omega\left(\frac{\sqrt{k}}{n \varepsilon}\right)$, which is the desired asymptotic result.
        To be precise, this rearranges to
        \[\alpha \ge \sqrt{ \left( \frac{1}{6} - 2\alpha^2 - 4n\delta \right)^2 \cdot \frac{k}{2n^2 \cdot (e^\varepsilon-1)^2} - \frac{1}{6n}}.\]
        If $\alpha \le 1/10$ and $\delta \le 1/100n$, then $\frac{1}{6} - 2\alpha^2 - 4n\delta \ge \frac{1}{10}$. If $k \ge 200 (e^\varepsilon-1)^2 n$, then $\left(\frac{1}{10}\right)^2 \cdot \frac{k}{2n^2(e^\varepsilon-1)^2} \ge \frac{1}{n}$. If all three of these conditions hold, then \[ \sqrt{\ex{}{\frac1k\|M(X)-\overline{X}\|_2^2}} = \alpha \ge \sqrt{ \frac{k}{200 \cdot n^2 \cdot (e^\varepsilon-1)^2}\left( 1 - \frac{1}{6}\right)} \ge \frac{\sqrt{k}}{16 n (e^\varepsilon-1)} .\] Hence, if $\delta \le 1/100n$ and $k \ge 200(e^\varepsilon-1)^2n$, then either $\alpha > 1/10$ or $\alpha \ge \sqrt{k}/16n(e^\varepsilon-1)$, as required.
    \end{proof}
    
    Now we prove the two lemmata that were used to prove Theorem \ref{thm:lowerbound}. We begin with the lemma showing that the correlation $Z$ must be large if $M$ is accurate.
    
    The lemma only contemplates one coordinate and then we sum over the $k$ coordinates in the proof of  Theorem \ref{thm:lowerbound}. That is, the function $f$ in the theorem is simply one coordinate of $M$ and we average out the randomness of $M$ and the other coordinates.

    \begin{lemma}\label{lem:fingerprinting}
        Let $f : \{0,1\}^d \to [0,1]$ be an arbitrary function. Let $P \in [0,1]$ be uniformly random and, conditioned on $P$, let $X_1, \cdots, X_n \in \{0,1\}$ be independent with $\ex{}{X_i}=P$ for each $i \in [n]$. Then \[\ex{X,P}{(f(X)-P) \cdot \sum_{i=1}^n (X_i-P)} + \ex{P}{ \ex{X}{f(X) - \overline{X}}^2} \ge \frac{1}{12}.\] 
    \end{lemma}
    By Jensen's inequality $\ex{P}{ \ex{X}{f(X) - \overline{X}}^2} \le \ex{P,X}{ \left(f(X) - \overline{X}\right)^2}$. Thus  \[\ex{X,P}{(f(X)-P) \cdot \sum_{i=1}^n (X_i-P) + \left( f(X) - \overline{X}\right)^2} \ge \frac{1}{12}.\] 
    
    To gain some intuition for the lemma statement, suppose $f(x) = \overline{x} = \frac1n \sum_{i=1}^n x_i$. Then \[\ex{}{\left( f(X)-P \right) \cdot \sum_{i=1}^n \left( X_i - P \right)} = \ex{}{\left( \overline{X} - P \right)\cdot\left( \sum_{i=1}^n X_i -P \right)} = \frac1n \sum_{i=1}^n  \ex{}{(X_i-P)^2} = \frac{1}{6}.\] The constant $\frac{1}{6} = \int_0^1 p(1-p) \mathrm{d}p$ in this example is slightly better than the constant $\frac{1}{12}$ in the general result.
    However, if $f(x) = \frac12$ is a constant function, then the constant is tight, as $\ex{P}{ \ex{X}{f(X) - \overline{X}}^2} = \ex{P}{\left(\frac12-P\right)^2} = \frac{1}{12}$.
    
    \begin{proof}[Proof of Lemma \ref{lem:fingerprinting}.]
        Define $g : [0,1] \to [0,1]$ by $g(p) = \ex{X \gets \mathcal{D}_p^n}{f(X)}$, where $\mathcal{D}_p^n$ denotes the product distribution over $\{0,1\}^n$ with each coordinate having mean $p$. Then
        \begin{align*}
            g'(p) &= \frac{\mathrm{d}}{\mathrm{d}p} \ex{X \gets \mathcal{D}_p^n}{f(X)} \\
            &= \sum_{x \in \{0,1\}^n} f(x) \frac{\mathrm{d}}{\mathrm{d}p} \prod_{\ell=1}^n \big( x_\ell \cdot p + (1-x_\ell) \cdot (1-p) \big) \\
            &= \sum_{x \in \{0,1\}^n} f(x)  \prod_{\ell=1}^n \big( x_\ell \cdot p + (1-x_\ell) \cdot (1-p) \big) \sum_{i=1}^n \frac{\frac{\mathrm{d}}{\mathrm{d}p} (x_i \cdot p + (1-x_i) \cdot (1-p))}{x_i \cdot p + (1-x_i) \cdot (1-p)} \tag{Product rule}\\
            &= \sum_{x \in \{0,1\}^n} f(x)  \prod_{\ell=1}^n \big( x_\ell \cdot p + (1-x_\ell) \cdot (1-p) \big) \sum_{i=1}^n \frac{2x_i-1}{x_i \cdot p + (1-x_i) \cdot (1-p)} \\
            &= \sum_{x \in \{0,1\}^n} f(x)  \prod_{\ell=1}^n \big( x_\ell \cdot p + (1-x_\ell) \cdot (1-p) \big) \sum_{i=1}^n \frac{x_i-p}{p (1-p)} \tag{Case analysis for $x_i \in \{0,1\}$} \\
            &= \ex{X \gets \mathcal{D}_p^n}{f(X) \cdot \sum_{i=1}^n \frac{X_i-p}{p(1-p)}}.
        \end{align*}
        Now we apply integration by parts to this derivative:
        \begin{align*}
            \ex{P \gets [0,1]}{\ex{X \gets \mathcal{D}_P^n}{f(X) \cdot \sum_{i=1}^n (X_i-P)}} 
            &= \int_0^1 g'(p) \cdot p(1-p) \mathrm{d}p \\
            &= \int_0^1 \left( \frac{\mathrm{d}}{\mathrm{d}p} g(p) \cdot p(1-p) \right) - g(p) \cdot (1-2p) \mathrm{d}p \\
            &= g(1) \cdot 1(1-1) - g(0) \cdot 0(1-0) + \int_0^1 g(p) \cdot (2p-1) \mathrm{d}p \\
            &= 2 \ex{P \gets [0,1]}{g(P) \cdot \left(P-\frac12\right)}.
        \end{align*}
        Using the fact that $\ex{P \gets [0,1]}{P-\frac12}=0$ and $\ex{X \gets \mathcal{D}_p^n}{X_i-p}=0$, we can center these expressions:
        \begin{align*}
            \ex{P \gets [0,1] \atop X \gets \mathcal{D}_P^n}{(f(X)-P) \cdot \sum_{i=1}^n (X_i-P)} &= 2 \ex{P \gets [0,1]}{\left( g(P) - \frac12 \right) \cdot \left( P - \frac12 \right)} \\
            &=\! \ex{P \gets [0,1]}{\!\left(\! g(P) \!-\! \frac12 \!\right)^2 \!+\! \left(\! P \!-\! \frac12 \!\right)^2 \!-\! \left(\left( g(P) \!-\! \frac12 \! \right) \!-\! \left( \! P \!-\! \frac12 \!\right) \!\right)^2\!}\\
            &\ge \ex{P \gets [0,1]}{ 0 + \left( P - \frac12 \right)^2 - \left( g(P) - P \right)^2} \\
            &= \frac{1}{12} - \ex{P \gets [0,1]}{\left( g(P) - P \right)^2}\\
            &= \frac{1}{12} - \ex{P \gets [0,1]}{\ex{X \gets \mathcal{D}_P^n}{f(X)-\overline{X}}^2}.
        \end{align*}
    \end{proof}
    \begin{lemma}\label{lem:indistinguishable_expectation}
        Let $X$ and $Y$ be random variables supported on $[-\Delta,\Delta]$ satisfying $\pr{}{X \in S} \le e^\varepsilon \cdot \pr{}{Y \in S} + \delta$ and $\pr{}{Y \in S} \le e^\varepsilon \cdot \pr{}{X \in S} + \delta$ for all measurable $S$. Then \[ \ex{}{X} \le \ex{}{Y}  + (e^\varepsilon - 1) \ex{}{|Y|} + 2\delta\Delta.\]
    \end{lemma}
    \begin{proof}
        \begin{align*}
            \ex{}{X} &= \int_0^\Delta \pr{}{X > t} - \pr{}{X < -t} \mathrm{d}t\\
            &\le \int_0^\Delta e^\varepsilon \cdot \pr{}{Y > t} + \delta - e^{-\varepsilon} \cdot (\pr{}{Y < -t} - \delta) \mathrm{d}t\\
            &= \int_0^\Delta ( \pr{}{Y > t} - \pr{}{Y < -t} ) + (e^\varepsilon-1) \cdot \pr{}{Y > t} + (1 - e^{-\varepsilon}) \cdot \pr{}{Y < -t} + (1+e^{-\varepsilon})\delta \mathrm{d}t\\
            &= \ex{}{Y} + (e^\varepsilon -1) \ex{}{\max\{0,Y\}} + (1-e^{-\varepsilon})\ex{}{\max\{0,-Y\}} + (1+e^{-\varepsilon})\delta\Delta\\
            &\le \ex{}{Y} + (e^{\varepsilon}-1) \ex{}{|Y|} + 2\delta\Delta,
        \end{align*}
        as $1-e^{-\varepsilon} \le e^\varepsilon -1$ and $1+e^{-\varepsilon} \le 2$.
    \end{proof}
    
    \begin{remark}
        The only part of the proof of Theorem \ref{thm:lowerbound} that uses differential privacy is Lemma \ref{lem:indistinguishable_expectation}.
        Thus, if we were to consider a different definition of differential privacy, as long as an analog of Lemma \ref{lem:indistinguishable_expectation} holds for this alternative definition, an analog of Theorem \ref{thm:lowerbound} would still apply. That is to say, this negative result is robust to our choice of privacy definition (unlike the the negative result in Section \ref{sec:basic_optimal}).
    \end{remark}

\section{Privacy Amplification by Subsampling}
    Thus far we have considered the composition of Gaussian mechanisms, and generic mechanisms satisfying pure or approximate DP.
    We now turn our attention to subsampled privacy mechanisms. These mechanisms introduce some additional quirks into the picture, which will force us to develop new tools.
    
    The premise of privacy amplification by subsampling is that we run a DP algorithm on some random subset of the data.
    The subset introduces additional uncertainty, which benefits privacy.
    In particular, there is some probability that your data is not included in the analysis, which can only enhance your privacy.
    Furthermore, a potential attacker does not know whether or not your data was dropped; this uncertainty can benefit your privacy even when your data is included.
    Privacy amplification by subsampling theorems make this intuition precise.
    
    Subsampling arises naturally. We often would like to collect the data of the entire population, but this is impractical. Thus we collect the data of a subset of the population and use statistical methods to generalize from this sample to the entire population.
    In particular, in deep learning applications, we will use stochastic gradient descent. That is, we choose a random subset of our training data (called a minibatch) and compute the gradient of the loss function with respect to this subset, rather than the entire dataset. This method reduces the computational cost for training.
    If we want to make deep learning differentially private, then we will add noise to the gradients and we should exploit privacy amplification by subsampling to analyze the privacy properties of this algorithm.
    
    In this section we will analyze subsampling precisely and we will show how it interacts with composition.

    \subsection{Subsampling for Pure or Approximate DP}
    
    We begin by analyzing privacy amplification by subsampling under pure or approximate differential privacy. This is a relatively simple result, but it will be instructive as we later attempt to derive more precise bounds.

    \begin{theorem}[Privacy Amplification by Subsampling for Approximate DP]\label{thm:subsampling_adp} 
        Let $U \subset [n]$ be a random subset.         
        For a dataset $x \in \mathcal{X}^n$, let $x_U \in \mathcal{X}^n$ denote the entries of $x$ indexed by $U$. That is, $(x_U)_i=x_i$ if $i \in U$ and $(x_U)_i=\bot$ if $i \notin U$, where $\bot \in \mathcal{X}$ is some null value.
    
        Assume that, for all $i \in [n]$, we can define $U_{-i} \subset [n] \setminus \{i\}$ such that the following two conditions hold.
        \begin{itemize}
            \item For all $x \in \mathcal{X}^n$ and $i \in [d]$, $x_U$ and $x_{U_{-i}}$ are always neighbouring datasets.
            \item For all $i \in [n]$, the marginal distribution of $U_{-i}$ conditioned on $i \in U$ is equal to the marginal distribution of $U$ conditioned on $i \notin U$.
        \end{itemize}

        Let $M : \mathcal{X}^{n} \to \mathcal{Y}$ satisfy $(\varepsilon,\delta)$-DP. 
        Define $M^U : \mathcal{X}^n \to \mathcal{Y}$ by $M^U(x) = M(x_U)$.
        
        Let $p = \max_{i \in [n]} \pr{U}{i \in U}$.
        Then $M^U$ is $(\varepsilon',\delta')$-DP for $\varepsilon' = \log(1+p(e^\varepsilon-1))$ and $\delta'=p\cdot\delta$.
    \end{theorem}
    For small values of $\varepsilon$, we have $\varepsilon' = \log(1+p(e^\varepsilon-1)) \approx p \cdot \varepsilon$. {More precisely, $\varepsilon' = \log(1+p(e^\varepsilon-1)) \le p \cdot (e^\varepsilon - 1)$ and, for $\varepsilon \le 1$, we have $e^\varepsilon-1 \le \varepsilon + \varepsilon^2 \le 2\varepsilon$.}
    
    The technical assumption about $U$ in the theorem statement is satisfied by many natural subsampling distributions: If $U$ is a uniformly random subset of $[n]$ of a fixed size $m$, then $U_{-i}$ can be obtained by replacing $i$ with a a uniformly random element that is not in $U$. If $U$ is Poisson subsamplied -- i.e., each $i \in [n]$ is independently included in $U$ with probability $p$ -- then, by independence, we can simply remove $i$, namely $U_{-i} = U \setminus\{i\}$.
    
    The technical assumption should be thought of as an independence assumption. For example, it rules out distributions of the form $\pr{U}{U=[n]}=p$ and $\pr{}{U=\emptyset}=1-p$, which do not yield meaningful privacy amplification.

   \begin{proof}[Proof of Theorem \ref{thm:subsampling_adp}.]
    Fix neighbouring inputs $x,x'\in\mathcal{X}^n$ and some measurable $S \subset \mathcal{Y}$. Let $i \in [n]$ be the index on which they differ (i.e., $x_j=x'_j$ for all $j \in [n] \setminus\{i\}$) and let $p_i = \pr{U}{i \in U}$. We have
    \begin{align*}
        \pr{M^U}{M^U(x) \in S} &= \ex{U}{\pr{M}{M(x_U) \in S}}\\
        &= (1-p_i) \cdot \ex{U}{\pr{M}{M(x_U) \in S} \mid i \notin U} + p_i \cdot \ex{U}{\pr{M}{M(x_U) \in S} \mid i \in U}\\
        &= (1-p_i) \cdot \ex{U}{\pr{M}{M(x'_U) \in S} \mid i \notin U} + p_i \cdot \ex{U}{\pr{M}{M(x_U) \in S} \mid i \in U}\\
        &= (1-p_i) \cdot a + p_i \cdot b,\\
        \pr{M^U}{M^U(x') \in S} &= (1-p_i) \cdot \ex{U}{\pr{M}{M(x'_U) \in S} \mid i \notin U} + p_i \cdot \ex{U}{\pr{M}{M(x'_U) \in S} \mid i \in U}\\
        &= (1-p_i) \cdot a + p_i \cdot b',
    \end{align*}
    where $a=\ex{U}{\pr{M}{M(x_U) \in S} \mid i \notin U}=\ex{U}{\pr{M}{M(x'_U) \in S} \mid i \notin U}$, $b=\ex{U}{\pr{M}{M(x_U) \in S} \mid i \in U}$, and $b'=\ex{U}{\pr{M}{M(x'_U) \in S} \mid i \in U}$.
    
    Note that $x_U$ and $x'_U$ are always neighbouring datasets. And, if $i \notin U$, then $x_U = x'_U$.
    Since $M$ is $(\varepsilon,\delta)$-DP, we have $\pr{M}{M(x_U) \in S} \le e^\varepsilon \cdot \pr{M}{M(x'_U) \in S} + \delta$ for all values of $U$; thus \[b = \ex{U}{\pr{M}{M(x_U) \in S} \mid i \in U} \le \ex{U}{e^\varepsilon\cdot\pr{M}{M(x'_U) \in S} + \delta \mid i \in U} = e^\varepsilon \cdot b' + \delta.\]
    However, this inequality alone is not sufficient to prove the claim. We also need to show that $b \le e^\varepsilon \cdot a + \delta$.
    Using our technical assumption, we have
    \begin{align*}
        b &= \ex{U}{\pr{M}{M(x_U) \in S} \mid i \in U}\\
        &\le \ex{U}{e^\varepsilon \cdot \pr{M}{M(x_{U_{-i}}) \in S} + \delta \mid i \in U} \tag{$x_{U_{-i}}$ is a neighbour of $x_U$} \\
        &= \ex{U}{e^\varepsilon \cdot \pr{M}{M(x_U) \in S} + \delta \mid i \notin U} \tag{$U_{-i}|_{i \in U}$ has the same distribution as $U|_{i \notin U}$}\\
        &= e^\varepsilon \cdot a + \delta.
    \end{align*}
    Now we can complete the proof: For any $\lambda \in [0,1]$,
    \begin{align*}
        \pr{M^U}{M^U(x) \in S} 
        &= (1-p_i) \cdot a + p_i \cdot b\\
        &\le (1-p_i) \cdot a + p_i \cdot ((1-\lambda) \cdot (e^\varepsilon \cdot a + \delta) + \lambda \cdot (e^\varepsilon \cdot b' + \delta)) \\
        &= (1-p_i+e^\varepsilon \cdot (1-\lambda)\cdot p_i) \cdot a + p_i \cdot e^\varepsilon \cdot \lambda \cdot b' + p_i \cdot \delta.
    \end{align*} 
    Set $\lambda = p_i + (1-p_i) \cdot e^{-\varepsilon}$ to obtain
    \begin{align*}
        \pr{M^U}{M^U(x) \in S} 
        &\le  (1-p_i+e^\varepsilon \cdot (1-\lambda)\cdot p_i) \cdot a + p_i \cdot e^\varepsilon \cdot \lambda \cdot b' + p_i\cdot\delta\\
        &= \left( 1 + p_i \cdot (e^\varepsilon -1) \right) \cdot \left( (1-p_i) \cdot a + p_i \cdot b' \right) + p_i \cdot \delta\\
        &= \left( 1 + p_i \cdot (e^\varepsilon -1) \right) \cdot \pr{M^U}{M_U(x') \in S} + p_i \cdot \delta\\
        &\le e^{\varepsilon'} \cdot \pr{M^U}{M_U(x') \in S} + \delta'.
    \end{align*} 
    \end{proof}
    
    Theorem \ref{thm:subsampling_adp} is tight:  Consider an algorithm $M : \{0,1,\bot\}^{n} \to \{0,1\}$ that sums its input (excluding $\bot$ values) and performs randomized response on whether or not the sum is 0.\footnote{Alternatively (and equivalently), consider an algorithm that adds discrete Laplace noise to the sum of its non-null inputs.} That is, if $y \in \{0,1,\bot\}^{n}$ satisfies $\sum_{i : y_i \ne \bot} y_i = 0$, then $\pr{}{M(y)=0}=\frac{e^\varepsilon}{e^\varepsilon=1}$ and $\pr{}{M(y)=1}=\frac{1}{e^\varepsilon=1}$ and, if $y \in \{0,1, \bot\}^{ n}$ satisfies $\sum_{i : y_i \ne \bot} y_i > 0$, then $\pr{}{M(y)=1}=\frac{e^\varepsilon}{e^\varepsilon=1}$ and $\pr{}{M(y)=0}=\frac{1}{e^\varepsilon-1}$.
    This algorithm satisfies $\varepsilon$-DP.
    
    Let $U \subset [n]$ be random and let $M^U : \{0,1\}^n \to \{0,1\}$ be as in Theorem \ref{thm:subsampling_adp}.
    Consider neighbouring datasets $x=(0,0,\cdots,0)$ and $x'=(1,0,0,\cdots,0)$.
    We have 
    \begin{align*}
        \pr{}{M^U(x)=0} &= \frac{e^\varepsilon}{e^\varepsilon+1},\\
        \pr{}{M^U(x)=1} &= \frac{1}{e^\varepsilon+1},\\
        \pr{}{M^U(x')=1} &= \frac{\pr{}{1 \in U} \cdot e^\varepsilon + \pr{}{1 \notin U}}{e^\varepsilon+1},\\
        \pr{}{M^U(x')=0} &= \frac{\pr{}{1 \in U} + \pr{}{1 \notin U} \cdot e^\varepsilon}{e^\varepsilon+1},\\
        e^{\varepsilon'}\ge\frac{\pr{}{M^U(x')=1}}{\pr{}{M^U(x)=1}} &= 1 + \pr{}{1 \in U} \cdot (e^\varepsilon -1),
    \end{align*}
    where $\varepsilon'$ is the pure DP parameter satisfied by $M^U$. We can assume without loss of generality that $p = \max_i \pr{}{i \in U} = \pr{}{1 \in U}$. Thus this bound matches the guarantee of Theorem \ref{thm:subsampling_adp}. (This example can be extended to approximate DP too.)
    
    \subsection{Addition or Removal versus Replacement for Neighbouring Datasets}\label{sec:addremovereplace}
    
    For this discussion of subsampling, we need to be careful about what it means for datasets to be neighbouring.
    There are three common definitions of what qualifies as neighbouring datasets: (i) addition or removal of one person's data, (ii) replacement of one person's data, or (iii) both.
    Each of these three options is a reasonable choice. Work on differential privacy often glosses over this choice -- often the choice is irrelevant. But it becomes relevant if we want sharp analyses of privacy amplification by subsampling.
    
    For the discussion of composition so far in this chapter, it does not matter at all how we define neighbouring datasets, as long as we are consistent.
    In general, it only matters slightly which we choose: A replacement can be accomplished by a combination of a removal and an addition. Thus, by group privacy, if the algorithm is $(\varepsilon,\delta)$-DP with respect to addition or removal, then it is $(2\varepsilon,(1+e^\varepsilon)\delta)$-DP with respect to replacement. Conversely, we can simulate a removal or addition by replacing the record with a ``null'' value ($\bot$ in the formalism of Theorem \ref{thm:subsampling_adp}). Thus DP with respect to replacement entails DP with respect to addition or removal with the same parameters, unless the semantics of the algorithm forbids null values.
    
    This subtlety already arises in Theorem \ref{thm:subsampling_adp}. Let's take a close look at the technical assumption: Theorem \ref{thm:subsampling_adp} assumes that, for all $i \in [n]$, we can define $U_{-i} \subset [n] \setminus \{i\}$ such that the following two conditions hold.
    \begin{itemize}
        \item For all $x \in \mathcal{X}^n$ and $i \in [d]$, $x_U$ and $x_{U_{-i}}$ are always neighbouring datasets.
        \item For all $i \in [n]$, the marginal distribution of $U_{-i}$ conditioned on $i \in U$ is equal to the marginal distribution of $U$ conditioned on $i \notin U$.
    \end{itemize}
    Suppose $U \subset [n]$ is a uniformly random subset of a fixed size $|U|=m$. Then we would define $U_{-i}$ to be $U$ with $i$ replaced by a uniformly random element that is not already in $U$. Thus, for $x_U$ and $x_{U_{-i}}$ to be neighbouring datasets, our neighbouring relation must allow replacement, not just addition or removal.
    
    However, if $U$ corresponds to Poisson subsampling (i.e., each $i \in [n]$ is included in $U$ independently with probability $p$), then $U_{-i}$ would just correspond to removing $i$. In that case, for $x_U$ and $x_{U_{-i}}$ to be neighbouring datasets, our neighbouring relation must allow addition and removal.
    
    It turns out to be easier to work with Poisson subsampling and assuming the neighbouring relation is addition or removal. In this case, the proof of Theorem \ref{thm:subsampling_adp} simplifies to the following. 
    
    \begin{proof}[Proof of Theorem \ref{thm:subsampling_adp} for the special case of Poisson sampling and addition or removal.]
            Let $U \subset [n]$ independently include each element with probability $p$.
            Let $x,x'\in\mathcal{X}^n$ be neighbouring datasets in terms of addition or removal. Without loss of generality, assume $x'$ is $x$ with $x_i$ removed (or, rather, replaced by $x'_i=\bot$).\footnote{To be formal, we assume $\bot \in \mathcal{X}$ is a null value that is equivalent to removing the item. In particular, for $x \in \mathcal{X}^n$ and $U \subset [n]$ we can define $x_U \in \mathcal{X}^n$ such that $(x_U)_i = x_i$ if $i \in U$ and $(x_U)_i = \bot$ if $i \in [n] \setminus U$.} For any measurable $S \subset \mathcal{Y}$,
            \begin{align*}
                \pr{M^U}{M^U(x) \in S} &= \ex{U}{\pr{M}{M(x_U) \in S}}\\
                &= (1-p) \cdot \ex{U}{\pr{M}{M(x_U) \in S} \mid i \notin U} + p \cdot \ex{U}{\pr{M}{M(x_U) \in S} \mid i \in U}\\
                &= (1-p) \cdot \ex{U}{\pr{M}{M(x'_U) \in S} \mid i \notin U} + p \cdot \ex{U}{\pr{M}{M(x_U) \in S} \mid i \in U}\\
                &= (1-p) \cdot \ex{U}{\pr{M}{M(x'_U) \in S}} + p \cdot \ex{U}{\pr{M}{M(x_U) \in S} \mid i \in U}\\
                &\le (1-p) \cdot \ex{U}{\pr{M}{M(x'_U) \in S}} + p \cdot \ex{U}{e^\varepsilon \cdot \pr{M}{M(x'_U) \in S} + \delta \mid i \in U}\\
                &= (1-p) \cdot \ex{U}{\pr{M}{M(x'_U) \in S}} + p \cdot \ex{U}{e^\varepsilon \cdot \pr{M}{M(x'_U) \in S} + \delta}\\
                &= (1-p+p\cdot e^\varepsilon) \cdot  \pr{M^U}{M^U(x') \in S} + p\cdot\delta
            \end{align*}
            and, by the same calculation,
            \begin{align*}
                \pr{M^U}{M^U(x) \in S} 
                &= (1-p) \cdot \ex{U}{\pr{M}{M(x'_U) \in S}} + p \cdot \ex{U}{\pr{M}{M(x_U) \in S} \mid i \in U}\\
                &\ge (1-p) \cdot \ex{U}{\pr{M}{M(x'_U) \in S}} + p \cdot \ex{U}{e^{-\varepsilon} \cdot (\pr{M}{M(x'_U) \in S}-\delta) \mid i \in U}\\
                &= (1-p+p\cdot e^{-\varepsilon}) \cdot  \pr{M^U}{M^U(x') \in S} - p \cdot e^{-\varepsilon} \cdot \delta\\
                &\ge \frac{1}{1-p+p \cdot e^\varepsilon} \cdot(\pr{M^U}{M^U(x')\in S}-p\cdot\delta).\tag{Lemma \ref{lem:frac}}
            \end{align*}
            The key step in the proof is the equality $\ex{U}{\pr{M}{M(x'_U) \in S} \mid i \notin U} = \ex{U}{\pr{M}{M(x'_U) \in S} \mid i \in U} = \ex{U}{\pr{M}{M(x'_U) \in S} }$. This holds because $x'_i=\bot$, so whether or not $i \in U$ is irrelevant for $x'_U$, and because the event $i \in U$ is independent from $U \setminus \{i\}$. 
    \end{proof}
    
    For the rest of this section, we will restrict our attention to Poisson subsampling and assume that the neighbouring relation corresponds to addition or removal of one individual's data.

\subsection{Subsampling \& Composition}\label{sec:subsamp_composition}

    How does composition work with subsampling? Of course, we can combine the advanced composition theorem (Theorem \ref{thm:advancedcomposition_approx}) with our privacy amplification by subsampling result (Theorem \ref{thm:subsampling_adp}). However, it turns out this is not the best way to analyze many realistic systems.
    
    Consider the following algorithm (which arises in differentially private deep learning applications).
    Let $x \in \mathcal{X}^n$ be the private input.
    Iteratively, for $t=1,\cdots,T$, we pick some function $q_t : \mathcal{X}^n \to \mathbb{R}^d$ and randomly sample a subset $U_t \subset [n]$; then we reveal $\mathcal{N}(q_t(x_{U_t}),\sigma^2 I_d)$.
    
    This algorithm interleaves composition with privacy amplification by subsampling. That is, we combine multivariate Gaussian noise addition (which is a form of composition over the $d$ coordinates) with subsampling and then we compose over the $T$ iterations.
    
    We can use Corollary \ref{cor:gauss_adp_exact_multi} to show that releasing $\mathcal{N}(q_t(x),\sigma^2 I_d)$ satisfies $(\varepsilon_0,\delta_0)$-DP for $\varepsilon_0 = O\left(\sqrt{\frac{\Delta_2^2}{\sigma^2} \log(1/\delta_0)}\right)$, where $\Delta_2 = \sup_{x,x'\in\mathcal{X}^n \atop \text{neighbouring}} \|q_t(x)-q_t(x')\|_2$ is the sensitivity of $q_t$.
    Then we can use Theorem \ref{thm:subsampling_adp} to show that, if $U_t$ is a Poisson sample which contains each element with probability $p$, then $\mathcal{N}(q_t(x_{U_t}),\sigma^2 I_d)$ is $(\varepsilon_1,\delta_1)$-DP where $\varepsilon_1 = \log(1+p\cdot(e^{\varepsilon_0}-1)) = O(p \cdot \varepsilon_0)$ and $\delta_1 = p\cdot\delta_0$.
    Finally, Theorem \ref{thm:advancedcomposition_approx} tells us that the composition over $T$ iterations satisfies $(\varepsilon,\delta)$-DP with $\varepsilon = O\left(\varepsilon_1 \cdot \sqrt{T \log(1/\delta_2)})\right)$ and $\delta= \delta_2 + T \cdot \delta_1$. Overall, we have \[\varepsilon = O\left(\frac{\Delta_2}{\sigma} \cdot p \cdot \sqrt{T} \cdot \log(T/\delta)\right).\]
    This result is asymptotically suboptimal because we have picked up two $\sqrt{\log(1/\delta)}$ terms. We obtained one from the Gaussian noise addition (Corollary \ref{cor:gauss_adp_exact_multi}) and another from the composition (Theorem \ref{thm:advancedcomposition_approx}). Both arise from bounding the tails of the privacy loss distribution. This is redundant; we should only need to bound the tails of the privacy loss distribution once.
    
    Intuitively, we started with a Gaussian privacy loss; then we applied a tail bound to obtain a $(\varepsilon_0,\delta_0)$-DP guarantee to which we applied the subsampling theorem; and then we converted this back into a concentrated DP guarantee to apply advanced composition and finally we applied a tail bound to convert this back to $(\varepsilon,\delta)$-DP.
    
    We are going to avoid this redundancy by analyzing privacy amplification by subsampling directly in terms of the privacy loss distribution, rather than needing to go via approximate DP. To do so, we need to introduce a new tool.
    
\subsection{R\'enyi Differential Privacy}

    R\'enyi differential privacy was introduced by \citet{mironov2017renyi} and was motivated by analyzing privacy amplification by subsampling interleaved with composition, which arises in differentially private deep learning \citep{abadi2016deep}.
    
    \begin{definition}[R\'enyi Differential Privacy]\label{defn:rdp}
        Let $M : \mathcal{X}^n \to \mathcal{Y}$ be a randomized algorithm. We say that $M$ satisfies $(\alpha,\varepsilon)$-R\'enyi differential privacy ($(\alpha,\varepsilon)$-RDP) if, for all neighbouring inputs $x,x'\in\mathcal{X}^n$, the privacy loss distribution $\privloss{M(x)}{M(x')}$ is well-defined (see Definition \ref{defn:priv_loss}) and \[\ex{Z \gets \privloss{M(x)}{M(x')}}{\exp((\alpha-1)Z)} \le \exp((\alpha-1)\cdot\varepsilon).\] 
    \end{definition}
    
    R\'enyi DP is closely related to concentrated DP (Definition \ref{defn:cdp}).
    Specifically, $\rho$-zCDP is equivalent to satisfying $(\alpha,\alpha\cdot\rho)$-RDP for all $\alpha \in (1,\infty)$.
    R\'enyi DP inherits the nice composition properties of concentrated DP:
    
    \begin{lemma}\label{lem:rdp_composition}
        Let $M_1 : \mathcal{X}^n \to \mathcal{Y}_1$ be $(\alpha,\varepsilon_1)$-RDP. Let $M_2 : \mathcal{X}^n \times \mathcal{Y}_1 \to \mathcal{Y}_2$ be such that, for all $y_1 \in \mathcal{Y}_1$, the algorithm $x \mapsto M(x,y_1)$ is $(\alpha,\varepsilon_2)$-RDP.
        Define $M : \mathcal{X}^n \to \mathcal{Y}_2$ by $M(x) = M_2(x,M_1(x))$. Then $M$ is $(\alpha,\varepsilon_1+\varepsilon_2)$-RDP.
    \end{lemma}
    The proof of Lemma \ref{lem:rdp_composition} is identical to that of Proposition \ref{prop:adaptivecomp}. Note that, while the $\varepsilon$ parameter adds up, the $\alpha$ parameter does not change. More generally, composing an $(\alpha_1,\varepsilon_1)$-RDP algorithm with an $(\alpha_2,\varepsilon_2)$-RDP algorithm yields $(\min\{\alpha_1,\alpha_2\},\varepsilon_1+\varepsilon_2)$-RDP.
    
    It is helpful to think of $\varepsilon$ in $(\alpha,\varepsilon)$-RDP as a function of $\alpha$, rather than a single number.
    This function can encode a rich variety of privacy guarantees. (Concentrated DP corresponds to a linear function.) In particular, it allows us to more precisely represent the kinds of guarantees obtained by subsampling.
    
    Concentrated DP corresponds to the privacy loss being subgaussian (i.e., $\rho$-zCDP implies $\pr{Z \gets \privloss{M(x)}{M(x')}}{Z>\tilde\varepsilon} \le \exp(-(\tilde\varepsilon-\rho)^2/4\rho)$ for all $\tilde\varepsilon \ge \rho$ and all neighbouring inputs $x$ and $x'$), whereas R\'enyi DP corresponds to the privacy loss being subexponential (i.e., $(\alpha,\varepsilon)$-RDP implies $\pr{Z \gets \privloss{M(x)}{M(x')}}{Z>\tilde\varepsilon} \le \exp(-(\alpha-1)(\tilde\varepsilon-\varepsilon))$). That is, R\'enyi DP is more appropriate for analyzing privacy loss distributions with slightly heavier tails than Gaussian.
    In contrast, pure DP corresponds to the privacy loss being bounded (i.e., $\varepsilon$-DP implies $\pr{Z \gets \privloss{M(x)}{M(x')}}{Z>\varepsilon}=0$). So we can view concentrated DP as a relaxation of pure DP and, in turn, R\'enyi DP is a relaxation of concentrated DP.
    
    R\'enyi DP is typically formulated in terms of R\'enyi divergences \citep{renyi1961measures}, which were studied in the information theory literature long before differential privacy was discovered.

    \begin{definition}[R\'enyi Divergences]
        Let $P$ and $Q$ be distributions over $\mathcal{Y}$.\footnote{We make the usual measure theoretic disclaimers: We assume the $P$ and $Q$ have the same sigma-algebra. We assume $P$ is absolutely continuous with respect to $Q$ -- i.e., $\forall S ~~ Q(S)=0 \implies P(S)=0$ -- so that the Radon-Nikodym derivative is well-defined; we denote the Radon-Nikodym derivative of $P$ with respect to $Q$ evaluated at $y$ by $P(y)/Q(y)$. More generally, if the absolute continuity assumption does not hold, then we define $\dr{\alpha}{P}{Q}=\infty$ for all $\alpha \in [1,\infty]$.}
        For $\alpha \in (1,\infty)$, define
        \begin{align*}
            \dr{\alpha}{P}{Q} &= \frac{1}{\alpha-1}\log\ex{Z \gets \privloss{P}{Q}}{\exp((\alpha-1)Z)}\\
            &= \frac{1}{\alpha-1} \log \ex{Y \gets P}{\left(\frac{P(Y)}{Q(Y)}\right)^{\alpha-1}}\\
            &= \frac{1}{\alpha-1} \log \ex{Y \gets Q}{\left(\frac{P(Y)}{Q(Y)}\right)^{\alpha}}.
        \end{align*}
        Also, define
        \begin{align*}
            \dr{1}{P}{Q} &= \lim_{\alpha \to 1} \dr{\alpha}{P}{Q} \\
            &= \ex{Z \gets \privloss{P}{Q}}{Z}\\
            &= \ex{Y \gets P}{\log\left(\frac{P(Y)}{Q(Y)}\right)},\\
            \dr{\infty}{P}{Q} &= \lim_{\alpha \to \infty} \dr{\alpha}{P}{Q}\\
            &= \sup\left\{\log\left(\frac{P(S)}{Q(S)}\right) : S \subset \mathcal{Y} , Q(S) > 0\right\}.
        \end{align*}
    \end{definition}
    
    Thus an equivalent definition of $M$ satisfying $(\alpha,\varepsilon)$-RDP is that $\dr{\alpha}{M(x)}{M(x')} \le \varepsilon$ for all neighbouring $x,x'$.
    
    We now state several key properties of R\'enyi divergences; most of these are properties we have proved earlier, but we now restate them in a new language.
    
    \begin{lemma}\label{lem:rdp_properties}
        Let $P,Q$ be probability distributions over $\mathcal{Y}$ with a common sigma-algebra such that $P$ is absolutely continuous with respect to $Q$.
        \begin{enumerate}
        
            \item \textbf{Postprocessing (a.k.a.~data processing inequality) \& non-negativity:} Let $f : \mathcal{Y} \to \mathcal{Z}$ be a measurable function. Let $f(P)$ denote the distribution on $\mathcal{Z}$ obtained by applying $f$ to a sample from $P$; define $f(Q)$ similarly.
            Then $0 \le \dr{\alpha}{f(P)}{f(Q)} \le \dr{\alpha}{P}{Q}$ for all $\alpha \in [1,\infty]$.
        
            \item \textbf{Composition:} If $P=P' \times P''$ and $Q=Q' \times Q''$ are product distributions, then $\dr{\alpha}{P}{Q} = \dr{\alpha}{P'}{Q'} + \dr{\alpha}{P''}{Q''}$ for all $\alpha \in [1,\infty]$.
            
            More generally, suppose $P$ and $Q$ are distributions on $\mathcal{Y} = \mathcal{Y}' \times \mathcal{Y}''$. Let $P'$ and $Q'$ be the marginal distributions on $\mathcal{Y}'$ induced by $P$ and $Q$ respectively. For $y' \in \mathcal{Y}'$, let $P''_{y'}$ and $Q''_{y'}$ be the conditional distributions on $\mathcal{Y}''$ induced by $P$ and $Q$ respectively. That is, we can generate a sample $Y=(Y',Y'') \gets P$ by first sampling $Y' \gets P'$ and then sampling $Y'' \gets P''_{Y'}$ and similarly for $Q$. Then $\dr{\alpha}{P}{Q} \le \dr{\alpha}{P'}{Q'} + \sup_{y' \in \mathcal{Y}'} \dr{\alpha}{P''_{y'}}{Q''_{y'}}$ for all $\alpha \in [1,\infty]$.

            \item \textbf{Monotonicity:} For all $1 \le \alpha \le \alpha' \le \infty$, $\dr{\alpha}{P}{Q} \le \dr{\alpha'}{P}{Q}$.

            \item \textbf{Gaussian divergence:} For all $\mu,\mu',\sigma\in\mathbb{R}$ with $\sigma>0$ and all $\alpha \in [1,\infty)$, \[\dr{\alpha}{\mathcal{N}(\mu,\sigma^2)}{\mathcal{N}(\mu',\sigma^2)} = \alpha \cdot \frac{(\mu-\mu')^2}{2\sigma^2}.\]
            
            \item \textbf{Pure DP to Concentrated DP:} For all $\alpha \in [1,\infty)$, \[\dr{\alpha}{P}{Q} \le \frac{\alpha}{8}\cdot\left(\dr{\infty}{P}{Q} + \dr{\infty}{Q}{P}\right)^2.\]
            
            \item \textbf{Quasi-convexity:} Let $P'$ and $Q'$ be probability distributions over $\mathcal{Y}$ such that $P'$ is absolutely continuous with respect to $Q'$. For $s \in [0,1]$, let $(1-s) \cdot P + s \cdot P'$ denote the convex combination of the distributions $P$ and $P'$ with weighting $s$.
            For all $\alpha \in (1,\infty)$ and all $s \in [0,1]$,
            \begin{align*}
                &\dr{\alpha}{(1-s) \cdot P + s \cdot P'}{(1-s) \cdot Q + s \cdot Q'} \\
                &~~\le \frac{1}{\alpha-1} \log \big( (1-s) \cdot \exp\left((\alpha-1) \dr{\alpha}{P}{Q} \right) + s \cdot \exp\left((\alpha-1) \dr{\alpha}{P'}{Q'} \right) \big) \\
                &~~\le \max\big\{ \dr{\alpha}{P}{Q} , \dr{\alpha}{P'}{Q'} \big\}
            \end{align*}
            and $\dr{1}{(1-s) \cdot P + s \cdot P'}{(1-s) \cdot Q + s \cdot Q'} \le (1-s) \cdot \dr{1}{P}{Q} + s \cdot \dr{1}{P'}{Q'}$.
            
            \item \textbf{Triangle-like inequality (a.k.a.~group privacy):} Let $R$ be a distribution on $\mathcal{Y}$ and assume that $Q$ is absolutely continuous with respect to $R$. For all $1 < \alpha < \alpha' < \infty$, \[\dr{\alpha}{P}{R} \le \frac{\alpha'}{\alpha'-1}\cdot \dr{\alpha \cdot \frac{\alpha'-1}{\alpha'-\alpha}}{P}{Q} + \dr{\alpha'}{Q}{R}.\]
            In particular, if $\dr{\alpha}{P}{Q} \le \rho_1 \cdot \alpha$ and $\dr{\alpha}{Q}{R} \le \rho_2 \cdot \alpha$ for all $\alpha \in (1,\infty)$, then $\dr{\alpha}{P}{R} \le (\sqrt{\rho_1}+\sqrt{\rho_2})^2 \cdot \alpha$ for all $\alpha \in (1,\infty)$.
            
            \item \textbf{Conversion to approximate DP:} For all measurable $S \subset \mathcal{Y}$, all $\alpha \in (1,\infty)$, and all $\tilde{\varepsilon}\ge\dr{\alpha}{P}{Q}$, \[P(S) \le e^{\tilde{\varepsilon}} \cdot Q(S) + e^{-(\alpha-1)(\tilde\varepsilon-\dr{\alpha}{P}{Q})} \cdot \frac1\alpha \cdot \left( 1- \frac1\alpha\right)^{\alpha-1} \le e^{\tilde{\varepsilon}} \cdot Q(S) + e^{-(\alpha-1)(\tilde\varepsilon-\dr{\alpha}{P}{Q})}.\]
        \end{enumerate}
    \end{lemma}
    \begin{proof}
        \begin{enumerate}
        
            \item \textit{Postprocessing (a.k.a.~data processing inequality) \& non-negativity:} See Lemma \ref{lem:postprocessing}. Non-negativity follows from setting $f$ to be a constant function and noting that the divergence between two point masses is zero.
        
            \item \textit{Composition:} See Proposition \ref{prop:adaptivecomp}.
            
            \item \textit{Monotonicity:} Let $1 < \alpha \le \alpha' < \infty$. (The cases where $\alpha=1$ and $\alpha'=\infty$ follow from continuity.) Let $f(x) = x^{\frac{\alpha'-1}{\alpha-1}}$. Then $f$ is convex and, by Jensen's inequality, \[e^{(\alpha'-1)\dr{\alpha}{P}{Q}} = f\left(\ex{Y \gets P}{\left(\frac{P(Y)}{Q(Y)}\right)^{\alpha-1}}\right) \le \ex{Y \gets P}{f\left(\left(\frac{P(Y)}{Q(Y)}\right)^{\alpha-1}\right)} = e^{(\alpha'-1)\dr{\alpha'}{P}{Q}},\] which implies $\dr{\alpha}{P}{Q} \le \dr{\alpha'}{P}{Q}$.

            \item \textit{Gaussian divergence:} See Lemma \ref{lem:gauss_cdp}.
            
            \item \textit{Pure DP to Concentrated DP:} See Proposition \ref{prop:pdp2cdp}.
            
            \item \textit{Quasi-convexity:} See Lemma B.6 of \citet{bun2016concentrated}.
            
            \item \textit{Triangle-like inequality (a.k.a.~group privacy):} 
            Let $\alpha \in (1,\infty)$.
            Let $p,q\in(1,\infty)$ satisfy $\frac1p+\frac1q=1$. By H\"older's inequality,
            \begin{align*}
                e^{(\alpha-1)\dr{\alpha}{P}{R}} &= \ex{Y \gets P}{\left(\frac{P(Y)}{R(Y)}\right)^{\alpha-1}} = \ex{Y \gets R}{\left(\frac{P(Y)}{R(Y)}\right)^\alpha}\\
                 &= \ex{Y \gets Q}{\frac{P(Y)}{Q(Y)} \cdot \left(\frac{P(Y)}{Q(Y)} \cdot \frac{Q(Y)}{R(Y)}\right)^{\alpha-1}}\\
                 &= \ex{Y \gets Q}{\left(\frac{P(Y)}{Q(Y)}\right)^\alpha \cdot \left(\frac{Q(Y)}{R(Y)}\right)^{\alpha-1}}\\
                 &\le \ex{Y \gets Q}{\left(\frac{P(Y)}{Q(Y)}\right)^{\alpha p}}^{1/p} \cdot \ex{Y \gets Q}{\left(\frac{Q(Y)}{R(Y)}\right)^{(\alpha-1)q}}^{1/q}\\
                 &= e^{\frac{\alpha p - 1}{p}\dr{\alpha p}{P}{Q}} \cdot e^{\frac{(\alpha-1)q}{q}\dr{(\alpha-1)q+1}{Q}{R}}.
            \end{align*}
            This rearranges to 
            \begin{align*}
                \dr{\alpha}{P}{R} &\le \frac{\alpha p - 1}{(\alpha-1)p} \dr{\alpha p}{P}{Q} + \dr{(\alpha-1)q+1}{Q}{R}\\
                &= \left( 1 + \frac{1}{(\alpha-1)q} \right) \cdot \dr{\alpha p}{P}{Q} + \dr{(\alpha-1)q+1}{Q}{R}\\
                &= \frac{\alpha'}{\alpha'-1}\cdot \dr{\alpha \cdot \frac{\alpha'-1}{\alpha'-\alpha}}{P}{Q} + \dr{\alpha'}{Q}{R},
            \end{align*}
            where the final equality sets $p = \frac{\alpha'-1}{\alpha'-\alpha}$ and $q=\frac{\alpha'-1}{\alpha-1}$
            
            Now assume $\dr{\alpha}{P}{Q} \le \rho_1 \cdot \alpha$ and $\dr{\alpha}{Q}{R} \le \rho_2 \cdot \alpha$ for all $\alpha \in (1,\infty)$.
            Then
            \begin{align*}
                \dr{\alpha}{P}{R} &\le \inf_{\alpha'>\alpha} \frac{\alpha'}{\alpha'-1}\cdot \dr{\alpha \cdot \frac{\alpha'-1}{\alpha'-\alpha}}{P}{Q} + \dr{\alpha'}{Q}{R}\\
                &\le \inf_{\alpha'>\alpha} \frac{\alpha'}{\alpha'-1}\cdot \alpha \cdot \frac{\alpha'-1}{\alpha'-\alpha} \cdot \rho_1 + \alpha' \cdot \rho_2\\
                &= \inf_{x>0} \alpha \cdot \frac{x+1}{x} \cdot \rho_1 + \alpha \cdot (x+1) \cdot \rho_2 \tag{Reparameterize $\alpha'=(x+1)\cdot\alpha$}\\
                &= \alpha \cdot \inf_{x>0} \rho_1 +\rho_2 + \frac1x \rho_1 + x \cdot \rho_2 \\
                &= \alpha \cdot (\rho_1 + \rho_2 + 2\sqrt{\rho_1 \cdot \rho_2}) \tag{$x = \sqrt{\rho_1/\rho_2}$}\\
                &= \alpha \cdot (\sqrt{\rho_1} + \sqrt{\rho_2})^2.
            \end{align*}
            
            \item \textit{Conversion to approximate DP:} See Proposition \ref{prop:cdp2adp}.
        \end{enumerate}
    \end{proof}
    
\subsection{Sharp Privacy Amplification by Poisson Subsampling for R\'enyi DP}\label{sec:sharp_rdp_subsampling}
    
    Now we analyze privacy amplification by subsampling under R\'enyi DP. We start with a R\'enyi DP guarantee and we obtain an amplified R\'enyi DP guarantee. The goal is to obtain a sharp analysis that avoids converting to approximate DP and back.
    
    For mathematical simplicity, we restrict our attention to Poisson subsampling and assume that neighbouring datasets correspond to addition or removal of one person's data.
    
    \begin{theorem}[Tight Privacy Amplification by Subsampling for R\'enyi DP]\label{thm:rdp_subsampling}
        Let $U \subset [n]$ be a random set that contains each element independently with probability $p$.
        For $x \in \mathcal{X}^n$ let $x_U \in \mathcal{X}^n$ be given by $(x_U)_i = x_i$ if $i \in U$ and $(x_U)_i = \bot$ if $i \notin U$, where $\bot\in\mathcal{X}$ is some fixed value.
        
        Let $\varepsilon:\mathbb{N}_{\ge 2}\to\mathbb{R} \cup \{\infty\}$ be a function.
        Let $M : \mathcal{X}^n \to \mathcal{Y}$ satisfy $(\alpha,\varepsilon(\alpha))$-RDP for all $\alpha \in \mathbb{N}_{\ge 2}$ with respect to addition or removal -- i.e., $x,x'\in\mathcal{X}^n$ are neighbouring if, for some $i \in [n]$, we have $x_i=\bot$ or $x'_i=\bot$, and $\forall j \ne i ~~ x_j=x'_j$.
        
        Define $M^U : \mathcal{X}^n \to \mathcal{Y}$ by $M^U(x) = M(x_U)$. Then $M^U$ satisfies $(\alpha,\varepsilon'_p(\alpha))$-RDP for all $\alpha \in \mathbb{N}_{\ge 2}$ where \[\varepsilon'_p(\alpha) = \frac{1}{\alpha-1} \log \left( (1-p)^{\alpha-1}(1 + (\alpha-1) p) + \sum_{k=2}^\alpha {\alpha \choose k} (1-p)^{\alpha-k} p^k \cdot e^{(k-1)\varepsilon(k)} \right).\]
    \end{theorem}
    Note that $(1-p)^{\alpha-1}(1 + (\alpha-1) p) \le  1$. It is easy to see from the proof that this analysis is tight. That is, if the assumption that $M$ satisfies $(\alpha,\varepsilon(\alpha))$-RDP for all $\alpha$ is tight for some fixed pair of neighbouring inputs, then the conclusion that $M^U$ satisfies $(\alpha,\varepsilon'_p(\alpha))$-RDP is also tight.

    Theorem \ref{thm:rdp_subsampling} only considers R\'enyi DP with orders $\alpha \in \mathbb{N}_{\ge 2} = \{2,3,4,\cdots\}$. This restriction arises because the proof uses a binomial expansion, which only works for integer exponents. In certain cases, it is possible to obtain an expression all $\alpha \in (1,\infty)$ using an infinite binomial series \citep{mironov2019r}.
    In general, we can use Monotonicity (part 3 of Lemma \ref{lem:rdp_properties}) to bound non-integer $\alpha$, namely for all $\alpha \in (1,\infty)$, $M^U$ satisfies $(\alpha,\varepsilon'_p(\lceil \alpha \rceil))$-RDP.
    
    \begin{proof}[Proof of Theorem \ref{thm:rdp_subsampling}.]
        Fix neighbouring datasets $x,x'\in\mathcal{X}^n$.
        Without loss of generality, assume that $x'$ is $x$ with one element removed -- i.e., $\exists i \in [n] ~~ (x'_i=\bot) \wedge (\forall j \in [n] \setminus\{i\} ~~ x_j=x'_j)$. Fix this $i$.

        Let $Q = M(x'_U) = M^U(x')$.
        Let $P = M(x_U)|_{i \in U}$ be the conditional distribution of $M(x_U)$ with $i \in U$.
        Note that $M(x_U)|_{i \notin U} = Q$ because $x_U=x'_U$ when $i \notin U$ and the event $i \in U$ is independent from $U \setminus \{i\}$. (This is where we use the Poisson subsampling assumption.)
        
        Thus we can express the distribution of $M^U(x)$ as a convex combination: $M(x_U) = p \cdot P + (1-p) \cdot Q$, since $p=\pr{}{i \in U}$. 
        
        For all $\alpha \in \mathbb{N}_{\ge 2}$, $M$ is assumed to be $(\alpha,\varepsilon(\alpha))$-RDP, so we have $\dr{\alpha}{P}{Q} \le \varepsilon(\alpha)$ and $\dr{\alpha}{Q}{P} \le \varepsilon(\alpha)$.
        
        To complete the proof we must show that \[\dr{\alpha}{p \cdot P + (1-p) \cdot Q}{Q} \le \varepsilon'_p(\alpha)\] and \[\dr{\alpha}{Q}{p \cdot P + (1-p) \cdot Q} \le \varepsilon'_p(\alpha)\] for all $\alpha \in \mathbb{N}_{\ge 2}$.
        
        Fix $\alpha \in \mathbb{N}_{\ge 2}$. We have
        \begin{align*}
            e^{(\alpha-1)\dr{\alpha}{p\cdot P + (1-p)\cdot Q}{Q}}
            &= \ex{Y \gets Q}{\left(\frac{p \cdot P(Y) + (1-p) \cdot Q(Y)}{Q(Y)}\right)^\alpha}\\
            &= \ex{Y \gets Q}{\left(1-p + p \cdot \frac{P(Y)}{Q(Y)}\right)^\alpha}\\
            &= \ex{Y \gets Q}{\sum_{k=0}^\alpha {\alpha \choose k} (1-p)^{\alpha-k} p^k\left(\frac{P(Y)}{Q(Y)}\right)^k}\tag{Binomial expansion}\\
            &= \sum_{k=0}^\alpha {\alpha \choose k} (1-p)^{\alpha-k} p^k\ex{Y \gets Q}{\left(\frac{P(Y)}{Q(Y)}\right)^k}\\
            &= (1-p)^\alpha + \alpha (1-p)^{\alpha-1} p + \sum_{k=2}^\alpha {\alpha \choose k} (1-p)^{\alpha-k} p^k \cdot e^{(k-1)\dr{k}{P}{Q}} \tag{$\ex{Y \gets Q}{\frac{P(Y)}{Q(Y)}}=1$}\\
            &\le (1-p)^{\alpha-1}(1 + (\alpha-1) p) + \sum_{k=2}^\alpha {\alpha \choose k} (1-p)^{\alpha-k} p^k \cdot e^{(k-1)\varepsilon(k)} \tag{$\dr{k}{P}{Q}\le \varepsilon(k)$}\\
            &= e^{(\alpha-1)\varepsilon'_p(\alpha)}.
        \end{align*}
        Note that $(1-p)^{\alpha-1}(1 + (\alpha-1) p) \le (e^{-p})^{\alpha-1} e^{(\alpha-1)p} = 1$.
        
        An identical calculation shows that \[\dr{\alpha}{p \cdot Q + (1-p) \cdot P}{P} \le \varepsilon'_p(\alpha).\]
        
        Finally, Theorem \ref{thm:rdp_ss_flip} gives \[\dr{\alpha}{Q}{pP+(1-p)Q} \le \max\left\{\begin{array}{c} \dr{\alpha}{pP+(1-p)Q}{Q}, \\ \dr{\alpha}{pQ+(1-p)P}{P} \end{array} \right\} \le \varepsilon'_p(\alpha).\]
    \end{proof}

    The following result shows that, in terms of subsampling for R\'enyi DP, it suffices to analyze one side of the add/remove neighbouring relation.
    \begin{theorem}\label{thm:rdp_ss_flip}
        Let $P$ and $Q$ be probability distributions that are absolutely continuous with respect to each other. Let $p \in [0,1]$ and $\alpha \in (1,\infty)$. Set $\lambda = \frac{(2\alpha-1)p}{(2\alpha-1)p+3(1-p)}$. Then \[e^{(\alpha-1)\dr{\alpha}{Q}{pP+(1-p)Q}} \le (1-\lambda) \cdot e^{(\alpha-1)\dr{\alpha}{pP+(1-p)Q}{Q}} + \lambda \cdot e^{(\alpha-1)\dr{\alpha}{pQ+(1-p)P}{P}}.\]        
    \end{theorem}
    Since $\lambda \in [0,1]$, this implies \[\dr{\alpha}{Q}{pP+(1-p)Q} \le \max\left\{\begin{array}{c} \dr{\alpha}{pP+(1-p)Q}{Q}, \\ \dr{\alpha}{pQ+(1-p)P}{P} \end{array} \right\}.\]
    \begin{proof}
        Define $f :(0,\infty) \to \mathbb{R}$ by \[f(x) = (1-\lambda) (1-p+p\cdot x)^\alpha + \lambda \cdot x \cdot \left(1-p+\frac{p}{x}\right)^\alpha - \left(1-p+p\cdot x\right)^{1-\alpha}.\]
        We have
        \begin{align*}
            & (1-\lambda) \cdot e^{(\alpha-1)\dr{\alpha}{pP+(1-p)Q}{Q}} + \lambda \cdot e^{(\alpha-1)\dr{\alpha}{pQ+(1-p)P}{P}} - e^{(\alpha-1)\dr{\alpha}{Q}{pP+(1-p)Q}} \\
            &= (1-\lambda) \cdot \ex{Y \gets Q}{\left(\frac{pP(Y)+(1-p)Q(Y)}{Q(Y)}\right)^\alpha} + \lambda \cdot \ex{Y \gets P}{\left(\frac{pQ(Y)+(1-p)P(Y)}{P(Y)}\right)^\alpha} \\&~~~~~~~~~~- \ex{Y \gets Q}{\left(\frac{Q(Y)}{pP(Y)+(1-p)Q(Y)}\right)^{\alpha-1}}\\
            &= (1-\lambda) \cdot \ex{Y \gets Q}{\left(p\frac{P(Y)}{Q(Y)}+(1-p)\right)^\alpha} + \lambda \cdot \ex{Y \gets P}{\left(p \left(\frac{P(Y)}{Q(Y)}\right)^{-1} + 1-p\right)^\alpha} \\&~~~~~~~~~~- \ex{Y \gets Q}{\left(p\frac{P(Y)}{Q(Y)}+(1-p)\right)^{1-\alpha}}\\
            &= (1-\lambda) \cdot \ex{Y \gets Q}{\left(p\frac{P(Y)}{Q(Y)}+(1-p)\right)^\alpha} + \lambda \cdot \ex{Y \gets Q}{\frac{P(Y)}{Q(Y)} \cdot \left(p \left(\frac{P(Y)}{Q(Y)}\right)^{-1} + 1-p\right)^\alpha} \\&~~~~~~~~~~- \ex{Y \gets Q}{\left(p\frac{P(Y)}{Q(Y)}+(1-p)\right)^{1-\alpha}} \tag{For any $g$, $\ex{Y \gets Q}{\frac{P(Y)}{Q(Y)} g(Y)} = \ex{Y \gets P}{g(Y)}$}\\
            &= \ex{Y \gets Q}{(1-\lambda) \cdot \left( 1 - p + p \cdot \frac{P(Y)}{Q(Y)}\right)^\alpha + \lambda \cdot \frac{P(Y)}{Q(Y)} \cdot \left( 1 - p + \frac{p}{ ~\frac{P(Y)}{Q(Y)}~ } \right)^\alpha - \left(1-p+p \cdot \frac{P(Y)}{Q(Y)}\right)^{1-\alpha}}\\
            &= \ex{X}{f(X)},
        \end{align*}
        where $X=\frac{P(Y)}{Q(Y)}$ for $Y \gets Q$. Thus our objective is to show that $\ex{X}{f(X)} \ge 0$.\footnote{Note that we assume $P$ and $Q$ are absolutely continuous with respect to each other -- i.e., $\forall S ~ P(S)=0 \iff Q(S)=0$. This ensures that the Radon-Nikodym derivative $\frac{P(y)}{Q(y)}$ is well-defined and, further that $\pr{Y \gets Q}{\frac{P(Y)}{Q(Y)}\le 0}=0$. Thus the function $f$ need only be defined on $(0,\infty)$.}
        
        We claim that $f$ is convex. Convexity implies $f(x) \ge f(1) + f'(1) \cdot (x-1)$ for all $x \in (0,\infty)$. Since $f(1) = 0$ and $\ex{}{X} = \ex{Y \gets Q}{\frac{P(Y)}{Q(Y)}} = 1$, this implies $\ex{}{f(X)} \ge f(1) + f'(1) \ex{}{X-1} = 0$, as required.
        
        It only remains to prove that $f$ is convex.
        We have, for all $x>0$,
        \begin{align*}
            f(x) &= (1-\lambda) (1-p+p\cdot x)^\alpha + \lambda \cdot x \cdot \left(1-p+\frac{p}{x}\right)^\alpha - \left(1-p+p\cdot x\right)^{1-\alpha},\\
            f'(x) &= (1-\lambda) \alpha p (1-p+p\cdot x)^{\alpha-1} + \lambda \cdot \left(1-p+\frac{p}{x}\right)^\alpha \\&~~~~~ - \lambda  \alpha \frac{p}{x}  \left(1-p+\frac{p}{x}\right)^{\alpha-1} + (\alpha-1) p \left(1-p+p\cdot x\right)^{-\alpha}\\
            f''(x) &= (1-\lambda) \alpha(\alpha-1) p^2 (1-p+p\cdot x)^{\alpha-2} - \lambda \alpha \frac{p}{x^2} \left(1-p+\frac{p}{x}\right)^{\alpha-1} + \lambda \alpha \frac{p}{x^2}  \left(1-p+\frac{p}{x}\right)^{\alpha-1} \\&~~~~~  + \lambda  \alpha(\alpha-1) \frac{p^2}{x^3}  \left(1-p+\frac{p}{x}\right)^{\alpha-2} - \alpha(\alpha-1) p^2 \left(1-p+p\cdot x\right)^{-\alpha-1}\\
            &= \alpha(\alpha-1)p^2 \left( (1-\lambda) (1-p+px)^{\alpha-2} + \lambda \frac{1}{x^3}\left(1-p+\frac{p}{x}\right)^{\alpha-2} - (1-p+px)^{-\alpha-1} \right)\\
            &= \frac{\alpha(\alpha-1)p^2}{(1-p+px)^{\alpha+1}} \left( (1-\lambda) (1-p+px)^{2\alpha-1} + \lambda \frac{1}{x^3}\left(1-p+\frac{p}{x}\right)^{\alpha-2} \cdot (1-p+px)^{\alpha+1} - 1 \right)\\
            &= \frac{\alpha(\alpha\!-\!1)p^2}{(1\!-\!p\!+\!px)^{\alpha+1}}\!\left(\!(1\!-\!\lambda) (1\!-\!p\!+\!px)^{2\alpha-1}\!+\!\lambda\!\left(\!\frac{1\!-\!p\!+\!px}{x}\!\right)^3\!\cdot\!\left(\!1\!-\!p\!+\!\frac{p}{x}\right)^{\alpha-2}\!\cdot\!(1\!-\!p\!+\!px)^{\alpha-2}\!-\!1\!\right)\\
            &\ge \frac{\alpha(\alpha-1)p^2}{(1-p+px)^{\alpha+1}} \left( (1-\lambda) (1-p+px)^{2\alpha-1} + \lambda \left(\frac{1-p+px}{x}\right)^3 \cdot 1 - 1 \right) \tag{Lemma \ref{lem:frac}}\\
            &= \frac{\alpha(\alpha-1)p^2}{(1-p+px)^{\alpha+1}} \left( \frac{3(1-p) (1-p+px)^{2\alpha-1} + (2\alpha-1)p \left(\frac{1-p}{x} + p\right)^3 - 3(1-p)-(2\alpha-1)p}{3(1-p)+(2\alpha-1)p} \right) \tag{$\lambda = \frac{(2\alpha-1)p}{(2\alpha-1)p+3(1-p)}$}\\
            &\ge 0. \tag{Lemma \ref{lem:wtf}}
        \end{align*}
    \end{proof}
    
    We now give the auxiliary lemmata used in the proof of Theorem \ref{thm:rdp_ss_flip}.
    
    \begin{lemma}\label{lem:frac}
        For all $p \in [0,1]$ and $x \in (0,\infty)$, \[\frac{1}{1-p+p/x} \le 1-p+p \cdot x .\]
    \end{lemma}
    \begin{proof}
        Let $f(t)=t+1/t$. Then $f'(t)=1-1/t^2$ and $f''(t) = 2/t^3>0$ for all $t>0$. Thus $f'(t)=0 \iff t=1$ and $f(x) \ge f(1) = 2$.
        Now
        \begin{align*}
            (1-p+p \cdot x)\cdot(1-p+p/x) &= p^2 + (1-p)^2 + p(1-p)(x+1/x)\\
            &\ge p^2 + (1-p)^2 + p(1-p) \cdot 2\\
            &= (p + (1-p))^2 = 1.
        \end{align*}
        Rearranging yields the result.
    \end{proof}
    
    \begin{lemma}\label{lem:wtf}
        For all $v \ge 1$, $p \in [0,1]$, and $x\in (0,\infty)$, \[3(1-p)(1-p+px)^v + vp\left(\frac{1-p}{x}+p\right)^3 \ge 3(1-p) + vp.\]
    \end{lemma}
    \begin{proof}
        Define $f : (0,\infty) \to \mathbb{R}$ by \[f(x) = 3(1-p)(1-p+px)^v+vp\left(\frac{1-p}{x}+p\right)^3.\]
        Our goal is to prove that $f(x) \ge f(1) = 3(1-p) + vp$ for all $x \in (0,\infty)$. It suffices to prove that $f$ is convex and that $f'(1)=0$.
        We have
        \begin{align*}
            f'(x) %
            &= 3vp(1-p) \left( (1-p+px)^{v-1} - \frac{1}{x^2} \left( \frac{1-p}{x} + p \right)^2 \right),\\
            f''(x) &= 3vp(1-p)\left( (v-1) p (1-p+px)^{v-2} + \frac{2}{x^3} \cdot \left( \frac{1-p}{x} + p \right)^2 + \frac{2}{x^2} \cdot \frac{1-p}{x^2}\cdot \left( \frac{1-p}{x} + p \right)\right).
        \end{align*}
        From these expressions, it is easy to see that $f'(1)=0$ and $f''(x)\ge 0$ for all $x \in (0,\infty)$.
    \end{proof}

\subsection{Analytic R\'enyi DP Bound for Privacy Amplification by Poisson Subsampling}
    Theorem \ref{thm:rdp_subsampling} gives a tight RDP bound for privacy amplification by Poisson subsampling. 
    However, the bound is in the form of a series. This is adequate for numerical purposes, but it is helpful for our understanding to have a simpler closed-form expression.
    
    In this subsection we will provide a simpler expression and attempt to build some understanding of how privacy amplification by subsampling applies to R\'enyi DP.
    
    \begin{theorem}[Asymptotic Privacy Amplification by Subsampling for R\'enyi DP]\label{thm:rdp_subsampling_asymptotic}
        Let $p \in [0,1/2]$ and $\rho\in(0,1]$. Define $\omega = \min\left\{ \frac{\log(1/p)}{4\rho} , 1 + p^{-1/4}\right\}$. Assume $\omega \ge 3 + 2 \frac{\log(1/\rho)}{\log(1/p)}$.
        
        Let $M : \mathcal{X}^n \to \mathcal{Y}$ satisfy $\rho$-zCDP with respect to addition or removal.\footnote{I.e., $x,x'\in\mathcal{X}^n$ are neighbouring if, for some $i \in [n]$, we have $x_i=\bot$ or $x'_i=\bot$, and $\forall j \ne i ~~ x_j=x'_j$, where $\bot\in\mathcal{X}$ is some fixed value.}
        
        Define $M^U : \mathcal{X}^n \to \mathcal{Y}$ by $M^U(x) = M(x_U)$, where $U \subset [n]$ be a random set that contains each element independently with probability $p$ and, for $x \in \mathcal{X}^n$, $x_U \in \mathcal{X}^n$ is given by $(x_U)_i = x_i$ if $i \in U$ and $(x_U)_i = \bot$ if $i \notin U$.
        
        Then $M^U$ satisfies $(\alpha,10 p^2 \rho \alpha)$-RDP for all $\alpha \in (1,\omega)$.
    \end{theorem}
    
    \begin{figure}
        \centering
        \includegraphics[width=0.75\textwidth]{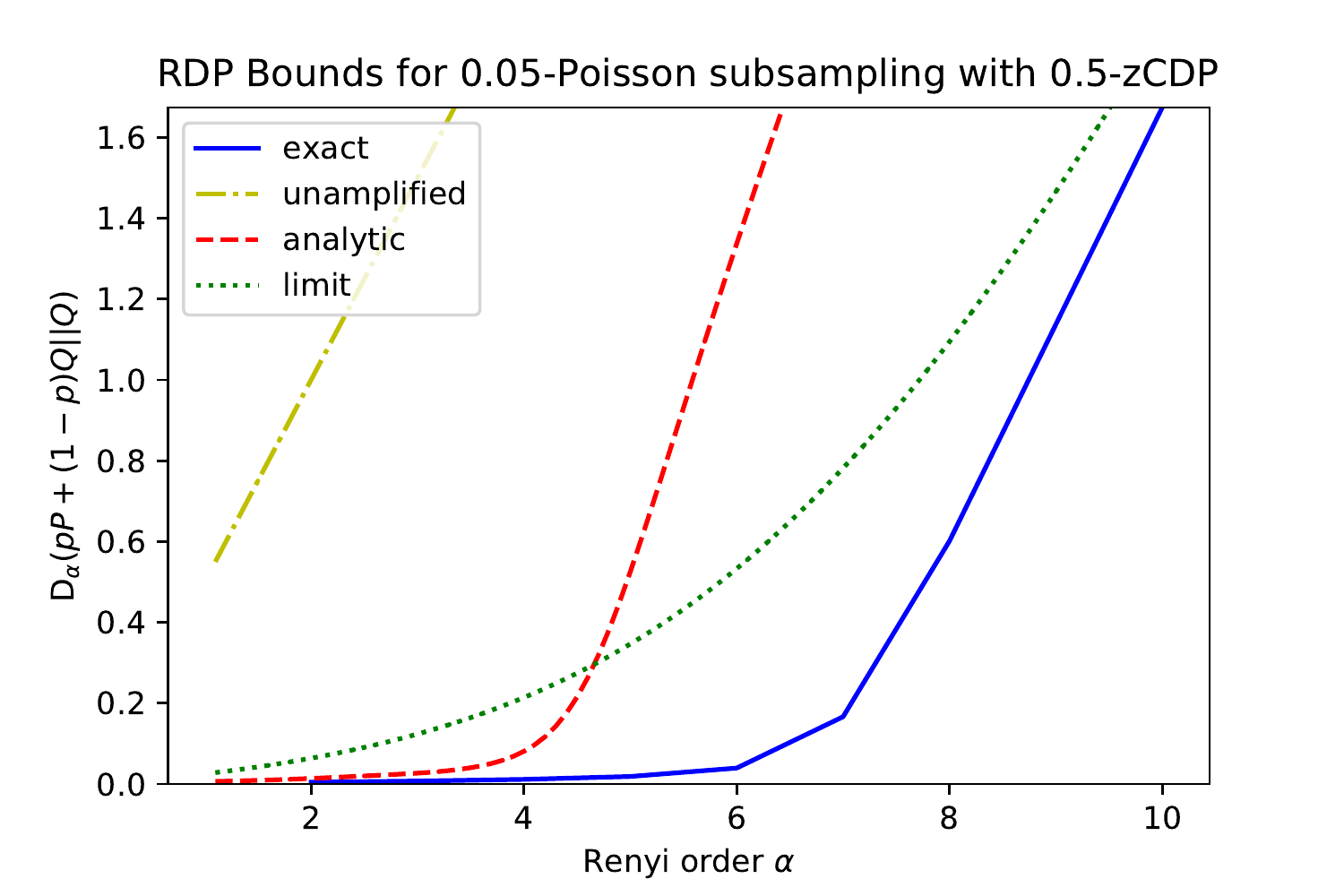}
        \caption{Comparison of R\'enyi divergence guarantees for Poisson subsampling -- i.e., including each person with probability $p=0.05$. The \texttt{unamplified} algorithm satisfies $0.5$-zCDP. The \texttt{exact} bound is given by Theorem \ref{thm:rdp_subsampling}. For comparison, we have the \texttt{analytic} upper bound from Proposition \ref{prop:divergence_subsampling_analytic} as well as the behaviour in the \texttt{limit} given by Proposition \ref{prop:rdp_subsampling_jensen}. %
        }
        \label{fig:rdp_subsampling}
    \end{figure}
    
    There are many caveats in the statement of Theorem \ref{thm:rdp_subsampling_asymptotic}, but the high level message is that Poisson subsampling a $p$ fraction of the dataset amplifies $\rho$-zCDP to something like $O(p^2 \cdot \rho)$-zCDP.
    We will discuss these caveats in a moment, but, ignoring these caveats and the constant factor in the guarantee, this is exactly the kind of guarantee we would hope for.
    
    Consider the following example, which illustrates what kind guarantee we would hope for. Suppose we have a query $q : \mathcal{X} \to [0,1]$ and a sensitive dataset $x \in \mathcal{X}^n$ and our goal is to estimate $q(x) := \frac{1}{n} \sum_i^n q(x_i)$.
    We could release a sample from $\mathcal{N}(q(x),\sigma^2)$, which satisfies $\frac{1}{2n^2\sigma^2}$-zCDP and has mean squared error $\sigma^2$.
    However, perhaps due to computational constraints, we might instead select a random $p$ fraction $U \subset [n]$ and instead release a sample from $M^U(x) = \mathcal{N}\left( \frac{1}{pn} \sum_{i \in U} q(x_i), \sigma^2 \right)$. We can calculate that the mean squared error of this algorithm is at most $\sigma^2 + \frac{1-p}{pn}$. Without amplification this satisfies $(\rho=\frac{1}{2p^2n^2\sigma^2})$-zCDP. With amplification, Theorem \ref{thm:rdp_subsampling_asymptotic} tells us that this satisfies $(\alpha,O(p^2 \cdot \rho))$-RDP for $\alpha$ not too large. Now $p^2 \cdot \rho = \frac{1}{2n^2\sigma^2}$ is exactly the guarantee that we obtained by simply evaluating $q$ on the entire dataset and avoiding subsampling. We cannot hope to do better than this.
    
    The constant factor of 10 in the theorem can be improved, but a constant factor loss is the price we pay for having a simpler expression; if we want tight constants we should apply Theorem \ref{thm:rdp_subsampling}.
    
    The main caveat in Theorem \ref{thm:rdp_subsampling_asymptotic} is the requirement that $\alpha \le \omega \le \frac{\log(1/p)}{4\rho}$.
    This is necessary, as the $(\alpha,\varepsilon(\alpha))$-RDP guarantee qualitatively changes when $\alpha \ge O(\log(1/p)/\rho)$. It changes from $\varepsilon(\alpha) = O(p^2 \rho \alpha)$ to $\varepsilon(\alpha)=\rho\alpha - O(\log(1/p))$.
    To see why this is inherent, consider the following lower bound. For all $p \in [0,1]$, all $\alpha \in (1,\infty)$, and all absolutely continuous probability distributions $P$ and $Q$, we have
    \[e^{(\alpha-1)\dr{\alpha}{pP+(1-p)Q}{Q}} = \ex{Y \gets Q}{\left(1-p+p\cdot\frac{P(Y)}{Q(Y)}\right)^\alpha} \ge \ex{Y \gets Q}{\left(p\cdot\frac{P(Y)}{Q(Y)}\right)^\alpha} = p^\alpha \cdot e^{(\alpha-1)\dr{\alpha}{P}{Q}}.\]
    Thus $\dr{\alpha}{pP+(1-p)Q}{Q} \ge \dr{\alpha}{P}{Q} - \frac{\alpha}{\alpha-1} \log(1/p)$.
    This tells us that, for large $\alpha$, we cannot have more than an additive improvement in the RDP guarantee, whereas for small $\alpha$ we have a multiplicative improvement. Proposition \ref{prop:rdp_subsampling_jensen} shows that this lower bound is tight.
    
    We now proceed to prove Theorem \ref{thm:rdp_subsampling_asymptotic}.
    
    \begin{lemma}\label{lem:taylor}
        Let $\alpha \in (1,\infty)$, $p \in \left[0,1-e^{-1}\right]$, and $x \in \left[0,\infty\right)$. If either $\alpha \le 2$ or $\alpha>2$ and $px \le \max \left\{ p , \frac{1-p}{\alpha-2} \right\}$, then \[(1-p+p\cdot x)^\alpha \le 1 + \alpha p (x-1) + \frac{e}{2}\alpha(\alpha-1)p^2(x-1)^2.\]
    \end{lemma}
    \begin{proof}
        We assume $p>0$. Otherwise the result is trivial.
        
        Define $f : [0,\infty) \to \mathbb{R}$ by \[f(x) = (1-p+px)^\alpha.\]
        For all $x\in[0,\infty)$, we have
        \begin{align*}
            f'(x) &= \alpha p (1-p+px)^{\alpha-1},\\
            f''(x) &= \alpha(\alpha-1) p^2 (1-p+px)^{\alpha-2},\\
            f'''(x) &= \alpha(\alpha-1)(\alpha-2) p^3 (1-p+px)^{\alpha-3}.
        \end{align*}
        By Taylor's theorem, for all $x \in \left[0,\infty\right)$, there exists $\xi_x \in \left[\min\{1,x\},\max\{1,x\}\right]$ such that
        \begin{align*}
            f(x) &= f(1) + f'(1)(x-1) + \frac12 f''(\xi_x)(x-1)^2\\
            &= 1 + \alpha p (x-1) + \frac12 f''(\xi_x) (x-1)^2.
        \end{align*}
        To complete the proof is suffices to show that $f''(\xi) \le e \cdot \alpha(\alpha-1) p^2$ in two cases: First, for all $\xi \in [0,\infty)$ assuming $\alpha \le 2$. Second, for all $\xi \in  \left[0,\max\left\{1,\frac{1-p}{p}\frac{1}{\alpha-2}\right\}\right]$ assuming $\alpha>2$. (Note that, $\xi_x \in \left[0,\max\left\{1,\frac{1-p}{p}\frac{1}{\alpha-2}\right\}\right]$ is implied by the assumptions $px \le \max \left\{ p , \frac{1-p}{\alpha-2} \right\}$ and $p>0$.)
        
        First, suppose $\alpha\le 2$. Then $f'''(x)\le 0$ for all $x \in [0,\infty)$. Thus $f''$ is decreasing (or constant) and, for all $\xi \in  \left[0,\infty\right)$, we have
        \begin{align*}
            f''(\xi) &\le f''(0)\\
            &= \alpha(\alpha-1)p^2 (1-p)^{\alpha-2}\\
            &\le \alpha(\alpha-1)p^2 \frac{1}{1-p} \tag{$\alpha > 1$}\\
            &\le \alpha(\alpha-1)p^2 \cdot e. \tag{$p \le 1-e^{-1}$}
        \end{align*}
        
        Second, assume $\alpha > 2$ and $px \le \max \left\{ p , \frac{1-p}{\alpha-2} \right\}$, which implies $\xi_x \le \max \left\{ 1 , \frac{1-p}{p}\frac{1}{\alpha-2} \right\}$.

        We have $f'''(x) > 0$ for all $x \in [0,\infty)$. Thus $f''$ is increasing and, for all $\xi \in \left[ 0 ,  \max \left\{ 1 , \frac{1-p}{p}\frac{1}{\alpha-2} \right\} \right]$, we have
        \begin{align*}
            f''(\xi) &\le f''\left( \max \left\{ 1 , \frac{1-p}{p}\frac{1}{\alpha-2} \right\} \right)\\
            &= \alpha(\alpha-1)p^2 \left( 1-p + p \cdot \max \left\{ 1 , \frac{1-p}{p}\frac{1}{\alpha-2} \right\}\right)^{\alpha-2} \\
            &= \alpha(\alpha-1)p^2 \max \left\{ 1 , (1-p)^{\alpha-2} \cdot \left( 1 + \frac{1}{\alpha-2} \right)^{\alpha-2} \right\} \\
            &\le \alpha(\alpha-1)p^2 \max \left\{ 1 , (1-p)^0 \cdot \left( e^{\frac{1}{\alpha-2}} \right)^{\alpha-2} \right\} \\
            &= \alpha(\alpha-1)p^2 \cdot e.
        \end{align*}
    \end{proof}
    \begin{lemma}\label{lem:taylor_high}
        Let $\alpha,\omega \in (1,\infty)$ with $\alpha \le \omega$, $p \in \left[0,1-e^{-1}\right]$, and $x \in \left[0,\infty\right)$. Then \[(1-p+p\cdot x)^\alpha \le 1 + \alpha p (x-1) + \frac{e}{2}\alpha(\alpha-1)p^2(x-1)^2 + \big((\alpha-1) p x\big)^\omega.\]
    \end{lemma}
    \begin{proof}
        We can assume $p>0$, as otherwise the result is trivial.
        
        If $\alpha \le 2$ or if $\alpha>2$ and $x \le \max \left\{ 1 , \frac{1-p}{p}\frac{1}{\alpha-2} \right\}$, then the result follows from Lemma \ref{lem:taylor}, as $\big((\alpha-1) p x\big)^\omega \ge 0$.
        
        Thus we assume $\alpha>2$ and $x \ge\max \left\{ 1 , \frac{1-p}{p}\frac{1}{\alpha-2} \right\}$.
        
        Since $x \ge 1$, we have $\alpha p (x-1) + \frac{e}{2}\alpha(\alpha-1)p^2(x-1)^2 \ge 0$. Therefore it suffices to prove that $(1-p+px)^\alpha \le ((\alpha-1) p x)^\omega$.
        
        The assumption $x \ge 1$ implies $1-p+px \ge 1$ and, hence, that $(1-p+px)^\alpha \le (1-p+px)^\omega$, as we have $\alpha \le \omega$.
        The assumption $x \ge \frac{1-p}{p}\frac{1}{\alpha-2}$ rearranges to $1-p \le px(\alpha-2)$, which implies $1-p+px \le (\alpha-1) p x$ and, hence, $(1-p+px)^\omega \le ((\alpha-1)px)^\omega$, as required.
    \end{proof}
    \begin{proposition}[Analytic Privacy Amplification by Subsampling for R\'enyi Divergence]\label{prop:divergence_subsampling_analytic}
        Let $P$ and $Q$ be probability distributions with $P$ absolutely continuous with respect to $Q$. Let $p \in \left[0,1-e^{-1}\right]$ and $\alpha,\omega \in (1,\infty)$ with $\alpha \le \omega$. Then
        \begin{align*}
            \dr{\alpha}{pP+(1-p)Q}{Q} &\le \frac{1}{\alpha-1} \log\left( 1 + \frac{e}{2} \alpha(\alpha-1) p^2 \left( e^{\dr{2}{P}{Q}}-1\right) + ((\alpha-1) p)^\omega \cdot e^{(\omega-1)\dr{\omega}{P}{Q}} \right)\\
            &\le \alpha \cdot \frac{e}{2} \cdot p^2 \cdot \left( e^{\dr{2}{P}{Q}}-1\right) + p \cdot \left( (\alpha-1) \cdot p \cdot e^{\dr{\omega}{P}{Q}}\right)^{\omega-1}.
        \end{align*}
    \end{proposition}
    \begin{proof}
        We have
        \begin{align*}
            &e^{(\alpha-1) \dr{\alpha}{pP+(1-p)Q}{Q}} \\
            &~= \ex{Y \gets Q}{\left(\frac{p \cdot P(Y) + (1-p) \cdot Q(Y) }{Q(Y)}\right)^\alpha}\\
            &~= \ex{Y \gets Q}{\left(1-p+p \cdot \frac{P(Y)}{Q(Y)}\right)^\alpha}\\
            &~\le \ex{Y \gets Q}{1 + \alpha p \left(\frac{P(Y)}{Q(Y)}-1\right) + \frac{e}{2}\alpha(\alpha-1)p^2\left(\frac{P(Y)}{Q(Y)}-1\right)^2 + \left((\alpha-1) p \frac{P(Y)}{Q(Y)}\right)^\omega} \tag{Lemma \ref{lem:taylor_high}}\\
            &~= 1 + \alpha p \left(1-1\right) + \frac{e}{2}\alpha(\alpha-1)p^2\left(e^{\dr{2}{P}{Q}}-1\right) + \left((\alpha-1) p \right)^\omega \cdot e^{(\omega-1)\dr{\omega}{P}{Q}}.
        \end{align*}
        The second inequality in the result follows from the fact that $\log(1+u) \le u$ for all $u>-1$.
    \end{proof}
    We also have the following simpler result that provides better bounds when the R\'enyi order $\alpha$ is large.
    \begin{proposition}\label{prop:rdp_subsampling_jensen}
        Let $P$ and $Q$ be probability distributions with $P$ absolutely continuous with respect to $Q$. Let $p \in \left[0,1\right]$ and $\alpha \in (1,\infty)$. Then
        \begin{align*}
            \dr{\alpha}{pP+(1-p)Q}{Q} &\le \frac{\alpha}{\alpha-1} \log\left( 1-p + p \cdot e^{(1-1/\alpha)\dr{\alpha}{P}{Q}} \right)\\
            &= \dr{\alpha}{P}{Q} - \frac{\alpha}{\alpha-1} \log (1/p)  + \frac{\alpha}{\alpha-1} \log\left( 1 + \frac{1-p}{p} \cdot e^{-\frac{\alpha-1}{\alpha}\dr{\alpha}{P}{Q}} \right)\\
            &\le \dr{\alpha}{P}{Q} - \frac{\alpha}{\alpha-1} \log (1/p)  + \frac{\alpha}{\alpha-1} \cdot \frac{1-p}{p} \cdot e^{-\frac{\alpha-1}{\alpha}\dr{\alpha}{P}{Q}}\\
        \end{align*}
    \end{proposition}
    \begin{proof}
        We assume $0<p<1$, as the result is immediate otherwise.
        By Jensen's inequality and the convexity of $v \mapsto v^\alpha$, for all $x \in [0,\infty)$ and all $\lambda \in (0,1)$, \[(1-p+px)^\alpha = \left((1-\lambda)\cdot\frac{1-p}{1-\lambda} + \lambda \cdot \frac{px}{\lambda} \right)^\alpha \le (1-\lambda)\cdot\left(\frac{1-p}{1-\lambda}\right)^\alpha + \lambda \cdot \left(\frac{px}{\lambda} \right)^\alpha. \]
        Now, for all $\lambda \in (0,1)$, we have
        \begin{align*}
            e^{(\alpha-1)\dr{\alpha}{pP+(1-p)Q}{Q}} &= \ex{Y \gets Q}{\left( 1-p + p \frac{P(Y)}{Q(Y)} \right)^\alpha}\\
            &\le \ex{Y \gets Q}{(1-\lambda)\cdot\left(\frac{1-p}{1-\lambda}\right)^\alpha + \lambda \cdot \left(\frac{p}{\lambda} \cdot \frac{P(Y)}{Q(Y)} \right)^\alpha}\\
            &= (1-\lambda)^{1-\alpha} \cdot (1-p)^\alpha + \lambda^{1-\alpha} \cdot p^\alpha \cdot e^{(\alpha-1)\dr{\alpha}{P}{Q}}.
        \end{align*}
        We can choose $\lambda$ to minimize this expression. It turns out to be optimal to set $\lambda = \frac{1}{1+\frac{1-p}{p} \cdot e^{-(1-1/\alpha)\dr{\alpha}{P}{Q}}}$. Now we have
        \begin{align*}
            &e^{(\alpha-1)\dr{\alpha}{pP+(1-p)Q}{Q}}\\
            &\le (1-\lambda)^{1-\alpha} \cdot (1-p)^\alpha + \lambda^{1-\alpha} \cdot p^\alpha \cdot e^{(\alpha-1)\dr{\alpha}{P}{Q}}\\
            &= \left(1 + \frac{p}{1-p}e^{(1-1/\alpha)\dr{\alpha}{P}{Q}}\right)^{\alpha-1} \cdot (1-p)^\alpha \\&~~~+ \left(1+\frac{1-p}{p} \cdot e^{-(1-1/\alpha)\dr{\alpha}{P}{Q}}\right)^{\alpha-1} \cdot p^\alpha \cdot e^{(\alpha-1)\dr{\alpha}{P}{Q}}\\
            &= \left(1-p + p \cdot e^{(1-1/\alpha)\dr{\alpha}{P}{Q}}\right)^{\alpha-1} \cdot (1-p) \\&~~~+ \left(p+(1-p) \cdot e^{-(1-1/\alpha)\dr{\alpha}{P}{Q}}\right)^{\alpha-1} \cdot p \cdot e^{(\alpha-1)\dr{\alpha}{P}{Q}}\\
            &= \left(1-p + p \cdot e^{(1-1/\alpha)\dr{\alpha}{P}{Q}}\right)^{\alpha-1} \cdot (1-p) \\&~~~+ \left(p \cdot e^{(1-1/\alpha)\dr{\alpha}{P}{Q}}+(1-p)\right)^{\alpha-1} \cdot p \cdot e^{(1-1/\alpha)\dr{\alpha}{P}{Q}}\\
            &= \left(1-p + p \cdot e^{(1-1/\alpha)\dr{\alpha}{P}{Q}}\right)^{\alpha}.
        \end{align*}
        Rearranging yields the result.
    \end{proof}
    
    \begin{proof}[Proof of Theorem \ref{thm:rdp_subsampling_asymptotic}.]
        Fix neighbouring inputs $x,x'\in\mathcal{X}^n$. Fix some $\alpha \in (1,\omega)$ with $\omega = \min\left\{ \frac{\log(1/p)}{4\rho} , 1 + p^{-1/4}\right\} \ge 3 + 2 \frac{\log(1/\rho)}{\log(1/p)}$.
        
        Without loss of generality $x'$ is $x$ with some element removed. That is, we can fix some $i \in [n]$ such that $x'_i=\bot$ and $x'_j=x_j$ for all $j \ne i$.
        
        Let $P=M(x_U)|_{i \in U}$ and let $Q=M(x_U)|_{i \notin U}$. Then $M^U(x) = M(x_U) = p P + (1-p)Q$. Also $M(x')=Q$
        
        Thus we must prove that $\dr{\alpha}{pP+(1-p)Q}{Q} \le 10p^2\rho\alpha$ and $\dr{\alpha}{Q}{pP+(1-p)Q} \le 10p^2\rho\alpha$. Since $M$ is assumed to be $\rho$-zCDP, we have $\dr{\alpha'}{P}{Q} \le \rho\alpha'$ and $\dr{\alpha'}{Q}{P} \le \rho\alpha'$ for all $\alpha'\in(1,\infty)$.
        
        By Proposition \ref{prop:divergence_subsampling_analytic},
        \begin{align*}
            \dr{\alpha}{pP+(1-p)Q}{Q} &\le \alpha \cdot \frac{e}{2} \cdot p^2 \cdot \left( e^{\dr{2}{P}{Q}}-1\right) + p \cdot \left( (\alpha-1) \cdot p \cdot e^{\dr{\omega}{P}{Q}}\right)^{\omega-1}\\
            &\le \alpha \cdot \frac{e}{2} \cdot p^2 \cdot \left( e^{2\rho}-1\right) + p \cdot \left( (\alpha-1) \cdot p \cdot e^{\omega \rho}\right)^{\omega-1}\\
            &\le \alpha \cdot \frac{e}{2} \cdot p^2 \cdot \left( e^{2\rho}-1\right) + p \cdot \left( p^{-1/4} \cdot p \cdot p^{-1/4}\right)^{\omega-1} \tag{$\alpha\le\omega=\min\{1+p^{-1/4},\log(1/p)/4\rho\}$}\\
            &= \alpha \cdot \frac{e}{2} \cdot p^2 \cdot \left( e^{2\rho}-1\right) + p^{\frac{1+\omega}{2}}\\
            &\le \alpha \cdot \frac{e}{2} \cdot p^2 \cdot \left( e^{2\rho}-1\right) + p^2 \cdot \rho \tag{$\omega \ge 3+2 \log(1/\rho)/\log(1/p)$}\\
            &= \alpha \cdot p^2 \cdot \rho \cdot \left( \frac{e}{2} \cdot \frac{e^{2\rho}-1}{\rho} + \frac{1}{\alpha}\right)\\
            &\le \alpha \cdot p^2 \cdot \rho \cdot 10. \tag{$\rho \in (0,1)$ and $\alpha \in (1,\omega)$}
        \end{align*}
        Symmetrically, we have $\dr{\alpha}{pQ+(1-p)P}{P} \le \alpha \cdot p^2 \cdot \rho \cdot 10$.
        By Theorem \ref{thm:rdp_ss_flip}, \[\dr{\alpha}{Q}{pP+(1-p)Q} \le \max\left\{\begin{array}{c} \dr{\alpha}{pP+(1-p)Q}{Q}, \\ \dr{\alpha}{pQ+(1-p)P}{P} \end{array} \right\} \le \alpha \cdot p^2 \cdot \rho \cdot 10.\]
    \end{proof}
    
    \subsection{How to Use Privacy Amplification by Subsampling}\label{sec:use_subsampling}
    
    The most common use case for privacy amplification by subsampling is analyzing noisy stochastic gradient descent.
    That is, we repeatedly sample a small subset of the data, compute a function on this subset, and add Gaussian noise.
    To be precise, let $x \in \mathcal{X}^n$ be the private input.
    Iteratively, for $t=1,\cdots,T$, we pick some function $q_t : \mathcal{X}^n \to \mathbb{R}^d$ and randomly sample a subset $U_t \subset [n]$; then we reveal $\mathcal{N}(q_t(x_{U_t}),\sigma^2 I_d)$.
    
    The addition of Gaussian noise satisfies concentrated DP. Specifically, Lemma \ref{lem:gauss_cdp} shows that releasing $\mathcal{N}(q_t(x),\sigma^2 I_d)$ satisfies $\frac{\Delta_2^2}{2\sigma^2}$-zCDP, where $\Delta_2 = \sup_{x,x'\in\mathcal{X}^n \atop \text{neighbouring}} \|q_t(x)-q_t(x')\|_2$ is the sensitivity of $q_t$. We can thus apply Theorem \ref{thm:rdp_subsampling} to obtain a tight R\'enyi DP guarantee for $\mathcal{N}(q_t(x_{U_t}),\sigma^2 I_d)$, where $U_t$ is a Poisson sample. Finally, we can apply the composition property of R\'enyi DP (Lemma \ref{lem:rdp_composition}) over the $T$ rounds and we can convert this final R\'enyi DP guarantee to approximate DP using Proposition \ref{prop:cdp2adp}.
    This is how differentially private deep learning is analyzed in practice by libraries such as TensorFlow Privacy \cite{tfprivacy,mcmahan2018general}. 
    
    We can also obtain an asymptotic analysis: Theorem \ref{thm:rdp_subsampling_asymptotic} shows that $\mathcal{N}(q_t(x_{U_t}),\sigma^2 I_d)$ with $U_t \subset [n]$ including each element independently with probability $p$ satisfies $\left(\alpha,5\alpha p^2 \Delta_2^2/\sigma^2\right)$-RDP for all $\alpha \in (1,\omega)$.
    Composition over $T$ rounds yields $\left(\alpha,5\alpha T p^2 \Delta_2^2/\sigma^2\right)$-RDP for all $\alpha \in (1,\omega)$, which implies $(\varepsilon,\delta)$-DP for all $\delta>0$ and \[\varepsilon = O\left(\frac{\Delta_2}{\sigma} \cdot p \cdot \sqrt{T \cdot \log(1/\delta)}\right).\]
    This bound is directly comparable to the bound from Section \ref{sec:subsamp_composition}, which was derived by converting back and forth between concentrated DP and approximate DP. The difference is that here we have a $\sqrt{\log(1/\delta)}$ whereas there we had a $\log(T/\delta)$ term. This is the asymptotic improvement obtained by keeping the analysis within RDP.
    This asymptotic improvement also translates into a significant improvement in practice.
    
    We have only analyzed Poisson subsampling, where the size of the sample is random. (Specifically, it follows a binomial distribution.\footnote{A binomial distribution is often well approximated by a Poisson distribution, hence the name.}) Naturally, other subsampling schemes may arise in practice. In particular, a fixed size sample is common. As discussed in Section \ref{sec:addremovereplace}, this corresponds to neighbouring datasets allowing the replacement of one individual's data, rather than addition or removal. In terms of R\'enyi divergences, we must analyze $\dr{\alpha}{pP+(1-p)Q}{pP'+(1-p)Q}$, whereas addition and removal correspond to $\dr{\alpha}{pP+(1-p)Q}{Q}$ and $\dr{\alpha}{Q}{pP'+(1-p)Q}$. However, we can apply group privacy (part 7 of Lemma \ref{lem:rdp_properties}) to reduce the analysis to the case we have already analyzed: For all $\alpha'>\alpha$, we have
    \[\dr{\alpha}{pP\!+\!(1\!-\!p)Q}{pP'\!+\!(1\!-\!p)Q} \!\le\! \frac{\alpha'}{\alpha'\!-\!1}\cdot \dr{\alpha \cdot \frac{\alpha'\!-\!1}{\alpha'\!-\!\alpha}}{pP\!+\!(1\!-\!p)Q}{Q} + \dr{\alpha'}{Q}{pP'\!+\!(1\!-\!p)Q}.\]
    Using group privacy does not yield the tightest bounds, but it suffices to show that, up to small constant factors, sampling a fixed size subset is the same as Poisson subsampling.
    
\section{Historical Notes \& Further Reading}
    \paragraph{Composition.}
    Differential privacy (specifically, pure DP) was introduced by Dwork, McSherry, Nissim, and Smith \cite{dwork2006calibrating}.\footnote{The name ``differential privacy'' does not appear in the original paper. It is attributed to Michael Schroeder \cite{Dwork_McSherry_Nissim_Smith_2017} and first appeared in a talk by Dwork \cite{10.1007/11787006_1}.}
    The original paper gives a form of basic composition (Theorem \ref{thm:basic_composition}), but does not state it in full generality; rather it states a result specific to Laplace noise addition. Approximate DP was introduced by Dwork, Kenthapadi, McSherry, Mironov, and Naor \cite{dwork2006our} and this work gave a more general statement of the basic composition result, as well as an analysis of the Gaussian mechanism (although not as tight as Corollary \ref{cor:gauss_adp_exact}). The tight analysis of the Gaussian mechanism (Corollaries \ref{cor:gauss_adp_exact} \& \ref{cor:gauss_adp_exact_multi}) is due to Balle and Wang \cite{balle2018improving}.
    
    The advanced composition theorem (Theorem \ref{thm:advancedcomposition_approx}) was proved by Dwork, Rothblum, and Vadhan \cite{dwork2010boosting}.\footnote{The original proof showed that the $k$-fold composition of $(\varepsilon,\delta)$-DP algorithms satisfies $(\varepsilon',k\delta+\delta')$-DP with $\delta'>0$ arbitrary and $\varepsilon'=k\varepsilon(e^\varepsilon-1) + \varepsilon\cdot\sqrt{2k\log(1/\delta')}$. The first term $k\varepsilon(e^\varepsilon-1)$ is slightly worse than Theorem \ref{thm:advancedcomposition_approx}, which gives $\frac12 k \varepsilon^2$ instead.}
    The key concepts of privacy loss distributions and concentrated DP were implicit in this proof, but they were only made explicit in a separate paper by Dwork and Rothblum \cite{dwork2016concentrated}. Bun and Steinke \cite{bun2016concentrated} refined the notion of concentrated DP and presented the definition that we use here (Definition \ref{defn:cdp}).
    
    Kairouz, Oh, and Viswanath \cite{kairouz2015composition} proved an optimal composition theorem for approximate differential privacy.
    Specifically, the $k$-fold composition of $(\varepsilon,\delta)$-differential privacy satisfies $(\varepsilon',\delta')$-differential privacy if and only if
    \[
        \frac{1}{(1+e^\varepsilon)^k}\sum_{\ell=0}^k {k \choose \ell} \cdot e^{\ell\varepsilon} \cdot \max\left\{ 0 , 1 - e^{\varepsilon' - (2\ell-k)\varepsilon} \right\} \le 1 - \frac{1-\delta'}{(1-\delta)^k}.
    \]
    This expression is rather complex, but the proof is relatively intuitive. The key insight is that we can reduce the analysis to the $k$-fold composition of a specific worst-case $(\varepsilon,\delta)$-DP mechanism.\footnote{Specifically, Corollary \ref{cor:kov} shows how to construct such a worst-case mechanism.} With probability $\delta$, this mechanism has infinite privacy loss. With probability $(1-\delta) \cdot \frac{e^\varepsilon}{1+e^\varepsilon}$, it has privacy loss $\varepsilon$. And, with probability  $(1-\delta) \cdot \frac{1}{1+e^\varepsilon}$, it has privacy loss $-\varepsilon$. 
    The privacy loss of the $k$-fold composition is the convolution of $k$ of these privacy losses. Thus, with probability $1-(1-\delta)^k$ the privacy loss of the composition is infinite. Otherwise -- i.e., with probability $(1-\delta)^k$ -- the privacy loss has a shifted binomial distribution. Namely, for all $\ell \in [k]\cup\{0\}$, \[\pr{}{Z=\varepsilon\cdot\ell - \varepsilon \cdot (k-\ell)} = (1-\delta)^k \cdot {k \choose \ell} \cdot \left(\frac{e^\varepsilon}{e^\varepsilon+1}\right)^\ell \cdot \left(\frac{1}{e^\varepsilon+1}\right)^{k-\ell},\] where $Z$ is the privacy loss of the $k$-fold composition of the worst-case $(\varepsilon,\delta)$-DP mechanism.
    Applying Proposition \ref{prop:privloss_adp} to this distribution yields the expression for the optimal composition theorem.
    
    Kairouz, Oh, and Viswanath \cite{kairouz2015composition} also considered \emph{heterogeneous} optimal composition. That is, the composition of $k$ mechanisms where each mechanism $j \in [k]$ has a different $(\varepsilon_j,\delta_j)$-DP guarantee. However, the expression becomes more complicated. Intuitively, this is because the privacy loss distribution can be supported on $2^k$ points in the heterogeneous case, whereas, in the homogeneous case, it is supported on only $k+1$ points. Thus it takes exponential time to compute the privacy loss distribution. To be precise, Murtagh and Vadhan \cite{murtagh2016complexity} showed that exactly computing the optimal composition is \#P-complete, even if $\delta_j=0$ for each $j \in [k]$. However, Murtagh and Vadhan also showed that the optimal composition theorem could be approximated to arbitrary precision in polynomial time.
    
    Although these composition results \cite{kairouz2015composition,murtagh2016complexity} are optimal, they are limited in that they begin by assuming some $(\varepsilon_j,\delta_j)$-DP guarantees about the algorithms being composed. However, we usually know more about the algorithms being composed than simply these two parameters.
    For example, we may know that the algorithms being composed are Gaussian noise addition.
    Incorporating this additional information allows us to prove even better bounds than optimal composition. 
    This was the main impetus for the development of concentrated DP and R\'enyi DP, which we have discussed.
    
    A recent line of work \cite{meiserm18,pmlr-v108-koskela20b,pmlr-v130-koskela21a,koskela2021computing,NEURIPS2021_6097d8f3,dong2019gaussian,pmlr-v151-zhu22c,canonne2020discrete,ghazi2022faster,alghamdi2022saddle}
    has explored optimal composition guarantees whilst incorporating additional information about the mechanisms being composed.
    To make these computations efficient they consider the (discrete) Fourier transform of the privacy loss.\footnote{To apply a \emph{discrete} Fourier transform, we must first discretize the privacy loss distribution, e.g., by rounding it to a grid. The choice of discretization determines the tightness of the final guarantee, and the computational complexity of computing it.} That is, where concentrated DP and R\'enyi DP consider the moment generating function of the privacy loss $\ex{Z \gets \privloss{M(x)}{M(x')}}{\exp(tZ)}$, these works look at the characteristic function\\$\ex{Z \gets \privloss{M(x)}{M(x')}}{\exp(itZ)}$, where $i^2=-1$.
    These methods provide composition guarantees which are arbitrarily close to optimal, which are thus better than what is attainable via concentrated DP or R\'enyi DP.
    
    The optimality of advanced composition (Theorem \ref{thm:lowerbound}) is due to Bun, Ullman, and Vadhan \cite{BunUV14}. We present the analysis following Kamath and Ullman \cite{kamath2020primer}.
    
    The composition results we have presented all assume that the privacy parameters of the algorithms being composed (i.e., $(\varepsilon_j,\delta_j)$ for $j \in [k]$ in the language of Theorem \ref{thm:advancedcomposition_approx}) are fixed. It is natural to consider the setting where these parameters are chosen adaptively \cite{RogersRUV16} -- i.e., $(\varepsilon_j,\delta_j)$ could depend on the output of $M_{j-1}$. 
    For the most part, the composition results carry over seamlessly to the setting of adaptively-chosen privacy parameters. In particular, for Concentrated or R\'enyi DP, as long as the sum of the adaptively-chosen privacy parameters remains bounded, we attain privacy with that bound \cite{FeldmanZ21}.
    Another extension is ``concurrent composition'' \cite{vadhan2021concurrent}, which applies when an adversary may simultaneously access multiple interactive DP systems. Fortunately, the standard composition results readily extend to this setting \cite{vadhan2022concurrent,lyu2022composition}.

    \paragraph{Privacy Amplification by Subsampling.}
    The first explicit statement of differential privacy amplification by subsampling was in a blog post by Smith \cite{smith2009}, although it appeared implicitly earlier \cite{klnrs08} and the privacy effects of sampling on its own had also been studied \cite{10.1007/11818175_12}.
    
    For approximate DP, Balle, Barthe, and Gaborardi \cite{NEURIPS2018_3b5020bb} provide a thorough analysis of privacy amplification by subsampling (cf.~Theorem \ref{thm:subsampling_adp}). They present tight results for Poisson sampling (i.e., including each element independently), sampling a subset of a fixed size (without replacement), and also sampling with replacement, which means a person's data may appear \emph{multiple} times in the subsampled dataset.
    
    As discussed in Sections \ref{sec:subsamp_composition} and \ref{sec:use_subsampling}, subsampling arises in differentially private versions of stochastic gradient descent \cite{10.5555/1953048.2021036,bassily2014private}. Abadi, Chu, Goodfellow, McMahan, Mironov, Talwar, and Zhang \cite{abadi2016deep} applied this in the context of deep learning. To obtain better analyses, they developed the ``Moments Accountant'' -- i.e., R\'enyi DP (although the connection to R\'enyi divergences was only made later \cite{mironov2017renyi,bun2016concentrated}).
    
    Abadi et al.~\cite{abadi2016deep} obtained asymptotic R\'enyi DP bounds for the Poisson subsampled Gaussian mechanism, but they used numerical integration for their implementation. Mironov, Talwar, and Zhang \cite{mironov2019r} improved these asymptotic results and gave a better numerical method for exact computation (cf.~Theorem \ref{thm:rdp_subsampling}); our presentation in Section \ref{sec:sharp_rdp_subsampling} largely follows theirs.
    Bun, Dwork, Rothblum, and Steinke \cite{bdrs18} prove asymptotic R\'enyi DP bounds for Poisson subsampling applied to a concentrated DP mechanism (cf.~Theorem \ref{thm:rdp_subsampling_asymptotic}).
    Zhu and Wang \cite{pmlr-v97-zhu19c} gave generic R\'enyi DP bounds for Poisson subsampling.\footnote{Mironov, Talwar, and Zhang \cite{mironov2019r} and Zhu and Wang \cite{pmlr-v97-zhu19c} both provide analogs of Theorem \ref{thm:rdp_ss_flip}. However, to the best of our knowledge, Theorem \ref{thm:rdp_ss_flip} is novel.}
    
    Moving away from Poisson subsampling, Wang, Balle, and Kasiviswanathan \cite{pmlr-v89-wang19b} provide R\'enyi DP results for sampling a fixed-size set without replacement.
    
    Koskela, J\"alk\"o, and Honkela \cite{pmlr-v108-koskela20b} provide expressions for the privacy loss distribution of the subsampled Gaussian (under both Poisson subsampling and sampling a fixed size set with or without replacement) which can be numerically integrated to obtain optimal composition results.
    
    Closely related to privacy amplification by subsampling is privacy amplification by \emph{shuffling} \cite{bittau2017prochlo,ErlingssonFMRTT19,cheu2019distributed,BalleBGN19,FeldmanMT21,feldman2022stronger}.
    Privacy amplification by shuffling is usually presented in terms of local differential privacy \cite{klnrs08}. That is, there are $n$ individuals who independently generate random messages that satisfy local $\varepsilon$-DP. Those messages are then ``shuffled'' so that the potential adversary/attacker cannot identify which message originated from which individual. The additional randomness of the shuffling amplifies the privacy to $\left(O\left(\varepsilon \cdot \sqrt{\frac{\log(1/\delta)}{n}}\right), \delta\right)$-DP.
    
    Intuitively, shuffling is similar to subsampling with composition. Suppose we repeatedly sample one individual uniformly at random and perform an $\varepsilon$-DP computation on their data and the number of repetitions is equal to the number of individuals $n$. We can analyze this as subsampling a $1/n$ fraction (fixed size set) composed $n$ times. Privacy amplification by subsampling (Theorem \ref{thm:subsampling_adp}) says each repetition is $\varepsilon'$-DP for $\varepsilon'=\log(1+\frac1n(e^\varepsilon-1))=O(\varepsilon/n)$. Advanced composition (Theorem \ref{thm:advancedcomposition_pure}) over the $n$ repetitions yields $(\varepsilon'',\delta)$-DP for $\varepsilon'' = O(\sqrt{n\log(1/\delta)} \cdot \varepsilon') = O\left(\varepsilon \cdot \sqrt{\frac{\log(1/\delta)}{n}}\right)$.
    
    In contrast, for shuffling, we sample without replacement, so no individual is sampled more than once. This means the samples are not independent, so we cannot appeal to the subsampling plus composition analysis. Nevertheless, this intuition leads to the correct result.

\section*{Acknowledgements}
\addcontentsline{toc}{section}{Acknowledgements}

    We thank Ferdinando Fioretto for soliciting this chapter and we thank Cl\'ement Canonne, Sewoong Oh, Adam Sealfon, and Yu-Xiang Wang for comments on the draft.
    
\newpage

\addcontentsline{toc}{section}{References}
\printbibliography

\end{document}